\ifdefined \conf 
\documentclass{article}
\usepackage{icml2018}
\else 
\documentclass[11pt]{article}
\usepackage{algorithm2e,fullpage}
\fi

\usepackage{amssymb, amsthm, amsmath, dsfont,subcaption,  graphicx,hyperref,array}

\newcommand{\X}{\ensuremath{\mathcal{X}}}
\newcommand{\M}{\ensuremath{\mathcal{M}}}
\newcommand{\R}{\ensuremath{\mathbb{R}}}
\renewcommand{\S}{\ensuremath{\mathcal{S}}}
\newcommand{\T}{\ensuremath{^{\mathsf{T}}}}
\renewcommand{\vec}[1]{\ensuremath{\pmb{#1}}}
\newcommand{\E}{\ensuremath{\mathrm{E}}}

\newcommand{\cut}[1]{}
\newcommand{\diag}{\mathrm{diag}}

\newtheorem{theorem}{Theorem}
\newtheorem{claim}[theorem]{Claim}
\newtheorem{proposition}[theorem]{Proposition}
\newtheorem{lemma}[theorem]{Lemma}
\newtheorem{corollary}[theorem]{Corollary}

\ifdefined \conf
\newcommand{\mycite}[1]{\yrcite{#1}}
\renewcommand{\paragraph}[1]{\vskip 0.075cm\noindent\textbf{#1}}

\newcommand{\myinlineequation}[1]{$#1$}
\newcommand{\myinlineequationdot}[1]{$#1$.\ }
\newcommand{\myinlineequationcomma}[1]{$#1$,\ }
\newcommand{\toggle}[2]{#1}
\newcommand{\onlyconf}[1]{#1}
\newcommand{\onlyfull}[1]{}
\newcommand{\myAppendix}{supplementary material}

\else 
\newcommand{\mycite}[1]{\cite{#1} }

\newcommand{\myinlineequation}[1]{\[#1\]}
\newcommand{\myinlineequationdot}[1]{\[#1\]}
\newcommand{\myinlineequationcomma}[1]{\[#1\]}
\newcommand{\toggle}[2]{#2}
\newcommand{\onlyconf}[1]{}
\newcommand{\onlyfull}[1]{#1}
\newcommand{\myAppendix}{Appendix}
\fi

\title{Locally Private Hypothesis Testing}
\author{Or Sheffet\thanks{This work was supported by the Natural Sciences and Engineering Council of Canada,  Grant \#2017-06701. The author is also an unpaid collaborator on NSF grant 1565387.}\\Dept of Computing Science\\University of Alberta\\\texttt{osheffet@ualberta.ca}}
\date{\today}

\begin{document}
	\ifdefined \conf 
	\twocolumn[
	\icmltitlerunning{Locally Private Hypothesis Testing}
	\icmltitle{Locally Private Hypothesis Testing}
	\icmlsetsymbol{equal}{*}
	
	\begin{icmlauthorlist}
		\icmlauthor{Or Sheffet}{equal,to}
	\end{icmlauthorlist}
	\icmlaffiliation{to}{Department of Computing Science, University of Alberta, Edmonton, Canada}
	\icmlcorrespondingauthor{Or Sheffet}{osheffet@ualberta.ca}
	\icmlkeywords{Differential Privacy, Hypothesis Testing, Randomized Response}
	\vskip 0.3in
	]
	\bibliographystyle{icml2018}
	\else 
	\bibliographystyle{alpha}
	\maketitle
	\fi
	\begin{abstract}
		We initiate the study of differentially private hypothesis testing in the local-model, under both the standard (symmetric) randomized-response mechanism~\cite{Warner65, KasiviswanathanLNRS08} and the newer (non-symmetric) mechanisms~\cite{BassilyS15, BassilyNST17}. First, we study the general framework of mapping each user's type into a signal and show that the problem of finding the maximum-likelihood distribution over the signals is feasible. Then we discuss the randomized-response mechanism and show that, in essence, it maps the null- and alternative-hypotheses onto new sets, an affine translation of the original sets. We then give sample complexity bounds for identity and independence testing under randomized-response. We then move to the newer non-symmetric mechanisms and show that there too the problem of finding the maximum-likelihood distribution is feasible. Under the mechanism of Bassily et al~\mycite{BassilyNST17} we give identity and independence testers with better sample complexity than the testers in the symmetric case, and we also propose a $\chi^2$-based identity tester which we investigate empirically. 
	\end{abstract}
	
	\section{Introduction}
	\label{sec:intro}
	Differential privacy is a mathematically rigorous notion of privacy that has become the de-facto gold-standard of privacy preserving data analysis. 
	Informally, $\epsilon$-differential privacy bounds the affect of a single datapoint on any result of the computation by $\epsilon$. \toggle{In recent years}{By now we have a myriad of differentially private analogues of numerous data analysis tasks. Moreover, in recent years} the subject of private hypothesis testing has been receiving increasing attention (see Related Work below). However, by and large, the focus of private hypothesis testing is in the centralized model (or the curated model), where a single trusted entity holds the sensitive details of $n$ users and runs the private hypothesis tester on the actual data.
	
	In contrast, the subject of this work is private hypothesis testing in the \emph{local}-model (or the distributed model), where a $\epsilon$-differentially private mechanism is applied \emph{independently} to each datum\onlyfull{, resulting in one noisy signal per each datum. Moreover, the noisy signal is quite close to being uniformly distributed among all possible signals, so any observer that sees the signal has a very limited advantage of inferring the datum's true type}. This model, which alleviates trust (each user can run the mechanism independently on her own and release the noisy signal from the mechanism), has gained much popularity in recent years, especially since it was adopted by Google's Rappor~\cite{ErlingssonPK14} and Apple~\cite{Apple17}. And yet, despite its popularity, and the fact that recent works~\cite{BassilyS15, BassilyNST17} have shown the space of possible locally-private mechanism is richer than what was originally thought, little is known about private hypothesis testing in the local-model.
	
	\subsection{Background: Local Differential Privacy \ifdefined\conf\else as a Signaling Scheme\fi}
	\label{subsec:def_local_model_as_signaling}
	
	We view the local differentially private model as a signaling scheme. Each datum / user has  a type $x$ taken from a predefined and publicly known set of possible types $\X$ whose size is $T=|\X|$. The differentially private mechanism is merely a randomized function $\M:([n],\X)\to\S$, mapping each possible type $\X$ of the $i$-th datum to some set of possible signals $\S$, which we assume to be $\epsilon$-differentially private: for any index $i$, any pair of types $x,x'\in \X$ and any signal $s\in \S$ it holds that $\Pr[\M(i,x) = s] \leq e^\epsilon \Pr[\M(i,x')=s]$.\footnote{For simplicity, we assume $\S$, the set of possible signals, is discrete. Note that this doesn't exclude mechanisms such as adding Gaussian/Gamma noise to a point in $\R^d$ --- such mechanisms require $\X$ to be some bounded subset of $\R^d$ and use the bound to set the noise appropriately. Therefore, the standard approach of discretizing $\X$ and projecting the noisy point to the closest point in the grid yields a finite set of signals $\S$.} In our most general results (Theorems~\ref{thm:max_likelihood_symmetric} and~\ref{thm:max_likelihood_non_symmetric}), we ignore the fact that $\M$ is $\epsilon$-differentially private, and just refer to any signaling scheme that transforms one domain (namely, $\X$) into another ($\S$). For example, a surveyer might unify rarely occurring types under the category of ``other'', or perhaps users report their types over noisy channels, etc. 
	
	We differentiate between two types of signaling schemes\onlyfull{, both anchored in differentially private mechanisms}: the \emph{symmetric} (or index-oblivious) variety, and the \emph{non-symmetric} (index-aware) type. A local signaling mechanism is called \emph{symmetric} \onlyfull{or \emph{index-oblivious}} if it is independent of the index of the datum. Namely, if for any $i\neq j$ we have that $\M(i,x) = \M(j,x) \stackrel{\rm def}= \M(x)$. A classic example of such a mechanism is \emph{randomized-response}~--- that actually dates back to before differential privacy was defined~\cite{Warner65} and was first put to use in differential privacy in~\cite{KasiviswanathanLNRS08}~--- where each user / datum $x$ draws her own signal from the set $\S= \X$ skewing the probability ever-so-slightly in favor of the original type. I.e. if the user's type is $x$ then $\M(x) = \begin{cases}   x, &\textrm{w.p. } \frac{e^\epsilon}{T-1+e^\epsilon} \cr x', &\textrm{ for any other $x'$ w.p. }\frac 1 {T-1+e^\epsilon}\end{cases}$. \onlyfull{This mechanism applies to all users, regardless of position in the dataset.}
	
	The utility of the above-mentioned symmetric mechanism scales polynomially with $T$ (or rather, with $|\S|$), which motivated the question of designing \toggle{local}{locally differentially-private} mechanisms with error scaling logarithmically in $T$. This question was recently answered on the affirmative by the works of Bassily and Smith~\mycite{BassilyS15} and Bassily et al~\mycite{BassilyNST17}, whose mechanisms are \emph{not} symmetric. In fact, both of them work by presenting each user $i$ with a mapping $f_i: \X \to\S$ (the mapping itself is chosen randomly, but it is public, so we treat it as a given), and the user then runs the standard randomized response mechanism \emph{on the signals} using $f_i(x)$ as the more-likely signal. (In fact, in both schemes, $\S = \{1,-1\}$: in~\cite{BassilyS15} $f_i$ is merely the $j$-th coordinate of a hashing of the types where $j$ and the hashing function are publicly known, and in~\cite{BassilyNST17} $f_i$ maps a u.a.r chosen subset of $\X$ to $1$ and its complementary to $-1$.\footnote{In both works, much effort is put to first reducing $T$ to the most frequent $\sqrt{n}$ types, and then run the counting algorithm. Regardless, the end-counts / collection of users' signals are the ones we care for the sake of hypothesis testing.}) \toggle{ So}{It is simple to identify each $f_i$ as a $0/1$-matrix of size $|\S|\times |\X|$; and --- even though current works use only a deterministic mapping $f_i$ --- we even allow for a randomized mapping, so $f_i$ can be thought of a $|\S|\times |\X|$ of entries in $[0,1]$ (such that for each $x\in \X$ we have $\sum_{s\in \S} \M_i(s,x) =1$). Regardless,} given $f_i$, the user then tosses her our private random coins to determine what signal she broadcasts. Therefore, each user's mechanism can be summarized in a $|\S|\times |\X|$-matrix, where $\M_i(s,x)$ is the probability a user of type $x$ sends the signal $s$. For example, using the mechanism of~\cite{BassilyNST17}, each user whose type maps to $1$ sends ``signal $1$'' with probability $\tfrac{e^\epsilon}{1+e^\epsilon}$ and ``signal $-1$'' with probability $\tfrac 1 {1+e^\epsilon}$. Namely, $\M_i(f_i(x), x) = \tfrac{e^\epsilon}{1+e^\epsilon}$ and $\M_i(f_i(x), x) = \tfrac{1}{1+e^\epsilon}$, where $f_i$ is the mapping $\X\to\{1,-1\}$ set for user $i$.
	
	\subsection{Our Contribution and Organization}
	\label{subsec:intro_contribution}
	
	This work initiates (to the best of our knowledge) the \emph{theory} of differentially private hypothesis testing in the local model. First we survey related work and preliminaries. Then, in Section~\ref{sec:symmetric_RR}, we examine the symmetric case and show that any mechanism (not necessarily a differentially private one) yields a distribution on the signals for which finding a maximum-likelihood hypothesis is feasible, assuming the set of possible hypotheses is convex. Then, focusing on the classic randomized-response mechanism, we show that the problem of maximizing the likelihood of the observed signals is strongly-convex and thus simpler than the original problem. More importantly, in essence we give a characterization of hypothesis testing under randomized response: the symmetric locally-private mechanism translates the original null hypothesis $H_0$ (and the alternative $H_1$) by a known affine translation into a different set $\varphi(H_0)$ (and resp. $\varphi(H_1)$). Hence, hypothesis testing under randomized-response boils to discerning between two different (and considerably closer in total-variation distance) sets, but in \emph{the exact same model} as in standard hypothesis testing as all signals were drawn from the same hypothesis in $\varphi(H_0)$. As an immediate corollary we give bounds on identity-testing~(Corollary~\ref{cor:identity_testing_RR})  and independence-testing~(Theorem~\ref{thm:independence_testing_RR}) under randomized-response. (The latter requires some manipulations and far less straight-forward than the former.) The sample complexity (under certain simplifying assumptions) of both problems is proportional to $|\X|^{2.5}$.
	
	In Section~\ref{sec:non-symmetric-RR} we move to the non-symmetric local-model. Again, we start with a general result showing that in this case too, finding an hypothesis that maximizes the likelihood of the observed signals is feasible when the hypothesis-set is convex. We then focus on the mechanism of Bassily et al~\mycite{BassilyNST17} and show that it also makes the problem of finding a maximum-likelihood hypothesis strongly-convex. We then give a simple identity tester under this scheme whose sample complexity is proportional to $|\X|^{2}$, and is thus more efficient than \emph{any} tester under standard randomized-response. Similarly, we also give an independence-tester with a similar sample complexity. In Section~\ref{subsec:experiment} we empirically investigate alternative identity-testing and independence-testing based on Pearson's $\chi^2$-test in this non-symmetric scheme, and identify a couple of open problems in this regime.
	
	\subsection{Related Work}
	\label{subsec:related_work}
	
	Several works have looked at the intersection of differential privacy and statistics~\cite{DworkL09, Smith11, ChaudhuriH12, DuchiJW13, DworkSZ15} mostly focusing on robust statistics; but only a handful of works study rigorously the significance and power of hypotheses testing under differential privacy\onlyfull{~\cite{VuS09, UhlerSF13, WangLK15, RogersVLG16, CaiDK17, Sheffet17, KarwaV18}}. Vu and Slavkovic~\mycite{VuS09} looked at the sample size for privately testing the bias of a coin. Johnson and Shmatikov~\mycite{JohnsonS13}, Uhler et al~\mycite{UhlerSF13} and Yu et al~\mycite{YuFSU14} focused on the Pearson $\chi^2$-test (the simplest goodness of fit test), showing that the noise added by differential privacy vanishes asymptotically as the number of datapoints goes to infinity, and propose a private $\chi^2$-based test which they study empirically. Wang et al~\mycite{WangLK15} and Gaboardi et al~\mycite{RogersVLG16}  who have noticed the issues with both of these approaches, have revised the statistical tests themselves to incorporate also the added noise in the private computation. 
	Cai et al~\mycite{CaiDK17} give a private identity tester based on noisy $\chi^2$-test over large bins, Sheffet~\mycite{Sheffet17} studies private Ordinary Least Squares using the JL transform, and Karwa and Vadhan~\mycite{KarwaV18} give matching upper- and lower-bounds on the confidence intervals for the mean of a population. All of these works however deal with the centralized-model of differential privacy.
	
	Perhaps the closest to our work are the works of Duchi et al~\mycite{DuchiJW13, DuchiWJ13} who give matching upper- and lower-bound on robust estimators in the local model. And while their lower bounds do inform as to the sample complexity's dependency on $\epsilon^{-2}$, they do not ascertain the sample complexity dependency on the size of the domain ($T=|\X|$) we get in Section~\ref{sec:symmetric_RR}. Moreover, these works disregard independence testing (and in fact~\cite{DuchiWJ13} focus on mean estimation so they apply randomized-response to each feature independently generating a product-distribution even when the input isn't sampled from a product-distribution). And so, to the best of our knowledge, no work has focused on hypothesis testing in the local model, let alone in the (relatively new) non-symmetric local model.
	
	\section{Preliminaries, Notation and Background}
	\label{sec:preliminaries}
	\paragraph{Notation.} We user $lower$-case letters to denote scalars, $\pmb{bold}$ characters to denote vectors and $CAPITAL$ letters to denote matrices. So $1$ denotes the number, $\vec 1$ denotes the all-$1$ vector, and  $1_{\X \times \X}$ denotes the all-$1$ matrix over a domain $\X$. We use $\vec e_x$ to denote the standard basis vector with a single $1$ in coordinate corresponding to $x$. To denote the $x$-coordinate of a vector $\vec v$ we use $v(x)$, and to denote the $(x,x')$-coordinate of a matrix $M$ we use $M(x,x')$. For a given vector $\vec v$, we use $\diag(\vec v)$ to denote the matrix whose diagonal entries are the coordinates of $\vec v$. For any natural $n$, we use $[n]$ to denote the set $\{1,2,...,n\}$.
	
	\paragraph{Distances and norms.} Unless specified otherwise $\|\vec v\|$ refers to the $L_2$-norm of $\vec v$, whereas $\|\vec v\|_1$ refers to the $L_1$-norm. We also denote $\|\vec v\|_{\frac 2 3} = \left(\sum_i |v_i|^{\frac 2 3}\right)^{\frac 3 2}$. For a matrix, $\|M\|_1$ denotes (as usual) the maximum absolute column sum. We identify a distribution $\vec p$ over a domain $\X$ as a $|\X|$-dimensional vector with non-negative entries that sum to $1$. This defines the \emph{total variation} distance between two distributions: $d_{\rm TV}(\vec p,\vec q) = \tfrac 1 2 \| \vec p -\vec q\|_1$. (On occasion, we will apply $d_{\rm TV}$ to vectors that aren't distributions, but rather nearby estimations; in those cases we use the same definition: the half of the $L_1$-norm.) It is known that the TV-distance is a metric overs distributions. We also use the $\chi^2$-divergence to measure difference between two distributions: $d_{\chi^2}(\vec p,\vec q) = \sum_x\tfrac{(p(x)-q(x))^2}{p(x)} = \left(\sum_x \tfrac { (q(x))^2 } {p(x)} \right) - 1$. The $\chi^2$-divergence is not symmetric and can be infinite, however it is non-negative and zeros only when $\vec p=\vec q$. We refer the reader to~\cite{SasonV16} for more properties of the total-variance distance the $\chi^2$-divergence.
	\paragraph{Differential Privacy.} An algorithm $\mathcal{A}$ is called $\epsilon$-differentially private, if for any two datasets $D$ and $D'$ that differ only on the details of a single user and any set of outputs $S$, we have that $\Pr[\mathcal{A}(D)\in S]\leq e^\epsilon \Pr[\mathcal{A}(D')\in S]$. The unacquainted reader is referred to the Dwork-Roth monograph~\cite{DworkR14} as an introduction to the rapidly-growing field of differential privacy.
	\paragraph{Hypothesis testing.} \toggle{The }{Hypothesis testing is an extremely wide field of study, see~\cite{HoggMC05} as just one of many resources about it. In general however, the} problem of hypothesis testing is to test whether a given set of samples was drawn from a distribution satisfying the null-hypothesis or the alternative-hypothesis. Thus, the null-hypothesis is merely a set of possible distributions $H_0$ and the alternative is disjoint set $H_1$. Hypothesis tests boils down to estimating a test-statistics $\theta$ whose distribution has been estimated under the null-hypothesis\onlyfull{ (or the alternative-hypothesis)}. We can thus {\tt reject} the null-hypothesis is the value of $\theta$ is highly unlikely, or {\tt accept} the null-hypothesis otherwise. We call an algorithm a \emph{tester} if the acceptance (in the completeness case) or rejection (in the soundness case) happen with probability $\geq 2/3$. Standard amplification techniques (return the median ofindependent tests) reduce the error probability from $1/3$ to any $\beta>0$ at the expense of increasing the sample complexity by a factor of $O(\log(1/\beta))$; hence we focus on achieving a constant error probability. One of the most prevalent and basic tests is the \emph{identity}-testing, where the null-hypothesis is composed of a single distribution $H_0 = \{\vec p\}$ and our goal is to accept if the samples are drawn from $\vec p$ and reject if they were drawn from any other $\alpha$-far (in $d_{\rm TV}$) distribution. Another extremely common tester is for \emph{independence} when $\X$ is composed of several features (i.e., $\X = \X^1 \times \X^2 \times... \times \X^d$) and the null-hypothesis is composed of all product distributions $H_0 = \{ \vec p^1 \times ...\times \vec p^d\}$ where each $\vec p^j$    is a distribution on the $j$th feature $\X^j$.

	\paragraph{Miscellaneous.} \onlyfull{The Chebyshev inequality states that for any random variable $X$, we have that $\Pr[ |X -\E[X]| > t ]\leq \frac{{\rm Var}(X)}{t^2}$. We also use the Heoffding inequality, stating that for $n$ iid random variables $X_1,...,X_n$ in the range $[a,b]$ we have that $\Pr[ \tfrac 1 n \sum_i X_i - \E[X_i] < -\alpha  ] \leq \exp(-2n\alpha^2/(b-a)^2)$ and similarly that $\Pr[ \tfrac 1 n \sum_i X_i - \E[X_i] > \alpha  ] \leq \exp(-2n\alpha^2/(b-a)^2)$. It is a particular case of the MacDiarmid inequality, stating that for every function $f$ such that if we have bounds $\forall i\forall x_1, x_2, ...,x_n, x'_i$ we have $| f(x_1,..,x_i,...x_n) - f(x_1,...,x_i',...x_n)   | <c_i$ then $\Pr[ f(X_1,...,X_n) - \E[f] < -\alpha  ] \leq \exp( -2\alpha^2 / \left(\sum_i c_i^2\right)   )$.\\
	A matrix $M$ is called \emph{positive semidefinite} (PSD) if for any unit-length vector $u$ we have $\vec u\T M \vec u\geq 0$. }We use $M\succeq 0$ to denote that $M$ is a positive semi-definite (PSD) matrix, and $M\succeq N$ to denote that $(M-N)\succeq0$. We use $M^\dagger$ to denote $M$'s pseudo-inverse.\onlyfull{When the rows of $M$ are independent, we have that $M^\dagger = M\T(MM\T)^{-1}$. }
	We emphasize that we made no effort to minimize constants in our proofs, and only strived to obtain asymptotic bounds ($O(\cdot), \Omega(\cdot)$).\onlyfull{ We use $\tilde O(\cdot), \tilde{\Omega(\cdot)}$ to hide poly-log factors.}

	\section{Symmetric Signaling Scheme}
	\label{sec:symmetric_RR}
	
	Recall, in the symmetric signaling scheme, each user's type is mapped through a random function $\M$ into a set of signals $\S$. This mapping is index-oblivious~--- each user of type $x\in \X$, sends the signal $s$ with the same probability $\Pr[\M(x)=s]$. We denote the matrix $G$ as the $(|\S|\times |\X|)$-matrix whose entries are $\Pr[\M(x)=s]$, and its $s$th-row by $\vec g_s$. Note that all entries of $G$ are non negative and that for each $x$ we have $\|G\vec e_x\|_1=1$. By garbling each datum \emph{i.i.d}, we observe the new dataset $(y_1,y_2,...,y_n) \in \S^n$. 
	
	\begin{theorem}
		\label{thm:max_likelihood_symmetric}
		For any convex set $H$ of hypotheses, the problem of finding the max-likelihood $\vec p\in H$ generating the observed signals $(y_1,..,y_n)$ is poly-time solvable. 
	\end{theorem}
	\onlyconf{\vspace{-0.545cm}}
	\begin{proof}
		Since $G(s,x)$ describes the probability that a user of type $x$ sends the signal $s$, any distribution $\vec p\in H$ over the types in $\X$ yields a distribution on $\S$ where
		\myinlineequationdot{\Pr[\textrm{user sends } s] = \sum\limits_{x\in \X} \Pr[\textrm{user sends } s |~ \textrm{user of type }x]~\cdot~\Pr[\textrm{user of type }x]\ = \sum_{x\in \X}G(s,x)\cdot p(x) = \vec g_s\T \vec p}
		Therefore, given the signal $(y_1,...,y_n)$, we can summarize it by a histogram over the different signals $\langle n_s \rangle_{s\in\S}$, and thus the likelihood of seeing this particular signal is given by:
		\myinlineequationdot{L(\vec p; y_1,...,y_n) = \prod_{i} \vec g_{y_i}\T\vec p = \prod_{s\in \S} (\vec g_s\T  p)^{n_s} = \exp\left( \sum_{s\in \S}  n_s \log(\vec g_s\T  \vec p)  \right)
		}
		\toggle{We clearly have that $\arg\max_{\vec p\in H}\left( L(\vec p; y_1,...,y_n)\right) 
		= \arg\min_{\vec p \in H} -\sum_{s\in \S} -\tfrac{ n_s}n \log(\vec g_s\T  \vec p)
		$.}{As ever, \begin{equation} 
		\label{eq:symmetric_case_optimization_problem}
		\arg\max_{\vec p\in H}\left( L(\vec p; y_1,...,y_n)\right) = \arg\min_{\vec p\in H} \left(-\tfrac 1 n \log(L(\vec p; y_1,...,y_n)) \right) = \arg\min_{\vec p \in H} -\tfrac 1 n\sum_{s\in \S}  n_s \log(\vec g_s\T  \vec p)	\end{equation}}
		Denoting the log-loss function as $f(\vec p) = -\sum_{s\in \S} \tfrac{ n_s}n \log(\vec g_s\T  \vec p)$, we get that its gradient is
		\myinlineequationcomma{\nabla f = -\tfrac 1 n\sum_{s\in \S} \frac {n_s} {\vec g_s\T \vec p} \cdot \vec g_s}
		and its Hessian is given by the $(|\X|\times |\X|)$-matrix 
		\myinlineequationdot{\tfrac 1 n\sum_{s\in S} \frac{n_s}{(\vec g_s\T \vec p)^2} \vec g_s \vec g_s\T}
		As $\sum_s \vec g_s \vec g_s\T$ is a PSD matrix, and each of its rank-$1$ summands is scaled by a positive number, it follows that the Hessian is a PSD matrix and that our loss-function is convex. Finding the minimizer of a convex function over a convex set is poly-time solvable (say, by gradient descent~\cite{Zinkevich03}), so we are done.		
	\end{proof}
	\onlyconf{\vspace{-0.4cm}}
	Unfortunately, in general the solution to this problem has no closed form (to the best of our knowledge). However, we can find a close-form solution under the assumption that $G$ isn't just any linear transformation but rather one that induces probability distribution over $\S$, the assumption that $|\S| \leq |\X|$ (in all applications we are aware of use fewer signals than user-types) and one extra-condition.
	
	\DeclareRobustCommand{\corCloseFormSolution}{Let $\vec q^*$ be the $|\S|$-dimensional vector given by $\langle \tfrac {n_s}{n}\rangle$. Given that $|\S|\leq |\X|$, that $G$ is a full-rank matrix satisfying $\|G\|_1 = 1$ and assuming that $\big( G^\dagger \vec q^* + \ker(G)\big) \cap H \neq \emptyset$, then any vector in $H$ of the form $\vec p^* + \vec u$ where $\vec p^* = G^\dagger \vec q^*$ and $\vec u \in \ker(G)$ is an hypothesis that maximizes the likelihood of the given signals $(y_1,...,y_n)$.}
	\begin{corollary}
		\label{cor:close_form_solution}
		\corCloseFormSolution
	\end{corollary}
	\DeclareRobustCommand{\proofCorCloseFormSolution}{
	\begin{proof}
		Our goal is to find some $\vec p \in H$ which minimizes $f(\vec p)$. Denoting $\vec q$ as the $|\S|$-dimensional vector such that $q(s) = \vec g_s \T \vec p$, we note that $G$ isn't just any linear transformation, but rather one that induces probability over the signals, and so $\vec q$ is a non-negative vector that sums to $1$. We therefore convert the problem of minimizing our loss function into the following optimization  problem
		\begin{align*}
		\min \phi(\vec p, \vec q) &= -\sum_{s\in S} n_s \log(q(s))\\
		\textrm{subject to} & \sum_s q(s) = 1\\
		& \forall s, q(s)\geq 0\\
		& \vec q = G \vec p\\
		& \vec p \in H
		\end{align*}
		Using Lagrange multipliers, it is easy to see that $\frac{\partial \phi}{\partial \vec q} = \langle \frac {-n_s}{q(s)} \rangle _{s\in \S}$ and that $\frac{\partial}{\partial \vec q}\left( \sum_{s\in \S} q(s)-1\right) = \vec 1 = \frac{\partial}{\partial \vec q}\left( \vec q-G\vec p=0\right)$ and so the minimizer is obtained when $\vec q$ equates all ratios $\frac {n_s}{q(s)} = \frac{n_{s'}}{q(s')}$ for all $s,s'$, namely when $\vec q =\vec q^*$. Since we assume $G^\dagger \vec q^* + \ker(G)$ has a non-empty intersection with $H$, then let $\vec p$ be any hypothesis in $H$ of the form $\vec p^* + \vec u$ where $\vec u\in \ker(G)$. We get that $(\vec p,\vec q)$ is the minimizer of $\phi$ satisfying all constraints. By assumption, $\vec p\in H$. Due to the fact that $G$ is full-rank and that $|\S| \leq |\X|$ we have that $G(\vec p^*+u)  = G\cdot  G^\dagger \vec q^* + \vec 0 = I \cdot \vec q^* = \vec q^*$, and by definition, $\vec q^*$ is a valid distribution vector (non-negative that sums to $1$).
	\end{proof}}
	\toggle{\vspace{-0.3cm} Proof deffered to the \myAppendix, Section~\ref{apx_sec:proofs_RR}.}{\proofCorCloseFormSolution}
	\onlyfull{If all conditions of Corollary~\ref{cor:close_form_solution} hold, we get a simple procedure for finding a minimizer for our loss-function: (1) Compute the pseudo-inverse $G^\dagger$ and find $\vec p^* = G^\dagger \vec q^*$; (2) find a vector $\vec u \in \ker(G)$ such that $\vec p^*+\vec u \in H$. (The latter steps requires the exact description of $H$, and might be difficult if $H$ is not convex. However, if $H$ is convex, then $H-\vec p^*$ is a shift of a convex body and therefore convex, so finding the point $\vec x \in H-\vec p^*$ which minimizes the distance to a given linear subspace is a feasible problem.)}
	
	\subsection{Hypothesis Testing under Randomized-Response} 
	\label{subsec:RR_hyp_testing}
	We now aim to check the affect of a particular $G$, the one given by the randomized-response mechanism. In this case $\S = \X$ and we denote $G$ as the matrix whose entries are $G(x,x') = \begin{cases} \rho+\gamma &\textrm{, if }x'=x \\ \rho &\textrm{, otherwise}\end{cases}$ where $\rho \stackrel{\rm def}=\tfrac 1 {T-1+e^\epsilon}$ and $\gamma \stackrel{\rm def}{=} \tfrac {e^\epsilon-1}{T-1+e^\epsilon}$. We get that $G = \rho \cdot 1_{\X\times \X} + \gamma I$ (where $1_{\X\times X}$ is the all-$1$ matrix). In particular, all vectors $\vec g_s = \vec g_x$, which correspond to the rows of $G$, are of the form: $\vec g_x = \rho \vec 1 + \gamma \vec e_x$. It follows that for any probability distribution $\vec p\in H$ we have that $\Pr[\textrm{seeing signal }x] = \vec g_x\T \vec p = \rho + \gamma p(x)$. We have therefore translated any $\vec p\in H$ (over $\X$) to an hypothesis $\vec q$ over $\S$ (which in this case $\S = \X$), using the affine transformation $\varphi(\vec p) = \rho \vec 1 + \gamma \vec p = T\rho \vec u_{\X} + \gamma \vec p$ when $\vec u_\X$ denotes the uniform distribution over $\X$. (Indeed, $\gamma = 1-T\rho$, an identity we will often apply.) Furthermore, at the risk of overburdening notation, we use $\varphi$ to denote the same transformation over scalars, vectors and even sets (applying $\varphi$ to each vector in the set). 
	\onlyfull{ 
	
	}As $\varphi$ is injective, we have therefore discovered the following theorem. 
	
	\begin{theorem}
		\label{thm:hyp_testing_RR}
		Under the classic randomized response mechanism, testing for any hypothesis $H_0$ (or for comparing $H_0$ against the alternative $H_1$) of the original distribution, translates into testing for hypothesis $\varphi(H_0)$ (or $\varphi(H_0)$ against $\varphi(H_1)$) for generating the signals $y_1,...,y_n$. 
	\end{theorem}
	\onlyconf{\vspace{-0.3cm}}
	Theorem~\ref{thm:hyp_testing_RR} seems very natural and simple, and yet (to the best of our knowledge) it was never put to words.
	
	Moreover, it is simple to see that under standard-randomized response, our log-loss function is in fact strongly-convex, and therefore finding $\vec p^*$ becomes drastically more efficient (see, for example~\cite{HazanKKA06}).
	\DeclareRobustCommand{\clmLikelihoodStronglyConvexityRR}{Given signals $y_1,...,y_n$ generated using standard randomized response with parameter $\epsilon<1$, we have that our log-loss function \onlyfull{from Equation~\eqref{eq:symmetric_case_optimization_problem}} is $\Theta(\epsilon^2\cdot  \frac{\min_x \{n_x\}}n)$-strongly convex. }
	\begin{claim}
		\label{clm:likelihood_strongly_convexity_RR}
		\clmLikelihoodStronglyConvexityRR
	\end{claim}
	\onlyconf{\vspace{-0.3cm}}
	Note that in expectation $n_x \geq \rho n$, hence with overwhelming probability we have $\min_x n_x \geq n/(2T)$\onlyfull{ so our log-loss function is $\Theta(\tfrac{\epsilon^2}{T})$-strongly convex}.
	\DeclareRobustCommand{\proofClmLikelihoodStronglyConvexityRR}{
	\begin{proof}
		Recall that for any $x\in \X$ we have $\vec g_x\T \vec p = \rho + \gamma p(x)$. Hence, our log-loss function $f(\vec p) = -\tfrac 1 n\sum_{x\in \X} n_x \log(\rho + \gamma p(x))$, whose gradient is the vector whose $x$-coordinate is $\frac{ \partial f}{\partial p(x)} = \frac{-\gamma n_x}{\rho + \gamma p(x)}$. The Hessian of $f$ is therefore the diagonal matrix whose diagonal entries are $\frac {\gamma^2 n_x}{(\rho+\gamma p(x))^2}$. Recall the definitions of $\gamma$ and $\rho$: it is easy to see that $\gamma \geq \epsilon \rho$, and since $\epsilon<1$ we also have that $e^\epsilon-1\leq 2\epsilon$, hence $\gamma \leq 2\epsilon \cdot \rho$. And so:
		\toggle{\begin{align*}
			\nabla^2 f &\succeq \frac{\min_x \{n_x\}}n \cdot \frac{\gamma^2}{(\rho + 2\epsilon\rho\cdot 1)^2} I 
			\cr& \succeq \min_x \{n_x\} \cdot \frac{\epsilon^2\rho^2}{\rho^2(1+2\epsilon)^2} I \succeq \min_x \{n_x\} \cdot \frac{\epsilon^2}{(1+2\epsilon)^2} I
			\end{align*}
			}{
		\[ \nabla^2 f \succeq \frac{\min_x \{n_x\}}n \cdot \frac{\gamma^2}{(\rho + 2\epsilon\rho\cdot 1)^2} I \succeq \min_x \{n_x\} \cdot \frac{\epsilon^2\rho^2}{\rho^2(1+2\epsilon)^2} I \succeq \min_x \{n_x\} \cdot \frac{\epsilon^2}{(1+2\epsilon)^2} I \]} making $f$ at least ($\frac {\epsilon^2}{9}\cdot  \frac{\min_x \{n_x\}}n$)-strongly convex.
	\end{proof}}
	\toggle{The proof is fairly straight-forawrd and is deferred to the \myAppendix, Section~\ref{apx_sec:proofs_RR}.}{\proofClmLikelihoodStronglyConvexityRR}
	
	A variety of corollaries follow from Theorem~\ref{thm:hyp_testing_RR}. In particular, a variety of detailing matching sample complexity upper- and lower-bounds translate automatically into the realm of making such hypothesis-tests over the outcomes of the randomized-response mechanism. We focus here on two of the most prevalent tests: identity testing and independence testing.
	
	\paragraph{Identity Testing.} Perhaps the simplest of the all hypothesis testing is to test whether a given sample was generated according to a given distribution or not. Namely, the null hypothesis is a single hypothesis $H_0 = \{\vec p\}$, and the alternative is $H_1 = \{\vec q:~ d_{\rm TV}(\vec p,\vec q) \geq \alpha\}$ for a given parameter $\alpha$. The seminal work of Valiant and Valiant~\mycite{ValiantV14} discerns that (roughly) $\Theta( \|\vec p\|_{\frac 2 3} / \alpha^2 )$ samples are sufficient and are necessary for correctly rejecting or accepting the null-hypothesis w.p.$\geq 2/3$.\footnote{For the sake of brevity, we ignore pathological examples where by removing $\alpha$ probability mass from $\vec p$ we obtain a vector of significantly smaller $\tfrac 2 3$-norm.}
	
	Here, the problem of identity testing under standard randomized response reduces to the problem of hypothesis testing between $\varphi(H_0) = \{\rho\vec 1 + \gamma \vec p:~\vec p \in H_0\}$ and $\varphi(H_1) = \{ \varphi(\vec q) :~\vec q ~{\rm satisfying}~d_{\rm TV}(\vec p,\vec q)\geq \alpha \}$. 
	
	\DeclareRobustCommand{\corIdentityTestingRR}{In order to do identity testing under standard randomized response with confidence and power $\geq 2/3$, it is necessary and sufficient that we get $\Theta(\frac {T^{2.5}}{\epsilon^2\alpha^2})$ samples.}
	\begin{corollary}
		\label{cor:identity_testing_RR}
		\corIdentityTestingRR
	\end{corollary}
	\DeclareRobustCommand{\proofCorIdentityTestingRR}{
	\begin{proof}
	For any $\vec q\in H_1$ it follows that $d_{\rm TV}(\varphi(\vec p), \varphi(\vec q)) = \tfrac 1 2 \| (\rho\vec 1 + \gamma\vec p) - (\rho\vec 1 + \gamma\vec q)  \|_1 = \tfrac \gamma 2 \|\vec p - \vec q\|_1 = \gamma \cdot d_{\rm TV}(\vec p, \vec q) \geq \gamma\alpha$. Recall that $\rho = \tfrac 1 {T-1+e^\epsilon}$ and $\gamma = \tfrac {e^\epsilon-1}{T-1+e^\epsilon}$, and so, for $\epsilon < 1$ we have $\tfrac 1 {T+2\epsilon} \leq \rho \leq \tfrac 1 T$ and $\tfrac \epsilon {T+2} \leq \gamma \leq \tfrac {2\epsilon} {T}$, namely $\rho = \Theta(1/T)$ and $\gamma = \Theta(\epsilon/T)$.
	Next, we bound $\|\rho\vec 1 + \gamma \vec p\|_{\frac 23}$:
	\toggle{\begin{align*}
		\left(\|\rho\vec 1 + \gamma \vec p \|_{\frac 2 3}\right)^{2/3} &= \sum_{x\in X} \left( \rho + \gamma p(x)\right)^{\frac 2 3} 
		\cr & \geq \sum_{x\in X} \max\{ \rho^{2/3}, \gamma^{2/3} p(x)^{2/3}  \} 
		\cr & \geq \max\left\{ T\rho^{2/3}, ~ \gamma^{2/3} \|\vec p\|_{\frac 2 3}^{2/3}   \right\}\\
		\intertext{ Using the fact that $(a+b)^{2/3} \leq a^{2/3} + b^{2/3}$ (See Proposition~\ref{pro:UB_2/3_norm} in Section~\ref{apx_sec:claims}) we also get }
		\sum_{x\in X} \left( \rho + \gamma p(x)\right)^{\frac 2 3} &\leq \sum_{x\in\X} \rho^{2/3} + \gamma^{2/3} p(x)^{2/3} 
		\cr &= T\rho^{2/3} + \gamma^{2/3}\|\vec p\|_{\frac 2 3}^{2/3}
		\end{align*}}{\begin{align*}
	\left(\|\rho\vec 1 + \gamma \vec p \|_{\frac 2 3}\right)^{2/3} &= \sum_{x\in X} \left( \rho + \gamma p(x)\right)^{\frac 2 3} \geq \sum_{x\in X} \max\{ \rho^{2/3}, \gamma^{2/3} p(x)^{2/3}  \} \geq \max\left\{ T\rho^{2/3}, ~ \gamma^{2/3} \|\vec p\|_{\frac 2 3}^{2/3}   \right\}\\
	\intertext{ Using the fact that $(a+b)^{2/3} \leq a^{2/3} + b^{2/3}$ (See Proposition~\ref{pro:UB_2/3_norm} in Section~\ref{apx_sec:claims}) we also get }
	&\sum_{x\in X} \left( \rho + \gamma p(x)\right)^{\frac 2 3} \leq \sum_{x\in\X} \rho^{2/3} + \gamma^{2/3} p(x)^{2/3} = T\rho^{2/3} + \gamma^{2/3}\|\vec p\|_{\frac 2 3}^{2/3}
	\end{align*}}
	It follows that the necessary and sufficient number of samples required for identity-testing under standard randomized response is proportional to
	\toggle{
		\begin{align*}
		&\Theta\left( \frac{\|\rho\vec 1 + \gamma\vec p\|_{\frac 2 3}}{\gamma^2\alpha^2}  \right) = \Theta(\left( \frac{\|\rho\vec 1 + \gamma\vec p\|_{\frac 2 3}^{2/3}}{\gamma^{4/3}\alpha^{4/3}}  \right)^{3/2}) 
		\cr &~= \Theta(  \left( \frac {T^{1/3} + \tfrac{\epsilon^{2/3}}{T^{2/3}}\|\vec p\|_{\frac 2 3}^{2/3}} {\tfrac{\epsilon^{4/3}}{T^{4/3}}\alpha^{4/3} }\right)^{\frac 3 2}   ) 
		\stackrel{(\ast)}= \Theta\left( \frac {T^{2.5}}{\epsilon^2\alpha^2} + \frac{ T\|\vec p\|_{\frac 2 3}   } {\epsilon\alpha^2}  \right)
		\end{align*}}{
	\begin{align*}
	\Theta\left( \frac{\|\rho\vec 1 + \gamma\vec p\|_{\frac 2 3}}{\gamma^2\alpha^2}  \right) &= \Theta(\left( \frac{\|\rho\vec 1 + \gamma\vec p\|_{\frac 2 3}^{2/3}}{\gamma^{4/3}\alpha^{4/3}}  \right)^{3/2}) = \Theta(  \left( \frac {T^{1/3} + \tfrac{\epsilon^{2/3}}{T^{2/3}}\|\vec p\|_{\frac 2 3}^{2/3}} {\tfrac{\epsilon^{4/3}}{T^{4/3}}\alpha^{4/3} }\right)^{3/2}   ) 
	\cr &= \Theta( \left(      \frac{T^{5/3}}{\epsilon^{4/3}\alpha^{4/3}} + \frac{ T^{2/3}\|\vec p\|_{\frac 2 3}^{2/3} } {\epsilon^{2/3}\alpha^{4/3}}\right)^{3/2} ) \stackrel{(\ast)}= \Theta\left( \frac {T^{2.5}}{\epsilon^2\alpha^2} + \frac{ T\|\vec p\|_{\frac 2 3}   } {\epsilon\alpha^2}  \right)
	\end{align*}}
	where the derivation marked by $(\ast)$ follows Proposition~\ref{pro:bound_3/2_norm} in Section~\ref{apx_sec:claims}.
	\onlyfull{
		
	}
	For any $T$-dimensional vector  $\vec x$ with $L_1$-norm of $1$ we have $\|x\|_{\frac 2 3} = \left(\sum\limits_{i=1}^T x(i)^{\frac 2 3}\right)^{\tfrac 3 2} \leq \sqrt T$. Thus $\|p\|\leq \sqrt T$ and therefore the first of the two terms in the sum is the greater one. The required follows.
	
	Comment: It is evident that the tester given by Valiant and Valiant~\mycite{ValiantV14} solves (w.p. $\geq 2/3$) the problem of identity-testing in the randomized response model using $\Theta(T^{2.5}/\epsilon^2\alpha^2)$ samples. However, it is not a-priori clear why their lower bounds hold for our problem. After all, the set $\varphi(H_1)$ is only a subset of $\{\vec q: d_{\rm TV}(\varphi(\vec p), \vec q) \geq \gamma\alpha\}$. Nonetheless, delving into the lower bound of Valiant and Valiant, the collection of distributions which is hard to differentiate from $\vec p$ given $o\left( \|\vec p\|_{\frac 2 3} \alpha^2 \right)$ samples is given by choosing suitable $\Delta(x)$ and then looking at the ensemble of distributions given by $\{p(x)\pm \Delta(x)\}$ for each $x\in \X$. Luckily, this ensemble is maintained under $\varphi$, mapping each such distribution to $\{\rho + \gamma p(x) \pm \gamma \Delta(x)\}$. The lower bound follows. 
	\end{proof}}
	\toggle{\onlyconf{\vspace{-0.3cm}}The proof uses the results of~\cite{ValiantV14} as a black-box and is mainly composed of calculations, so it is deferred to \myAppendix, Section~\ref{apx_sec:proofs_RR}.}{\proofCorIdentityTestingRR}
	
	\paragraph{Independence Testing.} Another prevalent hypothesis testing over a domain $\X$ where each type is composed of multiple feature is independence testing\onlyfull{ (examples include whether having a STEM degree is independent of gender or whether a certain gene is uncorrelated with cancer)}. Denoting $\X = \X^1 \times \X^2 \times ... \times \X^d$ as a domain with $d$ possible features (hence $T=|\X| = \prod_{j}|\X^j|\stackrel{\rm def}{=}\prod_j T^j$), our goal is to discern whether an observed sample is drawn from a product distribution or a distribution $\alpha$-far from any product distribution. In particular, the null-hypothesis in this case is a complex one: $H_0 = \{  \bar{\vec p} = \vec p^1 \times \vec p^2 \times ... \times \vec p^d \onlyfull{:~ \textrm{for each }j,~\vec p^j \textrm{ is a distrbution over }\X^j  }\}$ and the alternative is $H_1 = \{ {\vec q} : \min_{\bar{\vec p} \in H_0} d_{\rm TV}({\vec q},\bar{\vec p}) \geq \alpha\}$. To the best of our knowledge, the (current) tester with smallest sample complexity is of Acharya et al~\mycite{AcharyaDK15}, which requires $\Omega\left( (\sqrt{T} + \sum_{j} T^j)/\alpha^2\right)$ iid samples.

	We now consider the problem of testing for independence under standard randomized response.\onlyfull{\footnote{Note that if were to implement the feature-wise randomized response (i.e., run Randomize-Response per feature with privacy loss set to $\epsilon/d$) then we would definitely create signals that come from a product distribution. That is why we stick to the straight-forward implementation of Randomized Response even when $\X$ is composed of multiple features.}}	Our goal is to prove the following theorem.
	\begin{theorem}
		\label{thm:independence_testing_RR}
		There exists an algorithm that takes $n = \tilde\Omega( \frac{T^2}{\alpha^2\epsilon^2} \left( d^2(\max\limits_j \{T^j\})^2 + \sqrt{T} \right) )$ signals generated by applying standard randomized response (with $\epsilon<1$) on $n$ samples drawn from a distribution $\vec p$ \onlyfull{over a domain $\X = \X^1 \times... \times \X^d$} and with probability $\geq 2/3$ accepts if $\vec p\in H_0$, or rejects if $\vec p\in H_1$. Moreover, no algorithm can achieve such guarantee using $n = o( |\X|^{5/2}/(\alpha^2 \epsilon^2))$ signals.
	\end{theorem}
	\onlyconf{\vspace{-0.3cm}}
	\noindent Note that has to be at least two types per feature, so $d \leq \log_2(T)$, and if all $T^j$s are the same we have $(T^j)^2 \leq T^{\frac 2 d}$. Thus $T^{2.5}/(\alpha^2\epsilon^2)$ is the leading term in the above bound.
	
	\onlyconf{\noindent\textit{Proof.} }Theorem~\ref{thm:hyp_testing_RR} implies we are comparing $\varphi(H_0) = \{ \rho \vec 1_{\X} + \gamma (\vec p^1\times...\times \vec p^d) \}$ to $\varphi(H_1) = \{ \rho \vec 1_{\X} + \gamma \vec q :~ q\in H_1\}$. Note that $\varphi(H_0)$ is not a subset of product-distributions over~$\X$ but rather a convex combination (with publicly known weights) of the uniform distribution and $H_0$; so  we cannot run the independence tester of Acharya et al on the signals as a black-box. Luckily\onlyfull{~--- and similar to the identity testing case~---} it holds that $\varphi(H_1)$ is far from all distributions in $\varphi(H_0)$: for each $\vec q\in H_1$ and $\bar{\vec p} \in H_0$ we have $d_{\rm TV}(\varphi(\vec q), \varphi(\bar{\vec p})) \geq \gamma d_{\rm TV}(\vec q,\bar{\vec p}) \geq \gamma\alpha$. And so we leverage on the main result of Acharya et al (\mycite{AcharyaDK15}, Theorem~2): we first find a distribution $\rho \vec 1 + \gamma \bar{\vec z}\in \varphi(H_0)$ such that if the signals were generated by some $\rho \vec 1_{\X} + \gamma \bar{\vec p}  \in \varphi(H_0)$ then $d_{\chi^2}(\varphi(\bar{\vec z}), \varphi(\bar{\vec p})) \leq \gamma^2\alpha^2/500$, and then test if indeed the signals are likely to be generated by a distribution close to $\varphi(\bar{\vec z})$ using Acharya et al's algorithm. \onlyfull{Again, we follow the pattern of~\cite{AcharyaDK15} --- we construct $\bar{\vec z}$ as a product distribution $\bar{\vec z} = \vec z^1 \times \vec z^2 \times...\times \vec z^d$ where $\vec z^j$ is devised by projecting each signal onto its $j$th feature. Note that the $j$th-marginal of the distribution of the signals is of the form $T\rho \vec u_{\X^j} + \gamma\vec p^j$ (again, $\vec u_{\X^j}$ denotes the uniform distribution over $\X^j$). Therefore, for each $j$ we derive $\vec z^j$ by first approximating the distribution of the $j$th marginal of the signals via some $\tilde{\vec z}^j$, then we apply the inverse mapping from Corollary~\ref{cor:close_form_solution} to $\tilde{\vec z}^j$ so to get the resulting distribution $\vec z^j$ which we show to approximate the true $\vec p^j$.} We now give our procedure for finding the product-distribution $\bar{\vec z}$.
	
	Per feature $j$, given the $j$th feature of the signals $y^j_1,...,y^j_n$ where each $x^j\in \X^j$ appears $n_{x^j}$ times, our procedure for finding $\vec z^j$ is as follows.
	\onlyconf{\vspace{-0.3cm}}
	\begin{enumerate}
		\setlength\itemsep{0em}
		\setcounter{enumi}{-1}
		\item (Preprocessing:) Denote $\tau = {\alpha}/({10 d \cdot T^j})$. We call any type $x^j$ where $\tfrac{n_{x^j}}n \leq \tfrac{1-\gamma}{T^j} + \gamma\tau$ as \emph{small} and otherwise we say type $x^j$ is \emph{large}. Ignore all small types, and learn $\vec z^j$ only over large types. (For brevity, we refer to $n$ as the number of signals on large types and $T^j$ as the number of large types.) 
		\item Set the distribution $\tilde{\vec z}^j$ as the ``add-1'' estimator of Kamath et al~\mycite{KamathOPS15} for the signals: $\tilde{\vec z}^j(x^j) = \frac {1+n_{x^j}} {T^j + n}$.
		\item Compute $\vec z^j = \tfrac 1 \gamma \left(I - \tfrac {1-\gamma}{T^j} 1_{\X^j}\right)\tilde{\vec z}^j$.
	\end{enumerate}
	\vspace{-0.4cm}
	Once $\vec z^j$ is found for each feature $j$, set $\bar{\vec z} = \vec z^1 \times... \times \vec z^d$ run the test of Acharya et al~\mycite{AcharyaDK15} (Theorem 2) with $\varphi(\bar{\vec z})$  looking only at the large types from each feature, setting the distance parameter to $\tfrac {\alpha\gamma} 2$ and confidence $\tfrac 1 9$, to decide whether to accept or reject.

	In order to successfully apply the Acharya et al's test, a few conditions need to hold. First, the provided distribution $\varphi(\bar{\vec z})$ should be close to $\varphi(H_0)$. This however hold trivially, as $\bar{\vec z}$ is a product-distribution. Secondly, we need that $\varphi(\bar{\vec z})$ and $\varphi(\bar{\vec p})$ to be close in $\chi^2$-divergence, as we argue next.
	
	\DeclareRobustCommand{\lemZCloseToP}{Suppose that $n$, the number of signals, is at least $\Omega( \frac{d^2}{\alpha^2\gamma^2} \max_{j}\{T^j\} )$. Then the above procedure creates distributions $\vec z^j$ such that the product distribution $\bar{\vec z} = \vec z^1 \times \vec z^2 \times ... \times \vec z^d$ satisfies the following property. If the signals $y_1,...,y_n$ were generated by $\varphi(\bar{\vec p}
		)$ for some product-distribution $\bar{\vec p} = \vec p^1 \times ... \times \vec p^d$, then w.p. $\geq 8/9$ we have that $d_{\chi^2}(\varphi(\bar{\vec z}), \varphi(\bar{\vec p})) \leq \gamma^2\alpha^2 / 1000$.}
	\begin{lemma}
		\label{lem:z-close-to-p}
		\lemZCloseToP
	\end{lemma}
	\onlyconf{\vspace{-0.2cm}}
	We table the proof of Lemma~\ref{lem:z-close-to-p} \toggle{to Section~\ref{apx_sec:proofs_RR} in the \myAppendix}{for now}. Next, either completeness or soundness must happen: either the signals were taken from randomized-response on a product distribution\onlyfull{ (were generated using some $\varphi(\bar{\vec p}) \in \varphi(H_0)$)}, or they were generated by a distribution $\gamma\alpha/2$-far  from $\varphi(H_0)$. If no type of any feature was deemed as ``small''\onlyfull{ in our preprocessing stage}, this condition clearly holds; but we need to argue this continues to hold even when we run our tester on a strict subset of $\X$ composed only of large types in each feature. Completeness is straight-forward: since we remove types feature by feature, the types now come from a product distribution $\bar{\vec p}_{\rm large}=\vec p^1_{\rm large} \times ... \times \vec p^d_{\rm large}$ where each $\vec p^j_{\rm large}$ is a restriction of $\vec p^j$ to the large types of feature $j$\onlyfull{, and Lemma~\ref{lem:z-close-to-p} assures us that $\varphi(\bar{\vec z})$ and $\varphi(\bar{\vec p}_{\rm large})$ are close in $\chi^2$-divergence}. Soundness however is more intricate. We partition $\X$ into two subsets: ${\rm AllLarge} = \{(x^1,x^2,...,x^d)\in \X : ~ \forall j,~x^j\textrm{ is large}\}$ and ${\rm Rest}= \X \setminus {\rm AllLarge}$; and break $\vec q$ into $\vec q = \eta \vec q_{\rm Rest} + (1-\eta) \vec q_{\rm AllLarge}$, with $\eta = \Pr_{\vec q}[{\rm Rest}]$. \onlyfull{Using the Hoeffding bound,} Claim~\ref{clm:preprocessing_removes_small_types} \onlyconf{(proof deferred to the \myAppendix)} argues that $\eta < \tfrac \alpha 2$. Therefore, $d_{\rm TV}(\vec q, \vec q_{\rm AllLarge}) \leq \tfrac \alpha 2$, implying that $d_{\rm TV}(\varphi(\vec q_{\rm AllLarge}), \varphi(H_0)) > \alpha\cdot \gamma - \tfrac {\alpha\gamma}2 = \tfrac {\alpha\gamma}2$.	
	\DeclareRobustCommand{\clmPreprocessingRemovesSmallTypes}{Assume the underlying distribution of the samples is $\vec q$ and that the number of signals is at least $n=\Omega(\frac{d^2(\max_j T^j)^2} {\alpha^2\gamma^2} \log(d \max_j T^j) )$. Then w.p. $\geq 8/9$ our preprocessing step marks certain types each feature as ``small'' such that the probability (under $\vec q$) of sampling a type $(x^1,x^2,...,x^d)$ such that $\exists j, x^j \textrm{ is small}$ is $\leq \alpha/2$.}
	\begin{claim}
		\label{clm:preprocessing_removes_small_types}
		 \clmPreprocessingRemovesSmallTypes
	\end{claim}
	\onlyconf{\vspace{-0.2cm}}
	So, given that both Lemma~\ref{lem:z-close-to-p} and Claim~\ref{clm:preprocessing_removes_small_types} hold, we can use the test of Acharya et al, which requires a sample of size $n = \Omega ( \sqrt{T}/(\alpha\gamma)^2  )$. Recall that $\epsilon<1$ so $\gamma = \Theta(\epsilon/T)$, and we get that the sample size required for the last test is $n = \Omega (\frac {T^{2.5}}{\alpha^2\epsilon^2})$. Moreover, for this last part, the lower bound in Acharya et al~\mycite{AcharyaDK15} still holds (for the same reason it holds in the identity-testing case): the lower bound is derived from the counter example of testing whether the signals were generated from the uniform distribution (which clearly lies in $\varphi(H_0)$) or any distribution from a collection of perturbations which all belong to $\varphi(H_1)$ (See~\cite{Paninski08} for more details). Each of distribution is thus $\gamma\alpha$-far from $\varphi(H_0)$ and so any tester for this particular construction requires $\sqrt{T}/(\alpha\gamma)^2$-many samples. \toggle{This proves both the upper- and lower-bounds of Theorem~\ref{thm:independence_testing_RR}.\hfill $\qed$}{Therefore, once we provide the proofs of Lemma~\ref{lem:z-close-to-p} and Claim~\ref{clm:preprocessing_removes_small_types} our proof of Theorem~\ref{thm:independence_testing_RR} is done.}

	\DeclareRobustCommand{\proofLemZCloseToP}{
	\begin{proof} \onlyfull{[Proof of Lemma~\ref{lem:z-close-to-p}.]}
		Fix feature $j$. Let $\vec p^j$ be the marginal distribution of the distribution $\vec p$ which generated the samples (whether $\vec p$ belongs to $H_0$ or $H_1$) on the $j$th feature. It follows that projecting the signals onto their $j$th feature yields $y_1^j, y_2^j,..., y_n^j$ which were generated using $\varphi(\vec p^j) = (1-\gamma)\vec u_{\X_j} + \gamma \vec p^j$. Kamath et al~\mycite{KamathOPS15} have shown that w.p. $\geq 8/9$ it holds that $d_{\chi^2}(\tilde{\vec z}^j, \varphi(\vec p^j)) \leq \frac{9(T^j-1)}{n+1}$. We now apply the linear transformation $\vec z^j = \tfrac 1 \gamma \left(I - \tfrac {1-\gamma}{T^j} 1_{\X^j}\right)\tilde{\vec z}^j$ and similarly note that $\vec p^j = \tfrac 1 \gamma \left(I - \tfrac {1-\gamma}{T^j} 1_{\X^j}\right)\varphi(\vec p^j)$. We note that $\vec z^j$ is a valid probability distribution: for each $x^j\in \X^j$ we have that $\tilde{z}^j(x^j) > \tfrac{1-\gamma}{T^j} + \gamma\tau$ hence $z^j(x^j)> \tau >0$; and since $\sum_{x^j} \tilde z^j(x^j) =1$ then $\sum_{x^j} z^j(x^j) = \tfrac 1 \gamma \sum_{x^j} \tilde z^j(x^j) - \tfrac {1-\gamma}{T^j} = \tfrac 1 \gamma ( 1-(1-\gamma) )=1$. We thus bound the $\chi^2$-divergence between $\vec z^j$ and $\vec p^j$:
		\toggle{
			\begin{align*}
			&d_{\chi^2}(\vec z^j, \vec p^j) = \sum_{x^j} \frac{ \tfrac 1 {\gamma^2}\left( \tilde z^j(x^j) - \tfrac {1-\gamma}{T^j} - \varphi(p^j(x^j)) + \tfrac {1-\gamma}{T^j}  \right)^2  } {\tfrac 1 \gamma\left(\tilde z^j(x^j) - \tfrac {1-\gamma}{T^j}\right)}
			\cr &~ = \frac 1 \gamma \sum_{x^j} \frac{ (\tilde z^j(x^j) -\varphi(p^j(x^j)))^2 } {\tilde z^j(x^j) - \tfrac {1-\gamma}{T^j}} \cr
			&~\stackrel{(\ast)}\leq \frac {30d}{\alpha\gamma^2} \sum_{x^j}\frac{ (\tilde z^j(x^j) -\varphi(p^j(x^j)))^2 } {\tilde z^j(x^j)} \leq \frac {270d}{\alpha\gamma^2}\cdot \frac{T^j-1}{n+1}
			\end{align*}}{
		\begin{align*}
		d_{\chi^2}(\vec z^j, \vec p^j) &= \sum_{x^j} \frac{ \tfrac 1 {\gamma^2}\left( \tilde z^j(x^j) - \tfrac {1-\gamma}{T^j} - \varphi(p^j(x^j)) + \tfrac {1-\gamma}{T^j}  \right)^2  } {\tfrac 1 \gamma\left(\tilde z^j(x^j) - \tfrac {1-\gamma}{T^j}\right)} = \frac 1 \gamma \sum_{x^j} \frac{ (\tilde z^j(x^j) -\varphi(p^j(x^j)))^2 } {\tilde z^j(x^j) - \tfrac {1-\gamma}{T^j}} \cr
		&\stackrel{(\ast)}\leq \frac {30d}{\alpha\gamma^2} \sum_{x^j}\frac{ (\tilde z^j(x^j) -\varphi(p^j(x^j)))^2 } {\tilde z^j(x^j)} \leq \frac {270d}{\alpha\gamma^2}\cdot \frac{T^j-1}{n+1}
		\end{align*}}
		where the inequality in $(\ast)$ follows from the fact that $\tilde z^j(x^j) - \tfrac {1-\gamma}{T^j} > \gamma\tau = \tfrac {\alpha\gamma}{10d} \cdot \tfrac 1 {T^j}$ whereas  $\tilde z^j(x^j) \leq \tfrac {1-\gamma}{T^j} + \gamma \cdot 1 \leq \tfrac{ 1 + \gamma T^j  } {T^j} \leq \tfrac {1+\gamma T}{T^j} \leq \tfrac{1+e^\epsilon-1} {T^j} \leq \tfrac 3 {T^j}$ as $\epsilon<1$.
		
		Next, we use the product lemma from~\cite{Reiss89} (Lemma 3.3.10) (similar argument was made in~\cite{AcharyaDK15}). Assuming that $\sum_{j} d_{\chi^2}(\vec z^j, \vec p^j)\leq 1$ we now have that
		\toggle{
			\begin{align*} &d_{\chi^2}(\bar{\vec z}, \bar{\vec p}) \leq \exp \left(\sum_j d_{\chi^2}(\vec z^j, \vec p^j) \right)-1 
			\leq 2\sum_j d_{\chi^2}(\vec z^j, \vec p^j) \cr &~\leq \frac{600d}{\alpha\gamma^2}\cdot \frac{\sum_j (T^j-1)}{n+1} \leq \frac{600d^2}{\alpha\gamma^2}\cdot \frac{\max_j \{T^j\}}{n+1} \end{align*}}{
		\[ d_{\chi^2}(\bar{\vec z}, \bar{\vec p}) \leq \exp \left(\sum_j d_{\chi^2}(\vec z^j, \vec p^j) \right)-1 \leq 2\sum_j d_{\chi^2}(\vec z^j, \vec p^j) \leq \frac{600d}{\alpha\gamma^2}\cdot \frac{\sum_j (T^j-1)}{n+1} \leq \frac{600d^2}{\alpha\gamma^2}\cdot \frac{\max_j \{T^j\}}{n+1} \] }
		
		Finally, we can obtain the bound we are after:
		\toggle{
			\begin{align*}
			&d_{\chi^2}(\varphi(\bar{\vec z}), \varphi(\bar{\vec p})) = \sum_{\bar x\in \X} \frac{ (\rho + \gamma\bar z(\bar x) - \rho - \gamma\bar p(\bar x))^2  } {\rho + \gamma \bar z(\bar x)} 
			\cr &~~~\leq \gamma^2 \sum_{\bar x\in \X} \frac{ \bar z(\bar x) - \bar p(\bar x))^2  } {\gamma \bar z(\bar x)}  \leq \frac{600d^2}{\alpha\gamma}\cdot \frac{\max_j \{T^j\}}{n+1}
			\end{align*} }{
		\begin{align*}
		d_{\chi^2}(\varphi(\bar{\vec z}), \varphi(\bar{\vec p})) &= \sum_{\bar x\in \X} \frac{ (\rho + \gamma\bar z(\bar x) - \rho - \gamma\bar p(\bar x))^2  } {\rho + \gamma \bar z(\bar x)} \leq \gamma^2 \sum_{\bar x\in \X} \frac{ \bar z(\bar x) - \bar p(\bar x))^2  } {\gamma \bar z(\bar x)} = \gamma d_{\chi^2}(\bar{\vec z}, \bar{\vec p}) \leq \frac{600d^2}{\alpha\gamma}\cdot \frac{\max_j \{T^j\}}{n+1}
		\end{align*}}
		setting $n = \Omega( \frac{d^2}{\alpha^2\gamma^2} \max\limits_{j} \{T^j\} )$ gives the required bound of $d_{\chi^2}(\varphi(\bar{\vec z}), \varphi(\bar{\vec p})) \leq \alpha\gamma/1000$.
	\end{proof}}
	
	\DeclareRobustCommand{\proofClmPreprocessingRemovesSmallTypes}{
	\begin{proof}\onlyfull{[Proof of Claim~\ref{clm:preprocessing_removes_small_types}.]}
		Fix $\vec q$, the distribution that generated the samples. Thus, the signals were generated using the distribution $\varphi(\vec q)$. Fix a feature $j$ and look at the marginal of $\vec q$ with regards to the $j$th feature, $\vec q^j$. We call a type $x^j\in \X^j$ \emph{infrequent} if $q^j(x^j) \leq \tfrac {\alpha}{3d}\cdot \tfrac 1 {T^j}$. We now argue that w.p. $\geq 8/9$ all types deemed as small by our preprocessing step (for all $d$ features) are also infrequent. This follows immediately from the Hoeffding bound: If $x^j$ is frequent then $\varphi(q^j(x^j)) = \tfrac {1-\gamma}{T^j} + \gamma q^j(x^j) \geq \tfrac {1-\gamma}{T^j} +\tfrac {\gamma\alpha}{3d}\cdot \tfrac 1 {T^j}$, but as $x^j$ is deemed small the difference between $\tfrac{n_{x^j}}{n}$ and its expected value is at least $\tfrac {\alpha \gamma} {5d}\tfrac 1 {T^j}$, so Hoeffding assures us this event happens w.p. $\leq \exp(-2n \tfrac {\alpha^2 \gamma^2} {25d^2}\tfrac 1 {(T^j)^2})$. Applying the union bound, the probability that any of $\sum_j T^j$ types that might be deemed as small is actually frequent is thus upper bounded by $\left(\sum_j T^j \right)\cdot \exp(-2n \tfrac {\alpha^2 \gamma^2} {25d^2}\tfrac 1 {(T^j)^2})$. As $n = \Omega(\frac{d^2(\max_j T^j)^2} {\alpha^2\gamma^2} \log(d \max_j T^j)  )$ we infer that this bad event happens with probability $\leq 1/9$.
		
		Now that we've established that all infrequent types are also deemed small in our pre-processing, we bound $\Pr_{\vec q}[ (x^1,..,x^d):~ \exists j~{s.t.}~x^j\textrm{ is small}]$ using the union bound:
		\toggle{\begin{align*}
			&\sum_j \Pr_{\vec q^j} [ x^j :~ x^j\textrm{ is small} ] \leq \sum_j \Pr_{\vec q^j} [ x^j :~ x^j\textrm{ is infrequent} ] \cr
			&~~~=\sum_j \sum_{x^j~\rm{infrequent}} \Pr_{\vec q^j}[ x^j ] = \sum_j \sum_{x^j~\rm{infrequent}}\tfrac {\alpha}{3d}\cdot \tfrac 1 {T^j}
			\cr &~~~\leq \sum_j\tfrac {\alpha}{3d} = \tfrac \alpha 3 \hspace{5cm}\qedhere
			\end{align*}}{
		\begin{align*}
			\sum_j \Pr_{\vec q^j} [ x^j :~ x^j\textrm{ is small} ] &\leq \sum_j \Pr_{\vec q^j} [ x^j :~ x^j\textrm{ is infrequent} ] \cr
			&= \sum_j \sum_{x^j~\rm{infrequent}} \Pr_{\vec q^j}[ x^j ] = \sum_j \sum_{x^j~\rm{infrequent}}\tfrac {\alpha}{3d}\cdot \tfrac 1 {T^j} \leq \sum_j\tfrac {\alpha}{3d} = \tfrac \alpha 3 ~~~ \qedhere
		\end{align*}}
	\end{proof}}

	
	\section{Non-Symmetric Signaling Schemes}
	\label{sec:non-symmetric-RR}
	
	Let us recall the non-symmetric signaling schemes in~\cite{BassilyS15, BassilyNST17}. Each user, with true type $x\in\X$, is assigned her own mapping (the mapping is broadcast and publicly known) $f_i:\X \to \S$.  This sets her inherent signal to  $f_i(x)$, and then she runs standard (symmetric) randomized response \emph{on the signals}, making the probability of sending her true signal $f_i(x)$ to be  $e^\epsilon$-times greater than any other signal~$s~\neq~f_i(x)$.
	
	In fact, let us allow an even broader look. Each user is given a mapping $f_i:\X \to \S$, and denoting $T=|\X|$ and $S=|\S|$, we identify this mapping with a $(S\times T)$-matrix $G_i$. The column $\vec g_i^x= G_i \vec e_x$ is the probability distribution that a user of type $x$ is going to use to pick which signal she broadcasts. (And so the guarantee of differential privacy is that for any signal $s\in \S$ and any two types $x\neq x'$ we have that $g_i^x(s) \leq e^\epsilon g_i^{x'}(s)$.) Therefore, all entries in $G_i$ are non-negative and $\|G_i\|_1=1$ for all $i$s. 
	
	Similarly to the symmetric case, we first exhibit the feasibility of finding a maximum-likelihood hypothesis given the signals from the non-symmetric scheme. Since we view which signal in $\S$ was sent, our likelihood mainly depends on the \emph{row} vectors $\vec g_i^s$. \onlyconf{We prove the following theorem, proof deferred to Section~\ref{apx_sec:proofs_nonsymmetric} in the \myAppendix.}
	
	\DeclareRobustCommand{\thmMaxLikelihoodNonSymmetric}{For any convex set $H$\onlyfull{ of hypotheses}, the problem of finding the max-likelihood $\vec p\in H$ generating the observed non-symmetric signals $(y_1,..,y_n)$ is poly-time solvable.}
	\begin{theorem}
		\label{thm:max_likelihood_non_symmetric}
		\thmMaxLikelihoodNonSymmetric
	\end{theorem}
	\DeclareRobustCommand{\proofThmMaxLikelihoodNonSymmetric}{
	\begin{proof}
		Fix any $\vec p\in H$, a probability distribution on $\X$. Using the public $G_i$ we infer a distribution on $\S$, as
		\toggle{$\Pr[y_i = s] = {\sum\limits_{x\in \X} \Pr[ y_i=s | y_i \textrm{ picked using } G_i\vec e_x ]\cdot\Pr[ \textrm{user $i$ has type $x$} ]}
			= \sum_{x\in \X} \vec e_s\T G_i \vec e_x  \cdot p(x) = \vec e_s\T G_i \left( \sum_{x\in \X} p(x)\vec e_x\right) = \vec e_s\T G_i \vec p \stackrel{\rm def}={\vec g_i^s}\T \vec p$,}{
		\begin{align*}\Pr[y_i = s] &= \sum_{x\in \X} \Pr[ y_i=s |~ y_i \textrm{ chosen using } G_i\vec e_x ]\cdot \Pr[ \textrm{user $i$ is of type $x$} ] 
		\cr &= \sum_{x\in \X} \vec e_s\T G_i \vec e_x  \cdot p(x) = \vec e_s\T G_i \left( \sum_{x\in \X} p(x)\vec e_x\right) = \vec e_s\T G_i \vec p \stackrel{\rm def}={\vec g_i^s}\T \vec p  \end{align*}}
	with $\vec g_i^s$ denoting the row of $G_i$ corresponding to signal $s$.
	
	Therefore, given the observed signals $(y_1,...,y_n) \in \S^n$, the likelihood of any $p$ is given by
	\toggle{$L(\vec p;~ y_1,...,y_n) = \prod_i {\vec g^{y_i}_i}\T \vec p $.}{\[   L(\vec p;~ y_1,...,y_n) = \prod_i {\vec g^{y_i}_i}\T \vec p  \]}
	Naturally, the function we minimize is the negation of the average log-likelihood, namely
	\toggle{ 
		\begin{align} f(\vec p) = -\frac 1 n\sum_{i}\log \left( {\vec g_i^{y_i}}\T \vec p  \right)  \label{eq:loss_func_nonsymmetricRR} \end{align}}{ 
	\begin{equation} f(\vec p) = -\frac 1 n\sum_{i}\log \left( {\vec g_i^{y_i}}\T \vec p  \right)  = -\frac 1 n\sum_i \log\left( \sum_{x\in \X} G_i(y_i,x)p(x)  \right)\label{eq:loss_func_nonsymmetricRR} \end{equation}}
	\onlyfull{whose partial derivatives are:
	$ \frac{\partial f}{\partial x} = -\sum_i \frac{ G_i(y_i,x) }{\sum\limits_{x'\in \X} G_i(y_i,x')p(x') } $,}
	so the gradient of $f$ is given by
	\myinlineequationcomma{ \nabla f = -\frac 1 n\sum_i \frac{1}{{\vec g_i^{y_i}}\T \vec p } \vec g_i^{y_i}}
	and thus, the Hessian of $f$ is 
	\myinlineequationdot{ \nabla^2 f = \frac 1 n \sum_i \frac{1}{({\vec g_i^{y_i}}\T \vec p)^2 } {\vec g_i^{y_i}}{\vec g_i^{y_i}}\T}
	As the Hessian of $f$ is a non-negative sum of rank-$1$ PSD matrices, we have that $\nabla^2 f$ is also a PSD, so $f$ is convex. The feasibility of the problem $\min\limits_{\vec p \in H} f(\vec p)$ for a convex set $H$ follows.
	\end{proof}}
	\toggle{\onlyconf{\vspace{-0.3cm}}}{\proofThmMaxLikelihoodNonSymmetric}

	
	\onlyfull{Note that in our analysis, we inferred that $\Pr[y_i = s] = {\vec g_i^s}\T \vec p$. It follows that the expected fraction of users sending the signal $s$ is ${\rm E}\left[  \tfrac 1 n\sum_i \mathds{1}\{y_i=s\} \right] = \left(\tfrac 1 n \sum_{i=1}^n \vec g_i^s \right)\T \vec p = \vec e_s\T \left(\tfrac 1 n \sum_i G_i \right) \vec p \stackrel {\rm def} = \vec e_s\T G \vec p$. This proposed a similar approach to finding $\vec p \in H$ that suited for maximizing the likelihood of the observed signals. Set $\vec q$ to be a probability vector over $\S$ where $q(s) = \tfrac{n_s}n$ is the fraction of signals that are $s$; and then find a vector $\vec p = \ker(G) + G^\dagger \vec q$ that intersects $H$. While this approach may produce a valid $\vec p$, we focus on the hypothesis testing with guarantees to converge to the true distribution $\vec p$, based on the generation of the matrices $G_i$, as given in the more recent randomized response works.}

	\subsection{Hypothesis Testing under Non-Symmetric Locally-Private Mechanisms}
	\label{subset:non-symmetric-RR}
	
	Let us recap the differentially private scheme of Bassily et al~\mycite{BassilyNST17}. It this scheme, the mechanism uses solely two signals $\S = \{1,-1\}$ (so $S=2$). For every $i$ the mechanism sets $G_i$ by picking u.a.r for each $x\in \X$ which of the two signals in $\S$ is more likely; the chosen signal gets a probability mass of $\frac {e^\epsilon}{1+e^\epsilon}$ and the other get probability mass of $\frac 1 {1+e^\epsilon}$. We denote $\eta$ as the constant such that $\tfrac 1 2 +\eta = \frac {e^\epsilon}{1+e^\epsilon}$ and $\tfrac 1 2-\eta = \frac 1 {1+e^\epsilon}$; namely $\eta = \tfrac{e^\epsilon-1}{2(e^\epsilon+1)} = \Theta(\epsilon)$ when $\epsilon<1$. Thus, for every $s\in \{1,-1\}$ the row vector $\vec g_i^s$ is chosen such that each coordinate is chosen iid and uniformly from $\{\tfrac 1 2+\eta, \tfrac 1 2-\eta\}$. 
	(Obviously, there's dependence between $\vec g_i^1$ and $\vec g_i^{-1}$, as $\vec g_i^1 + \vec g_i^{-1} = \vec 1$, but the distribution of $\vec g_i^1$ is identical to the one of $\vec g_i^{-1}$.)
	
	First we argue that for any distribution $\vec p$, if $n$ is sufficiently large then w.h.p over the generation of the $G_i$s and over the signals we view from each user, then finding $\hat{\vec p}$ which maximizes the likelihood of the observed signals yields a good approximation to $\vec p$. To that end, it suffices to argue that the function we optimize is Lipfshitz and strongly-convex.
	
	\DeclareRobustCommand{\lemstatement}{Fix $\delta>0$ and assume that the number of signals we observe is $n = \Omega(T^3\log(1/\delta))$. Then w.p.$\geq 1-\delta$ it holds that the function $f(\vec p)$ we optimize (as given in Equation~\eqref{eq:loss_func_nonsymmetricRR}) is $\left(3\sqrt{T} \right)$-Lipfshitz and $\left(\tfrac{\eta^2}{2} \right)$-strongly convex over the subspace $\{ \vec x:~ \vec x\T \vec 1 = 0  \}$ (all vectors orthogonal to the all-$1$ vector).}
	\begin{lemma}
		\label{lem:max_likelihood_is_lipfshitz_stronglyconvex}
		\lemstatement
	\end{lemma}
	\onlyconf{\vspace{-0.3cm}}
	The proof of Lemma~\ref{lem:max_likelihood_is_lipfshitz_stronglyconvex} --- which (in part) is hairy due to the dependency between the matrix $G_i$ and the signal $y_i$ --- is deferred to Section~\ref{apx_sec:proofs_nonsymmetric} in the \myAppendix.

	\paragraph{Identity Testing.}
	Designing an Identity Test based solely on the maximum-likelihood is feasible, due to results like Cesa-Binachi et al~\mycite{CesaCG02} which allow us to compare between the risk of the result $\tilde{\vec p}$ of a online gradient descent algorithm to the original distribution $\vec p$ which generated the signals. Through some manipulations one can (eventually) infer that $|f(\vec p)-f(\tilde{\vec p})|=O(1/\sqrt n)$. However, since strong-convexity refers to the $L_2$-norm \emph{squared} of $\|\vec p - \tilde{\vec p}\|$, we derive the resulting bound is $\|\vec p -\tilde{\vec p}\|_1^2 \leq T\|\vec p - \tilde{\vec p}\|^2_2 = O(\tfrac 1 {\eta^2\sqrt{n}})$, which leads to a sample complexity bound proportional to $T^3/(\alpha\eta)^4$. This bound is worse than the bounds in Section~\ref{sec:symmetric_RR}. 
	
	We therefore design a different, simple, identity tester in the local non-symmetric scheme, based on the estimator given in~\cite{BassilyNST17}. The tester itself~--- which takes as input a given distribution $\vec p$, a distance parameter $\alpha>0$ and the $n$ signals~--- is quite simple.
	\onlyconf{\vspace{-0.3cm}}
	\begin{enumerate}
		\setlength\itemsep{0em}
		\item Given the $n$ matrices $G_1,..., G_n$ and the $n$ observed signals $y_1,...,y_n$, compute the estimator ${\vec \theta= \tfrac 1 n \sum_i \tfrac 1 \eta\left( \vec g_i^{y_i} - \tfrac 1 2 \vec 1\right)}$.
		\item If $d_{\rm TV}(\tfrac{1}{2\eta}\vec \theta,\vec p) \leq \tfrac \alpha 2$ then {\tt accept}, else {\tt reject}.
		\onlyconf{\vspace{-0.25cm}}
	\end{enumerate}
	\begin{theorem}
		\label{thm:identity_testing_nonsymmetric_RR}
		Assume $\epsilon<1$. If we observe $n = \Omega(\left( \tfrac{T}{\alpha\epsilon} \right)^2)$ signals generated by a distribution $\vec q$ then w.p. $\geq 2/3$ over the matrices $G_i$ we generate and the signals we observe, it holds that $d_{\rm TV}(\tfrac{1}{2\eta}\vec \theta,\vec q)\leq \alpha/2$.
	\end{theorem}
	\onlyconf{\vspace{-0.3cm}}
	The correctness of the tester now follows from checking for the two cases where either $\vec p = \vec q$ or $d_{\rm TV}(\vec p,\vec q) >\alpha$.
			\DeclareRobustCommand{\proofProVarianceOfTheta}{
				The columns of each $G_i$ are chosen independently, and moreover, the signal $y_i$ depends only on a single column. Therefore, it is clear that for each $x\neq x'$ we have that 
							\myinlineequationcomma{\E \left[ (\theta(x) - 2\eta f(x))(\theta(x') - 2\eta f(x')) \right] = \E [ \theta(x) - 2\eta f(x)]\cdot\E[\theta(x') - 2\eta f(x') ] = 0} so all the off-diagonal entries of the variance-matrix are $0$. And for each type $x\in \X$ we have that 
							\toggle{
								\begin{align*}
								&\E[ (\theta(x) - 2\eta f(x))^2  ] 
								\cr &~= \tfrac 1 {n^2}\sum_{i,i'}\resizebox{0.85\hsize}{!}{ $\E[\left(  \tfrac 1 \eta (g_i^{y_i}(x)-\tfrac 1 2) - 2\eta f(x) \right)
									\left(  \tfrac 1 \eta (g_{i'}^{y_{i'}}(x)-\tfrac 1 2) - 2\eta f(x) \right)]$}
								\cr & \stackrel{(\ast)}= \tfrac 1{n^2} \sum_i \E[  \left(  \tfrac 1 \eta (g_i^{y_i}(x)-\tfrac 1 2) - 2\eta f(x) \right)^2 ] 
								\cr &= \tfrac 1{n^2} \sum_i \resizebox{0.85\hsize}{!}{ $\E[  \left(  \tfrac 1 \eta (g_i^{y_i}(x)-\tfrac 1 2)\right)^2 ] -4\eta f(x)\cdot  \E[\tfrac 1 \eta (g_i^{y_i}(x)-\tfrac 1 2)] + 4\eta^2f(x)^2$}
								\cr& \stackrel{(\ast\ast)}=	\tfrac 1 {n^2} \left( \sum_i 1 - 4\eta f(x)  \cdot 2\eta \sum_i e_{x_i}(x) + \sum_i 4\eta^2f(x)^2 \right) 
								\cr & = \tfrac 1 n  - \tfrac 8 n \eta^2 f(x)^2 + \tfrac 4 n\eta^2 f(x)^2 = \tfrac {1-4\eta^2f(x)^2} n \leq \tfrac 1 n
								\end{align*} where $(\ast)$ follows from the independence of the $i$th and the $i'$th sample, and $(\ast\ast)$ holds because $\tfrac 1 \eta (g_i^{y_i}(x)-\tfrac 1 2) \in \{1,-1\}$.}
							{ 
								\begin{align*}
								\E[ (\theta(x) - 2\eta f(x))^2  ] &= \tfrac 1 {n^2}\sum_{i,i'} \E[\left(  \tfrac 1 \eta (g_i^{y_i}(x)-\tfrac 1 2) - 2\eta f(x) \right)\left(  \tfrac 1 \eta (g_{i'}^{y_{i'}}(x)-\tfrac 1 2) - 2\eta f(x) \right)]
								\intertext{independence between the $i$th sample and the $i'$-th sample gives}
								& = \tfrac 1{n^2} \sum_i \E[  \left(  \tfrac 1 \eta (g_i^{y_i}(x)-\tfrac 1 2) - 2\eta f(x) \right)^2 ] 
								\cr &= \tfrac 1{n^2} \sum_i \E[  \left(  \tfrac 1 \eta (g_i^{y_i}(x)-\tfrac 1 2)\right)^2 ] -4\eta f(x)\cdot  \E[\tfrac 1 \eta (g_i^{y_i}(x)-\tfrac 1 2)] + 4\eta^2f(x)^2
								\intertext{since we always have $\tfrac 1 \eta (g_i^{y_i}(x)-\tfrac 1 2) \in \{1,-1\}$ we get}
								& =	\frac 1 {n^2} \left( \sum_i 1 - 4\eta f(x)  \cdot 2\eta \sum_i e_{x_i}(x) + \sum_i 4\eta^2f(x)^2 \right) 
								\cr & = \tfrac 1 n  - \tfrac 8 n \eta^2 f(x)^2 + \tfrac 4 n\eta^2 f(x)^2 = \tfrac {1-4\eta^2f(x)^2} n \leq \tfrac 1 n		\end{align*}
							}
			}
		\onlyconf{\vspace{-0.25cm}}
		\begin{proof}
		In the first part of the proof we assume the types of the $n$ users were already drawn and are now fixed. We denote $x_i$ as the type of user $i$. We denote the frequency vector $\vec f = \langle \tfrac{n_x} n\rangle_{x\in \X}$, generated by counting the number of users of type $x$ and normalizing it by $n$. 
		
		Given $\vec f$, we examine the estimator $\vec \theta$. For each user $i$ we have that $\tfrac 1 \eta (\vec g_i^{y_i} - \tfrac 1 2 \vec 1) \in \{-1,1\}^T$. Because $x_i$, the type of user $i$, is fixed, then for each coordinate $x'\neq x_i$, the signal $y_i$ is \emph{independent} of the $x'$-column in $G_i$ ($y_i$ depends solely on the entries in the $x_i$-column). We thus have that $g_i^{y_i}(x')$ is distributed uniformly among $\{\tfrac 1 2 \pm \eta\}$ and so $\E[ \tfrac 1 \eta( g_i^{y_i}(x')-\tfrac 1 2) ] = 0$. In contrast, 
		\toggle{ $\Pr[~\tfrac 1 \eta( g_i^{y_i}(x_i)~-~\tfrac 1 2)~=~1~]$\\$ 
			= \sum_{s\in \{-1,1\}}\Pr[\tfrac 1 \eta( g_i^{s}(x_i)-\tfrac 1 2)=1 \textrm{ and } y_i = s]\allowbreak
			{= 2\cdot \tfrac 1 2\cdot (\tfrac 1 2 +\eta)}= \tfrac 1 2 +\eta 
			$. }{
		\begin{align*}
			\Pr[ \tfrac 1 \eta( g_i^{y_i}(x_i)-\tfrac 1 2)=1 ] &= \sum_{s\in \{-1,1\}}\Pr[\tfrac 1 \eta( g_i^{s}(x_i)-\tfrac 1 2)=1 \textrm{ and } y_i = s]
			\cr &= \sum_{s\in \{-1,1\}}\Pr[y_i = s |~ \tfrac 1 \eta( g_i^{s}(x_i)-\tfrac 1 2)=1 ]\cdot \Pr[\tfrac 1 \eta( g_i^{s}(x_i)-\tfrac 1 2)=1]
			\cr &= (\tfrac 1 2 +\eta)\cdot \tfrac 1 2 + (\tfrac 1 2 +\eta)\cdot \tfrac 1 2 = \tfrac 1 2 +\eta 
		\end{align*}}
		Therefore, \toggle{$\E[ \tfrac 1 \eta( g_i^{y_i}(x_i)-\tfrac 1 2) ] = (\tfrac 1 2 + \eta) - (\tfrac 1 2 -\eta) = 2\eta$.}{$\E[ \tfrac 1 \eta( g_i^{y_i}(x_i)-\tfrac 1 2) ] = 1\cdot (\tfrac 1 2 + \eta) + (-1)\cdot (\tfrac 1 2 -\eta) = 2\eta$.} It follows that $\E[\tfrac 1 \eta (\vec g_i^{y_i} - \tfrac 1 2 \vec 1)] = 2\eta \vec e_{x_i}$ and so $\E[\vec \theta] = 2\eta \vec f$.
		
		Next we examine the variance of $\vec \theta$ 
		\toggle{, and argue the following (proof deffered to \myAppendix).
			\vspace{-0.1cm}
			\begin{proposition}
				\label{pro:variance_of_theta}
				$\E [ (\vec \theta - 2\eta\vec f)(\vec \theta - 2\eta\vec f)\T ] \preceq \tfrac 1 n  I$
			\end{proposition}
		\vspace{-0.3cm} So as a result, the expected $L_2$-difference}{
		. We argue that $\E [ (\vec \theta - 2\eta\vec f)(\vec \theta - 2\eta\vec f)\T ] \preceq \tfrac 1 n  I$. \proofProVarianceOfTheta
		
		It thus follows that}
		\myinlineequationdot{{\E[ \|\vec\theta-2\eta\vec f\|^2 ]} = \E[ {\rm trace}((\vec \theta - 2\eta\vec f)(\vec \theta - 2\eta\vec f)\T) ] = {\rm trace}( \E[(\vec \theta - 2\eta\vec f)(\vec \theta - 2\eta\vec f)\T] ) \leq \tfrac T n}
		Chesbyshev's inequality assures us that therefore $\Pr[ \tfrac 1 {2\eta}\|\vec\theta-2\eta\vec f\| > \tfrac{\sqrt{6T}}{2\eta \sqrt{n}} ] \leq \frac{T/n}{6T/n} = \tfrac 1 6$.
		
		So far we have assumed $\vec f$ is fixed, and only looked at the event that the coin-tosses of the mechanism yielded an estimator far from its expected value. We now turn to bounding the distance between $\vec f$ and its expected value $\vec q$ (the distribution that generated the types).
		\onlyfull{ 
			
		}
		Indeed, it is clear to see that the expected value of $\vec f = \tfrac 1 n\sum_i \vec e_{x_i}$  is $\E[\vec f]=\vec q$. Moreover, it isn't hard (and has been computed before many times, e.g.~Agresti~\mycite{Agresti03}) to see that $\E[ (\vec f-\vec q)(\vec f-\vec q)\T ] = \tfrac 1 n \left(  \diag(\vec q) - \vec q \vec q\T\right)$. Thus $\E[ \|\vec f-\vec q\|^2 ] = {\rm trace}(\tfrac 1 n \left(  \diag(\vec q) - \vec q \vec q\T\right))\onlyfull{=\tfrac 1 n \sum\limits_{x\in \X} q(x)-q^2(x)} = \tfrac 1 n (1- \|\vec q\|^2)$. Therefore, applying Chebyshev again, we get that w.p. at most $1/6$ over the choice of types by $\vec q$, we have that
		$\Pr[ \|\vec f-\vec q\| > \sqrt{6/n} ] \leq \tfrac{1/n}{6/n} = \tfrac 1 6$.
		
		Combining both results we get that w.p. $\geq 2/3$ we have that
		\myinlineequation{\|\tfrac 1 {2\eta} \vec \theta - \vec q\|_1 \leq  \sqrt T \|\tfrac 1 {2\eta} \vec \theta - \vec q\| \leq \sqrt T\left( \|\tfrac 1 {2\eta} \vec \theta - \vec f\| + \| \vec f - \vec q\|\right)  \leq \sqrt{ \tfrac{6T^2}{4\eta^2 n}} + \sqrt{ \tfrac {6T} n } \leq \alpha} since we have  $n=\Omega(\tfrac{T^2}{\eta^2\alpha^2})$. Recall that $\eta = \Theta(\epsilon)$ and that $d_{\rm TV}(\vec x,\vec y) = \tfrac 1 2 \|\vec x-\vec y\|_1$, and the bound of $\tfrac \alpha 2 $ is proven.
	\end{proof}
	\toggle{\vspace{-0.3cm}}{
	
	}
	\paragraph{Independence Testing.} Similarly to the identity tester, we propose a similar tester for independence. Recall that in this case, $\X$ is composed of $d$ features, hence $\X = \X^1\times \X^2 \times...\times \X^d$, with our notation $T = |\X|$ and $T^j = |\X^j|$ for each $j$. Our tester should {\tt accept} when the underline distribution over the types is some product distribution $\vec p$, and should {\tt reject} when the underline distribution over the types is $\alpha$-far from any product distribution. \onlyfull{
	
	}The tester, whose input is the $n$ signals and a distance parameter $\alpha~>~0$, is as follows.
	\onlyconf{\vspace{-0.3cm}}
	\begin{enumerate}
		\setlength\itemsep{0em}
		\item Given the $n$ matrices $G_1,..., G_n$ and the $n$ observed signals $y_1,...,y_n$, compute the estimator ${\vec \theta = \tfrac 1 n \sum_i \tfrac 1 \eta\left( \vec g_i^{y_i} - \tfrac 1 2 \vec 1\right)}$.
		\item For each feature $j$ compute $\vec \theta^j$ --- the $j$th marginal of $\tfrac 1 {2\eta}\vec \theta$ (namely, for each $x^j\in\X^j$ sum all types whose $j$th feature is $x^j$). Denote $\bar{\vec \theta} = \vec \theta^1 \times ... \times \vec \theta^d$.
		\item If $d_{\rm TV}(\tfrac{1}{2\eta}\vec \theta,\bar{\vec \theta}) \leq \tfrac \alpha 2$ then {\tt accept}, else {\tt reject}.
	\end{enumerate}
	\onlyconf{\vspace{-0.3cm}}	
	\DeclareRobustCommand{\thmIndependenceNonsymmetricRR}{Assume $\epsilon<1$. Given $n = \Omega(  \frac{T}{\alpha^2\epsilon^2} \left(T +  d^2\sum_j T^j  \right) )$ iid drawn signals from the non-symmetric locally-private mechanism under a dataset whose types were drawn iid from some distribution $\vec q$, then w.p. $\geq 2/3$ over the matrices $G_i$ we generate and the types in the dataset we have the following guarantee. If $\vec q$ is a product distribution, then $d_{\rm TV}(  \tfrac 1 {2\eta} \vec \theta, \bar{\vec \theta}) \leq \tfrac \alpha 2$, and if $\vec q$ is $\alpha$-far from any product distribution then $d_{\rm TV}(  \tfrac 1 {2\eta} \vec \theta, \bar{\vec \theta}) > \tfrac \alpha 2$.}
	\begin{theorem}
		\label{thm:independence_nonsymmetricRR}
		\thmIndependenceNonsymmetricRR
	\end{theorem}
	\DeclareRobustCommand{\proofThmIndependenceNonsymmetricRR}{
	\begin{proof}
		The proof follows the derivations made at the proof of Theorem~\ref{thm:identity_testing_nonsymmetric_RR}. For the time being, we assume the types of the $n$ users are fixed and denote the frequency vector $\vec f = \langle \tfrac{n_x}n \rangle_T$. Moreover, for each feature $j$ we denote the marginal frequency vector as $\vec f^j$. Recall that we have shown that $\E[ \tfrac 1 {2\eta} \vec \theta ] = \vec f$ and that $\E[ (\tfrac 1 {2\eta}\vec \theta-\vec f)(\tfrac 1 {2\eta}\vec \theta-\vec f)\T  ] \preceq \tfrac 1 {4\eta^2 n} I$.
		
		Fix a feature $j$. The way we obtain $\vec \theta^j$ is by summing the entries of $\tfrac 1 {2\eta}\vec \theta$ for each type $x^j \in \X^j$. This can be viewed as a linear operator $M^j$, of dimension $T^j \times T$, where the $x^j$-row of $M^j$ has $1$ for each $x\in \X$ whose $j$-th feature is $x^j$ and $0$ anywhere else. Since each column has a single $1$, it follows that for every two distinct types $x^j$ and $y^j$, the dot-product of the $x^j$-row and the $y^j$-row of $M^j$ is $0$. Thus, since each row has exactly $\prod\limits_{j'\neq j} T^{j'} = \frac T{T^j}$ ones, we have that $(M^j)(M^j)\T = \frac T {T^j} I_{\X^j \times \X^j}$. 
		
		And so, for each feature $j$ we have that $\E[\vec \theta^j] = \E[M^j (\tfrac 1 {\eta}\vec \theta)] = M^j \vec f = \vec f^j$. Moreover, we also have
		\toggle{
			$
			\E[ (\vec \theta^j-\vec f^j)(\vec \theta^j-\vec f^j) ] = \E[ M^j (\tfrac 1 {2\eta}\vec \theta-\vec f)(\tfrac 1 {2\eta}\vec \theta-\vec f)\T   {M^j}\T  ] \preceq  \tfrac T{4\eta^2nT^j} I
			$.}{
		\begin{align*}
		 \E[ (\vec \theta^j-\vec f^j)(\vec \theta^j-\vec f^j) ] &= \E[ M^j (\tfrac 1 {2\eta}\vec \theta-\vec f)(\tfrac 1 {2\eta}\vec \theta-\vec f)\T   (M^j)\T  ] \preceq \tfrac 1 {4\eta^2} (M^j)(M^j)\T = \tfrac T{4\eta^2nT^j} I
		\end{align*}}
		As a result, $\E[ \|\vec \theta^j - \vec f^j\|^2 ] \leq {\rm trace}(\tfrac T{4\eta^2nT^j} I)= \tfrac T {4\eta^2n}$, and the Union-bound together with Chebyshev inequality gives that
		\toggle{
			$\Pr[ \exists~j, \textrm{ s.t. } \| \vec p^j - \vec f^j\| > \tfrac 1 {2\eta}\sqrt{\tfrac {12dT} {n}} ] < \sum_{j=1}^d\tfrac 1 {12 d} = \tfrac 1 {12}$.}{
		\[  \Pr[ \exists~{\rm coordinate}~j, \textrm{ s.t. } \| \vec p^j - \vec f^j\| > \tfrac 1 {2\eta}\sqrt{\tfrac {12dT} {n}} ] < \sum_{j=1}^d\tfrac 1 {12 d} = \tfrac 1 {12}\]}

		We now consider the randomness in $\vec f$.  For every $j$ we denote $\vec q^j$ as the marginal of $\vec q$ on the $j$th feature. Not surprisingly we have that $\E[\vec f] = \vec q$ and that for each feature $\E[ \vec f^j ]= M^j \E[\vec f]= \vec q^j$. Moreover, some calculations give that $\E[ (\vec f^j - \vec q^j)(\vec f^j-\vec q^j)\T  ] = \tfrac 1 n M^j\left( \diag(\vec q ) - \vec q \vec q\T \right)(M^j)\T = \tfrac 1 n \left(  \diag(\vec q^j) - (\vec q^j)(\vec q^j)\T \right)    $. As a result, for each $j$ we have $\E [ \|\vec f^j -\vec q^j\|^2 ] = \tfrac 1 n (1 - \|\vec q^j\|^2)\leq \tfrac 1 n$. Again, the union-bound and the Chebyshev inequality give that
		\myinlineequationdot{\Pr[ \exists j, \textrm{ s.t. } \|\vec f^j - \vec q^j\| > \sqrt{\tfrac{12dT} {n}} ] < \tfrac d {12d} = \tfrac 1 {12}}
		
		And so, w.p. $\geq 5/6$ we get that for each features $j$ we have
		\myinlineequationcomma{\|\vec\theta^j - \vec q^j\|_1 \leq \sqrt{T^j} \|\vec \theta^j  - \vec q^j\| \leq \sqrt{T^j}( 1+\tfrac 1 {2\eta} )\sqrt{\tfrac {12dT} {n}} \leq \sqrt{T^j} \cdot \sqrt{\tfrac {12dT} {\eta^2n}} } where in the last step we used the fact that $\eta<\tfrac 1 2$ hence $( 1+\tfrac 1 {2\eta} )<\tfrac 1 \eta$. We set $n$ large enough to have $\|\vec\theta^j - \vec q^j\|_1 \leq 1$, and in particular it implies that for any $j$ we also have $\|\vec \theta^j\|\leq 2$. We thus apply the bound on the product of the $\vec{\theta}^j$s to derive that (Proposition~\ref{apx_pro:l1norm_tensor} in the 
		\myAppendix)
		\myinlineequationdot{
		\| \vec \theta^1 \times ...\times \vec \theta^d - \vec q^1 \times ... \times \vec q^d\|_1 \leq 2\sum_{j} \sqrt{T^j} \sqrt{\frac {12dT} {\eta^2n}} \leq 2\sqrt d \cdot \sqrt{\sum_j T^j} \sqrt{\frac {12dT} {\eta^2n}} 
		}
 		
 		Moreover, in the proof of Theorem~\ref{thm:identity_testing_nonsymmetric_RR} we have shown that $\Pr[  \|\tfrac 1 {2\eta}\vec \theta  - \vec q\|_1 > \sqrt{ \tfrac{12T^2}{\eta^2n} }  ] < \tfrac 1 6$. In conclusion, setting $n = \Omega( \frac T {\alpha^2\eta^2} \left( T + d^2 \sum_j T^j\right) )$ we have that w.p. $\geq 2 / 3$ both of the following relations holds: 
 		\begin{align*}
 		&\| \bar{\vec \theta} -  \vec q^1 \times ... \times \vec q^d\|_1 \leq \tfrac 1 2 \alpha \cr 
 		&\|\tfrac 1 {2\eta}\vec \theta  - \vec q\|_1 \leq \tfrac 1 2 \alpha 
 		\end{align*}
 		
 		Now, if $\vec q$ is a product distribution that we have that $\vec q = \vec q^1 \times ... \times \vec q^d$ and hence $\|\tfrac 1 {2\eta} \vec \theta - \bar{\vec \theta}\|_1 \leq \alpha$. In contrast, if $\vec q$ is $\alpha$-far (in total-variation distance, and so $(2\alpha)$-far in $L_1$-norm) from any product distribution, then in particular $\|\vec q - \vec q^1 \times ... \times \vec q^d\|_1\geq 2\alpha$ and we get that
 		\myinlineequationdot{\|\tfrac 1 {2\eta} \vec \theta - \bar{\vec \theta}\|_1 \geq \|\vec q - \vec q^1 \times ... \times \vec q^d\|_1 - \| \bar{\vec \theta} -  \vec q^1 \times ... \times \vec q^d\|_1 -  \|\tfrac 1 {2\eta}\vec \theta  - \vec q\|_1 \geq \alpha}
	\end{proof}}
	\toggle{\onlyconf{\vspace{-0.25cm}} (Proof deffered to the \myAppendix, Section~\ref{apx_sec:proofs_nonsymmetric}.)}{\proofThmIndependenceNonsymmetricRR}
	
	\toggle{
	\paragraph{Open Problems.} (1) Is there a tester with a better sample complexity? The experiment in Section~\ref{subsec:experiment} leads us to conjecture that there exists a tester with sample complexity of $T^{1.5}/(\eta\alpha)^2$. There could exist better testers, of smaller sample complexity, which leads to the second question.\\
	(2) Can one derive lower bounds for identity/independence testing in this model, where each sample has its own distribution, related to the original distribution over types? In Section~\ref{apx_sec:figures} in the \myAppendix~we give more details as to possible venues to tackle both problems, relating them to the problem of learning a mixture-model of product distributions.
	}{
	\paragraph{Open Problem.} The above-mentioned testers are quite simple, and it is also quite likely that it is not optimal. In particular, we conjecture that the $\chi^2$-based test we experiment with is indeed a valid tester of sample complexity $T^{1.5}/(\eta\alpha)^2$. Furthermore, there could be other testers of even better sample complexity. Both the improved upper-bound and finding a lower-bound are two important open problem for this setting. We suspect that the way to tackle this problem is similar to the approach of Acharya et al~\mycite{AcharyaDK15}; however following their approach is difficult for two reasons. First, one would technically need to give a bound on the $\chi^2$-divergence between $\tfrac 1 {2\eta}\vec \theta$ and $\vec q$ (or $\vec f$). Secondly, and even more challenging, one would need to design a tester to determine whether the observed collection of random vectors in $\{1,-1\}^{T}$ is likely to come from the mechanism operating on a distribution close to $\tfrac 1 {2\eta}\vec \theta$. This distribution over vectors is a \emph{mixture model} of product-distributions (but not a product distribution by itself); and while each product-distribution is known (essentially each of the $T$ product distributions is a product of random $\{1,-1\}$ bits except for the $x$-coordinate which equals $1$ w.p. $\tfrac 1 2 +\eta$) it is the weights of the distributions that are either $\vec p$ or $\alpha$-far from  $\vec p$. Thus one route to derive an efficient tester can go through learning mixture models~--- and we suspect that is also a route for deriving lower bounds on the tester. A different route could be to follow the maximum-likelihood (or the loss-function $f$ from Equation~\eqref{eq:loss_func_nonsymmetricRR}), with improved convexity bounds proven directly on the $L_1/L_\infty$-norms.
	}
	
	\onlyconf{\vspace{-0.2cm}}
	\subsection{Experiment: Proposed $\chi^2$-Based Testers}
	\label{subsec:experiment}
	
	Following the derivations in the proof of Theorem~\ref{thm:identity_testing_nonsymmetric_RR}, we can see that ${\rm Var}(\vec \theta) = \tfrac 1 n \left(I - 4\eta^2 \diag( \vec f^2 )\right)$. As ever, we assume $\epsilon$ is a small constant and as a result the variance in $2\eta\vec f$ (which is approximately $\tfrac{4\eta^2}n \diag(\vec p)$) is significantly smaller than the variance of $\vec \theta$. This allows us to use the handwavey approximation $\vec f \approx \vec p$, and argue that we have the approximation ${\rm Var}(\vec \theta) \approx \tfrac 1 n \left(I - 4\eta^2 \diag( \vec p^2 )\right)\stackrel{\rm def}{=}\tfrac 1 nM$.
	\onlyfull{
		
	} Central Limit Theorem thus give that
	\myinlineequationdot{{\sqrt{n} M^{-1/2} ( \vec{\theta}-2\eta \vec p) \stackrel{n\to\infty}\rightarrow {\cal N}(0, I)}} Therefore, it stands to reason that the norm of the LHS is distributed like a $\chi^2$-distribution, namely,
	\onlyconf{\vspace{-0.2cm}}
	\[  P(\vec{\theta}) \stackrel{\rm def}{=}~~~n \sum_{x\in \X} \frac{(\theta(x)-2\eta \cdot p(x))^2}{1-4\eta^2 p(x)^2} \stackrel{n\to\infty}{\rightarrow} \chi^2_T\onlyconf{\vspace{-0.3cm}}\]
	Our experiment is aimed at determining whether $P(\vec\theta)$ can serve as a test statistic and assessing its sample complexity.
	
	\paragraph{Setting and Default Values.} We set a true ground distribution on $T$ possible types, $\vec p$. We then pick a distribution $\vec q$ which is $\alpha$-far from $\vec p$ using the counter example of Paninski~\mycite{Paninski08}: we pair the types and \emph{randomly} move $\tfrac {2\alpha}T$ probability mess between each pair of matched types.\onlyfull{\footnote{(Nit-picking) For odd $T$ we ignore the last type and shift $\tfrac {2\alpha} {T-1}$ mass.}} We then generate $n$ samples according to $\vec q$, and apply the non-symmetric $\epsilon$-differentially private mechanism of~\cite{BassilyNST17}. Finally, we aggregate the suitable vectors to obtain our estimator $\vec{\theta}$ and compute $P(\vec \theta)$. If we decide to {\tt accept/reject} we do so based on comparison of $P$ to the $\tfrac 2 3$-quantile of the $\chi^2_T$-distribution, so that in the limit we {\tt reject} only w.p. $1/3$ under the null-hypothesis. We repeat this \emph{entire} process $t$ times.
	\toggle{ We have set the default values}{
		
	Unless we vary a particular parameter, its value is set to the following defaults:} $T=10$, $\vec p = \vec u_T$ (uniform on $[T]$), $\alpha = 0.2$, $n=1000$, $\epsilon = 0.25$ and therefore $\eta = \tfrac 1 2 \tfrac {e^\epsilon-1}{e^\epsilon+1}$, and $t = 10000$. 
	
	\paragraph{Experiment 1: Convergence to the $\chi^2$-distribution in the null case.} First we ask ourself whether our approximation, denoting $P(\vec \theta) \approx \chi^2_T$ is correct when indeed $\vec p$ is the distribution generating the signals. To that end, we set $\alpha=0$ (so the types are distributed according to $\vec p$) and  plot the $t$ empirical values of $P$ we in our experiment, varying both the sample size $n \in \{10,100,1000,10000\}$ and the domain size~$T~\in~\{10,25,50,100\}$.\\
	The \emph{results} are consistent --- $P$ is distributed like a $\chi^2_T$-distribution. Indeed, the mean of the $t$ sample points is~$\approx T$ (the mean of a $\chi^2_T$-distribution). \onlyfull{The only thing we did find (somewhat) surprising is that even for fairly low values of $n$ the empirical distribution mimics quite nicely the asymptotic $\chi^2$-distribution. }The results themselves appear in Figure~\ref{fig:exp1} in the \myAppendix, Section~\ref{apx_sec:figures}.
		
	\paragraph{Experiment 2: Divergence from the $\chi^2$-distribution in the alternate case.} Secondly, we asked whether $P$ can serve as a good way to differentiate between the null hypothesis (the distribution over the types is derived from $\vec p$) and the alternative hypothesis (the distribution over the types if $\geq \alpha$-far from $\vec p$). We therefore ran our experiment while varying $\alpha$ (between $0.25$ and $0.05$) and increasing $n$.\\
	Again, \onlyfull{non surprisingly,} the \emph{results} show that the distribution does shift towards higher values as $n$ increases. \onlyfull{For low values of $n$ the distribution of outputs does seem to be close to the $\chi^2$-distribution, but as $n$ grows, the shift towards higher means begins. }The results are given in Figure~\ref{fig:exp2} in the \myAppendix, Section~\ref{apx_sec:figures}.
	
	\paragraph{Experiment 3: Sample Complexity.} Next, we set to find the required sample complexity for rejection. We fix the $\alpha$-far distribution from $\vec p$, and first do binary search to hone on an interval $[n_L, n_U]$ where the empirical rejection probability is between $30\%-35\%$; then we equipartition this interval and return the $n$ for which the empirical rejection probability is the closest to $33\%$. We repeat this experiment multiple times, each time varying just one of the 3 most important parameters, $T$, $\alpha$ and $\epsilon$. We maintain two parameters at default values, and vary just one parameter: $T~\in~\{5,10,15,..,100\}, \alpha~\in~\{ 0.05,0.1,0.15,...,0.5\}$, $ \epsilon~\in~\{0.05,0.1,0.15,...,0.5\}$.\\
	The \emph{results} are shown in Figure~\ref{fig:exp3}, where next to each curve we plot the curve of our conjecture in a dotted line.\footnote{We plot the dependency on $\alpha$ and $\epsilon$ on the same plot, as both took the same empirical values.}	We conjecture initially that $n \propto T^{c_T}\cdot \alpha^{c_\alpha} \cdot \epsilon^{c_\epsilon}$. And so, for any parameter $\xi \in \{T,\alpha,\epsilon\}$, if we compare two experiments $i,j$ that differ only on the value of this parameter and resulted in two empirical estimations $N_i, N_j$ of the sample complexity, then we get that $c_\xi \approx \frac{\log( N_i/N_j )}  {\log( \xi_i/\xi_j )}$. And so for any $\xi \in \{T,\alpha,\epsilon\}$ we take the median over of all pairs of $i$ and $j$ and we get the empirical estimations of $c_\epsilon = -1.900793, c_\alpha = -1.930947$ and $c_T = 1.486957$. This leads us to the conjecture that the actual sample complexity according to this test is $\frac {T^{1.5}}{\alpha^2\epsilon^2}$.

	\paragraph{Open Problem.} Perhaps even more interesting, is the experiment we wish we could have run: a $\chi^2$-based independence testing. Assuming the distribution of the type is a product distribution $\bar{\vec p} = \vec p^1 \times ... \times \vec p^d$, the proof of Theorem~\ref{thm:independence_nonsymmetricRR} shows that for each feature $j$ we have ${\rm Var}(\vec \theta^j - \vec p^j) \approx \tfrac 1 {4\eta^2 n} \tfrac T {T^j} I_{\X^j}$. Thus $4\eta^2 n \tfrac {T^j}T \|\vec{\theta}^j-\vec p^j\|^2 \stackrel{n\to\infty}{\to} \chi^2_{T^j}$. However, the $d$ estimators $\vec \theta^j$ are not independent, so it is \emph{not true} that $\sum_j 4\eta^2 n \tfrac {T^j}T \|\vec{\theta}^j-\vec p^j\|^2 \stackrel{n\to\infty}{\to} \chi^2_{\sum_j T^j}$. Moreover, even if the estimators of the marginals were independent\toggle{,\footnote{E.g. by assigning each example $i$ to one of the $d$ estimators, costing only $d=\log(T)$ factor in sample complexity}}{(say, by assigning each example $i$ to one of the $d$ estimators, costing only $d=\log(T)$ factor in sample complexity),} we are still unable to determine the asymptotic distribution of $\|\bar{\vec \theta}-\bar{\vec p}\|^2$ (only a bound, scaled by $O(\max_j T_j)$, using Proposition~\ref{apx_pro:l1norm_tensor} in the \myAppendix), let alone the asymptotic distribution of $\|\tfrac 1 {2\eta}\vec{\theta} - \bar{\vec{\theta}}\|^2$.
	
	Nonetheless, we did empirically measure the quantity $Q(\vec\theta) \stackrel{\rm def}=n \sum_x \frac{  (\tfrac 1 {2\eta} \theta(x) - \bar{\theta}(x))^2} {\bar \theta(x)}$ under the null ($\alpha=0$) and the alternative ($\alpha=0.25$) hypothesis with $n = 25,000$ samples in each experiment. The results (given in Figure~\ref{fig:exp4} in the \myAppendix) show that the distribution of $Q$~--- albeit not resembling a $\chi^2$-distribution~--- is different under the null- and the alternative-hypothesis, so we suspect that there's merit to using this quantity as a tester. We thus leave the design of a $\chi^2$-based statistics for independence in this model as an open problem. 

	\toggle{\begin{figure}[h]}{\begin{figure}[t]}
		\centering	\onlyconf{\vspace{-0.27cm}}	\onlyfull{\hspace{-1cm}}
		\begin{subfigure}[h]{0.48\textwidth}
			\toggle{\includegraphics[scale=0.305]{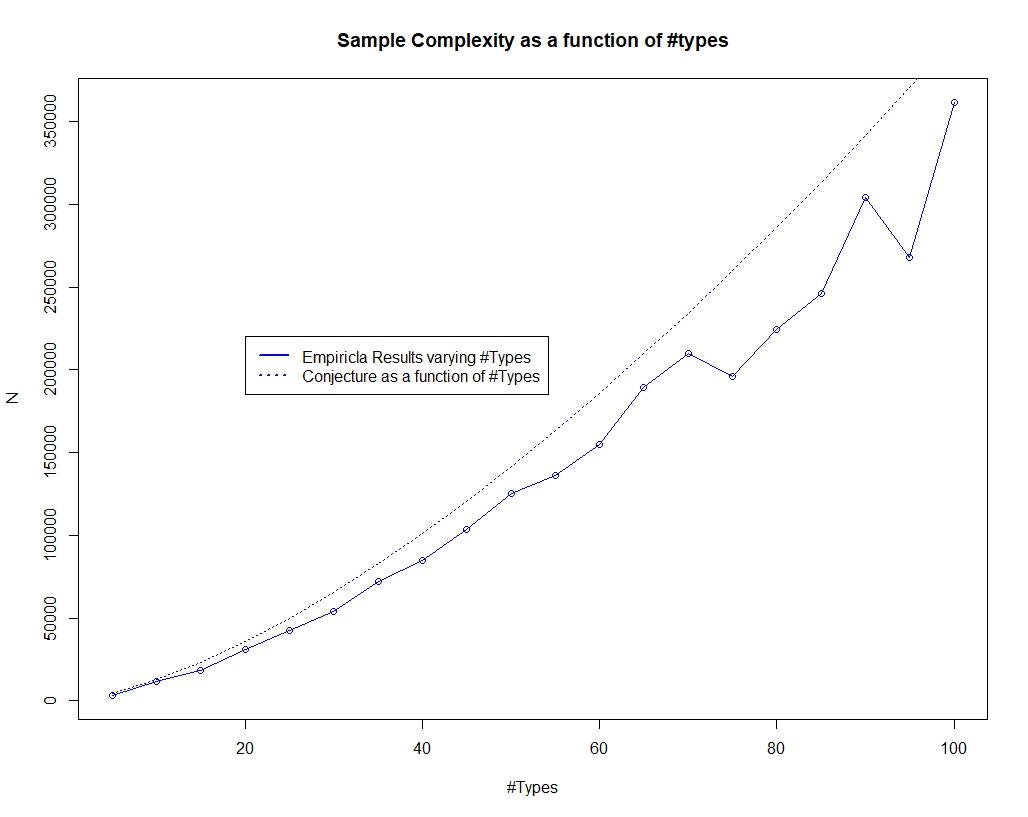}}{\includegraphics[scale=0.305]{exp3_T.jpeg}}
		\end{subfigure}
		\begin{subfigure}[h]{0.48\textwidth}
			\toggle{\includegraphics[scale=0.305]{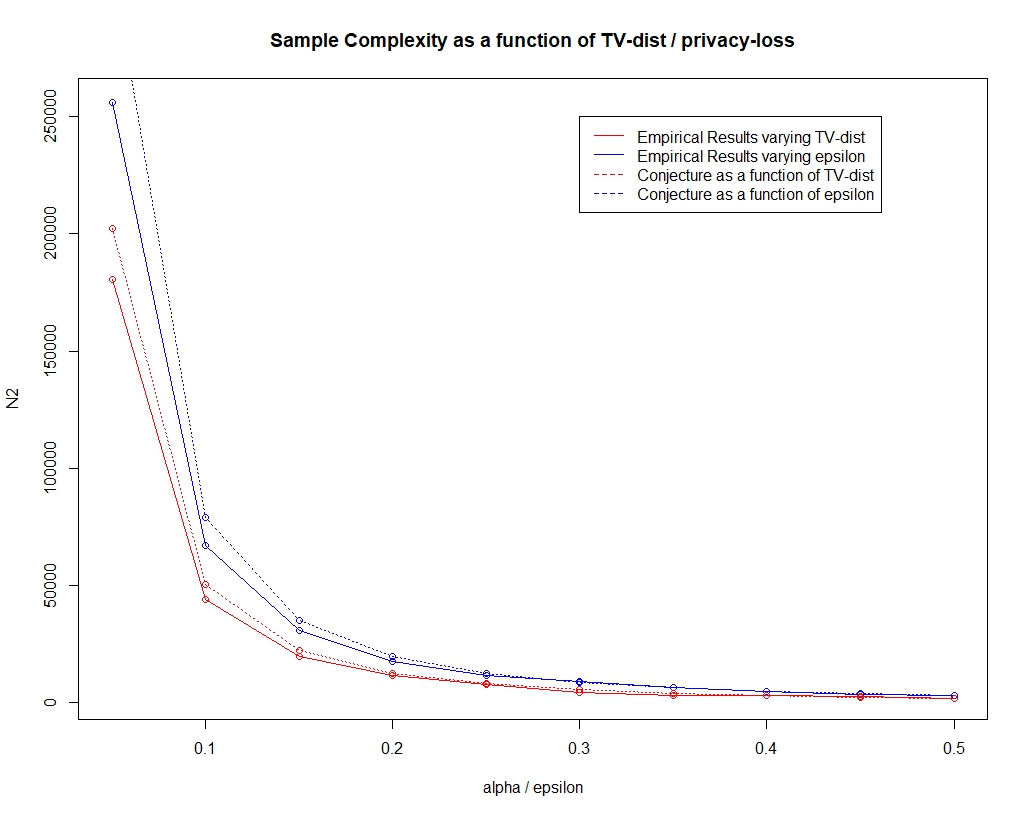}}{\includegraphics[scale=0.305]{exp3_alpha_epsilon.jpeg}}
		\end{subfigure}
		\centering
		\caption{ \label{fig:exp3} Empirical sample-complexity to have the tester reject w.p.$\sim 2/3$ under the alternative hypothesis.\onlyfull{\\  {\small (Best seen in color) }  We used binary search to zoom in on a sample complexity under which the rejection probability is $\approx 2 / 3$. We maintained the default value and only varied one parameters. (Both $\alpha$ and $\epsilon$ take the same empirical values, so we present those results in the same plot.) Next to each curve we present our conjecture for the required sample complexity: $\frac {T^{1.5}}{\alpha^2\epsilon^2}$ (dotted line).} }
	\end{figure}
	
	\onlyconf{\newpage}
	\bibliography{paper}
	
	\onlyconf{\newpage
	\ 
	\newpage}
	\appendix
	\onlyconf{\twocolumn[\center\baselineskip 18pt
		\toptitlebar{\Large\bf Locally Private Hypothesis Testing~---~Supplementary Material}\bottomtitlebar]  \thispagestyle{empty}}
	\section{Additional Claims}
	\label{apx_sec:claims}
	
	\begin{proposition}
		\label{pro:UB_2/3_norm}
		For any $a,b>0$ we have $(a+b)^{2/3} \leq a^{2/3}+b^{2/3}$.
	\end{proposition}
	\begin{proof}
		Let $f(t) \stackrel{\rm def}{=} (t+1)^{2/3} - t^{2/3}-1$. Clearly, $f(0)=0$. Moreover, $f'(t) = \tfrac 2 3\left( (t+1)^{-1/3} - t^{1/3} \right)$. Since $t+1 > t > 0$ it follows that $(t+1)^{-1/3} < t^{-1/3}$ and so $f'(t)< 0$ for any $t \in (0,\infty)$. Therefore, for any $t>0$ we have $f(t)<f(0)=0$. Fix $a,b>0$ and now we have:
		\begin{align*}
		0 &> \left(\tfrac a b + 1\right)^{2/3} - \left(\tfrac a b\right)^{2/3} - 1 = \left(  \tfrac {a+b} b\right)^{2/3} -  \left(\tfrac a b\right)^{2/3} - 1
		\end{align*}
		hence $a^{2/3}+b^{2/3} > (a+b)^{2/3}$.
	\end{proof}
	\begin{proposition}
		\label{pro:bound_3/2_norm}
		For any $a,b>0$ we have $(a+b)^{3/2} = \Theta\left(a^{3/2} + b^{3/2} \right)  $.
	\end{proposition}
	\begin{proof}
		Clearly, due to the non-negativity of $a$ and $b$ we have $(a+b)^{3/2} \leq (2\max\{a,b\})^{3/2} \leq \sqrt{8}(a^{3/2}+b^{3/2})$. Similarly,  $a^{3/2}+b^{3/2} \leq 2(a+b)^{3/2}$.
	\end{proof}
	\begin{claim}
		\label{clm:concentration_gram_matrix}
		Fix two constants $0<\eta<\mu<1$. Let $\vec x_1, \vec x_2, ..., \vec x_n$ be a collection of $n$ vectors in $\R^d$ whose entries are generated iid and uniformly among $\{\mu-\eta,\mu+\eta \}$. If $n = \Omega(\frac{d^2\log^2(d/\delta\eta)}{\eta^2})$ then for any unit-length vector $\vec u \in \R^d$ we have $\tfrac 1 n \sum_i (\vec x_i\T\vec u)^2 > \eta^2/3$. 
	\end{claim}
	\begin{proof}
		Denote $\vec y_1, ..., \vec y_n$ the collection of $n$ vectors such that $\vec y_i = \tfrac 1 \eta \left( \vec x_i - \mu \vec 1\right)$. Therefore, for each $\vec y_i$, its coordinates are chosed iid an uniformly among $\{-1,1\}$. Therefore $\E[\vec y_i]=\vec 0$, and $\E[\vec y_i \vec y_i\T] = I_{d\times d}$. As a $\tfrac {\eta}{12\sqrt d}$-cover of the unit-sphere in $\R^d$ contain $O(\eta^{-d\log(d)})$ points (see Vershynin~\mycite{Vershynin10} for proof), standard Hoeffding-and-union bound yield that \myinlineequationdot{\Pr\left[ \exists \vec u \textrm{ in the cover s.t. } \left|\left(\tfrac 1 n \sum_i \vec y_i \right)\T \vec u - 0\right| >  \tfrac{\eta}{12\sqrt d}\right] \leq (\tfrac 1 \eta)^{100d\log(d)}\cdot 2\exp(-2n \tfrac{\eta^2}{144d}) < \tfrac{\delta}{2}} The triangle inequality thus assures us that for any unit-length vector in $\R^d$ we have $\left|\left(\tfrac 1 n \sum_i \vec y_i \right)\T \vec u\right| \leq  \tfrac{\eta}{6\sqrt d}$.
		Moreover,	standard matrix-concentration results~\cite{Vershynin10} on the spectrum of the matrix $\tfrac 1 n \sum_i \vec y_i \vec y_i\T$ give that w.p. $1-\tfrac\delta 2$ we have that for any unit-length vector $\vec u$ it holds that 
		\[  \vec u \T \left( \tfrac 1 n \sum_i \vec y_i \vec y_i\T - I  \right) \vec u = O\left(  \frac {\sqrt d + \log(1/\delta)} {\sqrt n}  \right) \leq \frac{1}{3}   \] by our choice of $n$. 
		
		Assume both events hold. As for each $i$ we have $\vec x_i = \mu \vec 1 + \eta \vec y_i$ then it holds that $\tfrac 1 n\sum_i \vec x_i \vec x_i\T  = \mu^2 1_{d\times d} + \frac{\mu\eta}n ( \sum_i \vec 1 \vec y_i\T + \vec y_i \vec 1\T ) +  \frac{\eta^2}n \sum_i \vec y_i \vec y_i\T$, thus for each unit-length $\vec u$ we have 
		\begin{align*}
		\frac 1 n \sum_i (\vec x_i\T \vec u)^2 &= \mu^2 (\vec 1\T \vec u)^2 +  {2\mu\eta} (\vec 1\T \vec u)\cdot \left(\tfrac 1 n \sum_i \vec y_i \right)\T \vec u 
		\onlyconf{\cr& ~~~~~~~~~}+ \eta^2 \vec u\T\left(\tfrac 1 n \sum_i \vec y_i \vec y_i\T \right)\vec u
		\cr &\geq 0  - 2\cdot 1 \cdot\eta \sqrt{d}  \cdot \tfrac{\eta}{6\sqrt d} + \eta^2 \cdot \tfrac 2 3= \eta^2/3
		\end{align*}
	\end{proof}
	
	\begin{proposition}
		\label{apx_pro:l1norm_tensor}
		Let $\|\cdot\|$ be any norm satisfying $\|\vec u~\otimes~\vec v\| = \|\vec u\|\|\vec v\|$ (such as the $L_p$-norm fro any $p\geq 1$).
		Let $\vec x_1, \vec x_2,\vec y_1,\vec y_2$ be vectors whose norms are all bounded by some $c$. Then $\| \vec x_1 \otimes \vec y_1 - \vec x_2\otimes \vec y_2\| \leq c\left(\|\vec x_1-\vec x_2\|+\|\vec y_1-\vec y_2\|\right)$.
	\end{proposition}
	\begin{proof}
		\toggle{
			\begin{align*}
			&\| \vec x_1 \otimes \vec y_1 - \vec x_2\otimes \vec y_2\|
			\cr &~~= \| \vec x_1 \otimes \vec y_1 - \vec x_1 \otimes \vec y_2 + \vec x_1\otimes \vec y_2 - \vec x_2\otimes \vec y_2\|
			\cr &~~ \leq  \| \vec x_1 \otimes \vec y_1 - \vec x_1 \otimes \vec y_2\| + \|\vec x_1\otimes \vec y_2 - \vec x_2\otimes \vec y_2\|
			\cr &~~= \|\vec x_1\| \cdot \|\vec y_1 -\vec y_2\| + \|\vec y_2\| \cdot \|\vec x_1-\vec x_2\| 
			\cr &~~\leq c\left(\|\vec x_1 -\vec x_2\| + \|\vec y_1-\vec y_2\|\right)\qquad\qquad\qquad\qedhere
			\end{align*}
			}{
		\begin{align*}
		 \| \vec x_1 \otimes \vec y_1 - \vec x_2\otimes \vec y_2\|
		 & = \| \vec x_1 \otimes \vec y_1 - \vec x_1 \otimes \vec y_2 + \vec x_1\otimes \vec y_2 - \vec x_2\otimes \vec y_2\|
		 \cr & \leq  \| \vec x_1 \otimes \vec y_1 - \vec x_1 \otimes \vec y_2\| + \|\vec x_1\otimes \vec y_2 - \vec x_2\otimes \vec y_2\|
		 \cr &= \|\vec x_1\| \cdot \|\vec y_1 -\vec y_2\| + \|\vec y_2\| \cdot \|\vec x_1-\vec x_2\| 
		 \cr &\leq c\left(\|\vec x_1 -\vec x_2\| + \|\vec y_1-\vec y_2\|\right)\qquad\qquad\qquad\qedhere
		 \end{align*}
		}
	\end{proof}
	
	\onlyconf{
	\section{Missing Proofs: Symmetric Scheme}
	\label{apx_sec:proofs_RR}

	\begin{corollary}[Corollary~\ref{cor:close_form_solution} restated.]
		\corCloseFormSolution
	\end{corollary}
	\proofCorCloseFormSolution
	
	\begin{claim}[Claim~\ref{clm:likelihood_strongly_convexity_RR}]
		\clmLikelihoodStronglyConvexityRR
	\end{claim}
	\proofClmLikelihoodStronglyConvexityRR
	
	\begin{corollary}[Corollary~\ref{cor:identity_testing_RR} restated.]
		\corIdentityTestingRR
	\end{corollary}
	\proofCorIdentityTestingRR
	
	\begin{lemma}[Lemma~\ref{lem:z-close-to-p} restated.]
		\lemZCloseToP
	\end{lemma}
	\proofLemZCloseToP
	
	\begin{claim}[Claim~\ref{clm:preprocessing_removes_small_types} restated.]
	\clmPreprocessingRemovesSmallTypes
	\end{claim}
	\proofClmPreprocessingRemovesSmallTypes
	}
	
	\toggle{ \section{Missing Proofs: Non-Symmetric Scheme} }{ \section{Missing Proofs} }
	\label{apx_sec:proofs_nonsymmetric}
	
	\onlyconf{
	\begin{theorem}[\ref{thm:max_likelihood_non_symmetric} restated.]
		\thmMaxLikelihoodNonSymmetric
	\end{theorem}
	\proofThmMaxLikelihoodNonSymmetric
	}
	
	\begin{lemma}[Lemma~\ref{lem:max_likelihood_is_lipfshitz_stronglyconvex} restated.]
		\lemstatement
	\end{lemma}
	
	\begin{proof}
		Once the $G_i$s have been picked, we view the $n$ signals and are face with the maximum-likelihood problem as defined in Theorem~\ref{thm:max_likelihood_non_symmetric}. As a result of this particular construction, it is fairly evident to argue that the function $f$ whose minimum we seek is Lipfshitz: the contribution of each user to the gradient of $f$ is $({{\vec g_i^{y_i}} \T \vec p})^{-1} {\vec g_i^{y_i}}$. Since our optimization problem is over the probability simplex, then for each $\vec p$ we always have ${\vec g_i^{y_i}} \T \vec p \geq \tfrac 1 2-\eta >\tfrac 1 4$, whereas $\|{\vec g_i^{y_i}}\| \leq (\tfrac 1 2+\eta) \sqrt T \leq \tfrac 3 4 \sqrt T$. Therefore, our function $f(\vec p)$ is $({3\sqrt T})$-Lipfshitz.
		
		The argument which is hairier to make is that $f$ is also $\Theta(\eta^2)$-strongly convex over the subspace orthogonal to the all-$1$ vector; namely, we aim to show that for any unit-length vector $\vec u$ such that $\vec u\T\vec1 = 0$ we have that $\vec u (\nabla^2 f)\vec u \geq \tfrac{\eta^2}{2}$. Since each coordinate of each $\vec g_i^s$ is non-negative and upper bounded by $\tfrac 1 2+\eta \leq 1$, then it is evident that for any probability distribution $\vec p$ and any unit-length vector $\vec u$ we have $\vec u\T (\nabla^2 f(\vec p))\vec u \geq \tfrac 1 n \sum_{i} \tfrac{ ({\vec g_i^{y_i}}\T \vec u)^2 }{({\vec g_i^{y_i}}\T \vec p)^2} \geq \tfrac 1 n \sum_i ({\vec g_i^{y_i}}\T \vec u)^2$, it suffices to show that the least eigenvalue of $\left( \tfrac 1 n \sum_i {\vec g_i^{y_i}}{\vec g_i^{y_i}}\T \right)$ which is still orthogonal to $\vec 1$  is at least $\eta^2/2$.
		
		Let $h(\vec v_1,..., \vec v_n)$ be the function that maps $n$ vectors in $\{\tfrac 1 2+\eta,\tfrac 1 2-\eta\}^T$ to the least-eigenvalue of the matrix $\tfrac 1 n \sum_i \vec v_i\vec v_i\T$ on ${\cal U} = \{ \vec x\in \R^T:~ \vec x\T \vec 1 = 0  \}$. As ever, our goal is to argue that w.h.p we have that $h(\vec g_1^{y_1},....,\vec g_n^{y_n}) \approx \E_{\vec g_i^{y_i}}[h(\vec g_1^{y_1},....,\vec g_n^{y_n})]$. However, it is unclear what is $\E_{\vec g_i^{y_i}}[h(\vec g_1^{y_1},....,\vec g_n^{y_n})]$, and the reason for this difficulty lies in the fact that at each day $i$ we pick either the ``1''-signal or the ``-1''-signal \emph{based on the choice of $\vec g_i^{1}$ and $\vec g_i^{-1}$}. Namely, for each user $i$, the user's type $x$ is chosen according to $\vec p$ and that is independent of $G_i$. However, once $G_i$ is populated, the choice of the signal is determined by the column corresponding to type $x$. Had it been the case that each user's signal is fixed, or even independent of the entries of $G_i$, then it would be simple to argue that w.h.p. the value of $h$ is $\geq \eta^2/3$. However, the dependence between that two row vectors we choose for users and the signal sent by the user makes arguing about the expected value more tricky.

		So let us look at $\sum_i {\vec g_i^{y_i}}{\vec g_i^{y_i}}\T$. The key to unraveling the dependence between the vectors $\vec g_i^1, \vec g_i^{-1}$ and the signal $y_i$ is by fixing the types of the $n$ users in advance. After all, their type is chosen by $\vec p$ independently of the matrices $G_i$. Now, once we know user $i$ is of type $x_i$ then the signal $y_i$ is solely a function of the $x_i$-column of $G_i$, \emph{but the rest of the columns are independent of $y_i$}. Therefore, every coordinate $g_i^{y_i}(x')$ for any $x'\neq x_i$ is still distributed uniformly over $\{\tfrac 1 2+\eta, \tfrac 1 2-\eta \}$, and simple calculation shows that
		\toggle{ 
			\begin{align*} &\Pr[g_i^{y_i}(x_i) = \tfrac 1 2+\eta] 
			\cr &~~=  \sum_{s\in\{1,-1\}} \Pr[ g_i^{s}(x_i) = \tfrac 1 2+\eta \textrm{ and } y_i = s] 
			\cr & ~~= \sum_{s\in \{1,-1\}} \Pr[ y_i = s | g_i^s(x_i)=\tfrac 1 2+\eta]\Pr[g_i^s(x_i)=\tfrac 1 2+\eta]
			\cr & ~~= 2\cdot \tfrac 1 2  (\tfrac 1 2+\eta) = \tfrac 1 2+\eta
			\end{align*}}{
		\begin{align*} \Pr[g_i^{y_i}(x_i) = \tfrac 1 2+\eta] &=  \sum_{s\in\{1,-1\}} \Pr[ g_i^{s}(x_i) = \tfrac 1 2+\eta \textrm{ and } y_i = s] 
		\cr & = \sum_{s\in \{1,-1\}} \Pr[ y_i = s | g_i^s(x_i)=\tfrac 1 2+\eta]\Pr[g_i^s(x_i)=\tfrac 1 2+\eta]
		\cr & = \tfrac 1 2 \sum_{s\in \{1,-1\}} \Pr[y_i = s | g_i^s(x_i)=\tfrac 1 2+\eta] = \tfrac{ (\tfrac 1 2+\eta) + (\tfrac 1 2+\eta)  } 2 = \tfrac 1 2+\eta
		\end{align*}}
		hence, $\Pr[g_i^{y_i}(x_i)= \tfrac 1 2-\eta] = \tfrac 1 2-\eta$ and so $\E [ g_i^{y_i}(x_i)] = (\tfrac 1 2+\eta)^2+(\tfrac 1 2-\eta)^2 = 2\tfrac 1 4 + 2\eta^2 = \tfrac 1 2+2\eta^2$.
		We thus have that $\E[\vec g_i^{y_i}] = \tfrac 1 2\vec 1 + 2\eta^2 \vec e_{x_i}$.

		Note that $\E[ (g_i^{y_i}(x_i)-\tfrac 1 2)^2  ] = \eta^2$ as we always have that $g_i^{y_i}(x_i)-\tfrac 1 2 \in \{-\eta,\eta\}$. Thus, $\E [ ({\vec g_i^{y_i}-\tfrac 1 2 \vec 1})({\vec g_i^{y_i}}-\tfrac 1 2 \vec 1)\T ] = \eta^2 I$. It follows that 
		\myinlineequationdot{\E [ {\vec g_i^{y_i}}{\vec g_i^{y_i}}\T ]  = \eta^2 I -(\tfrac 1 2)^2 1_{\X \times \X} + \tfrac 1 2 \left( \E [ {\vec g_i^{y_i}}] \vec 1\T + \vec 1 \E [ {\vec g_i^{y_i}}]\T\right) = (\tfrac 1 2)^2 1_{\X\times \X} + \eta^2 I + \eta^2 \left(\vec e_{x_i}\vec 1\T + \vec 1 \vec e_{x_i}\T\right)}
		
		We therefore have that for any unit length $\vec u$ which is orthogonal to $\vec 1$ we have $\E[ \vec u\T \sum_i {\vec g_i^{y_i}}{\vec g_i^{y_i}}\T \vec u ] = \eta^2n$, or in other words: $\E [ P_{\cal U}(\tfrac 1 n \sum_{i} {\vec g_i^{y_i}}{\vec g_i^{y_i}}\T)  ]=\eta^2 I$ (with $P_{\cal U}$ denoting the projection onto the subspace ${\cal U}$). The concentration bound for any unit-length $\vec u\in{\cal U}$ follows from standard Hoeffding and Union bounds on a $\tfrac 1 4$-cover of the unit-sphere in ${\cal U}$. The argument is standard and we bring it here for completion.
		
		Let $\vec u_1,..., \vec u_m$ be a $\tfrac 1 8$-cover of the unit-sphere in ${\cal U}$. Standard arguments (see Vershynin~\mycite{Vershynin10} Lemma 5.2) give that $m  = O(20^T)$. Moreover, for any matrix $M$, suppose we know that for each $\vec u_j$ it holds that $\tfrac 3 4\eta^2 < \vec u_j\T M \vec u_j \leq \|M\vec u_j\|$. Then let $\vec u$ be the unit-length $\vec u\in {\cal U}$ on which $\vec u\T M\vec u$ is minimized (we denote the value at $\vec u$ as $\sigma_{\min}(M)$) and let $\vec u_j$ its vector in the cover. Then we get 
		\begin{align*} \sigma_{\min}(M) &= \vec u\T M \vec u  
		\cr & = \vec u_j\T M \vec u_j - \vec u_j\T M (\vec u_j - \vec u) + (\vec u_j-\vec u)\T M \vec u 
		\cr &\geq \tfrac 3 4 \eta^2 - \|M\vec u_j\| \cdot \tfrac 1 8 -\tfrac 1 8 \sigma_{\min}(M)
		\cr \Rightarrow \tfrac 9 8 \sigma_{\min}(M) &\geq \tfrac 5 8 \eta^2
		\end{align*} so $\sigma_{\min}(M) > \eta^2/2$. We therefore argue that for each $\vec u_j$ it holds that $\Pr[  \vec u_j\T \left(\tfrac 1 n \sum_{i} {\vec g_i^{y_i}}{\vec g_i^{y_i}}\T \right)\vec u_j < \tfrac 3 4 \eta^2  ] < \delta/20^T$ and then by the union-bound the required will hold.
		
		Well, as shown, $\E[\vec u_j\T \left(\tfrac 1 n \sum_{i} {\vec g_i^{y_i}}{\vec g_i^{y_i}}\T \right)\vec u_j] = \eta^2$. Denote $X_i$ as the random variable $(\vec u_j\T {\vec g_i^{y_i}})^2$ and note that due to orthogonality to $\vec 1$ we have that \[ 0\leq X_i = (\vec u_j \T ({\vec g_i^{y_i}} -\tfrac 1 2 \vec 1))^2 \leq \|\vec u_j\|^2 \cdot \|\vec g_i^{y_i}-\tfrac 1 2\vec 1\|^2 = \eta^2 T
		\] The Hoeffding bound now assures us that $\Pr[\vec u_j\T \left(\tfrac 1 n \sum_{i} {\vec g_i^{y_i}}{\vec g_i^{y_i}}\T \right)\vec u_j < \tfrac 3 4 \eta^2  ] = \Pr[ \vec u_j\T \left(\tfrac 1 n \sum_{i} {\vec g_i^{y_i}}{\vec g_i^{y_i}}\T \right)\vec u_j - \E[\vec u_j\T \left(\tfrac 1 n \sum_{i} {\vec g_i^{y_i}}{\vec g_i^{y_i}}\T \right)\vec u_j]< -\tfrac 1 4 \eta^2  ] \leq \exp(\tfrac{-2n^2\eta^4/16} {n\cdot \eta^4T^2} ) = \exp(-n/8T^2) \leq \delta/20^T$ for $n = \Omega(T^3 \ln(1/\delta))$.
	\end{proof}
	
	\onlyconf{
	\begin{proposition}[Proposition~\ref{pro:variance_of_theta} restated.]
		
		\[\E [ (\vec \theta - 2\eta\vec f)(\vec \theta - 2\eta\vec f)\T ] \preceq \tfrac 1 n  I\]
	\end{proposition}
	\begin{proof}
	\proofProVarianceOfTheta
	\end{proof}
}

	\onlyconf{		
	\begin{theorem}[Theorem~\ref{thm:independence_nonsymmetricRR} restated.]
		\thmIndependenceNonsymmetricRR
	\end{theorem}
	\proofThmIndependenceNonsymmetricRR
	}

	\section{Additional Figures}
	\label{apx_sec:figures}
	
	For completion, we bring here the results of our experiments.
	
	Figure~\ref{fig:exp1} details the empirical distribution of $P(\vec \theta)$ we get under the null-hypothesis, under different sample complexities ($n=\{10,100,1000,10000\}$) for different sizes of domains ($T=\{10,25,50,100\}$). Next to the curves we also draw the curve of the $\chi^T$-distribution. Since all curves are essentially on top of one another, it illustrates our point: the distribution of $P(\vec \theta)$ under the null-hypothesis is (very close) to the $\chi^2_T$-distribution.
	
	\toggle{\begin{figure*}[b]
			\centering
			\begin{subfigure}[t]{0.45\textwidth}
				\centering
				\includegraphics[scale=0.25]{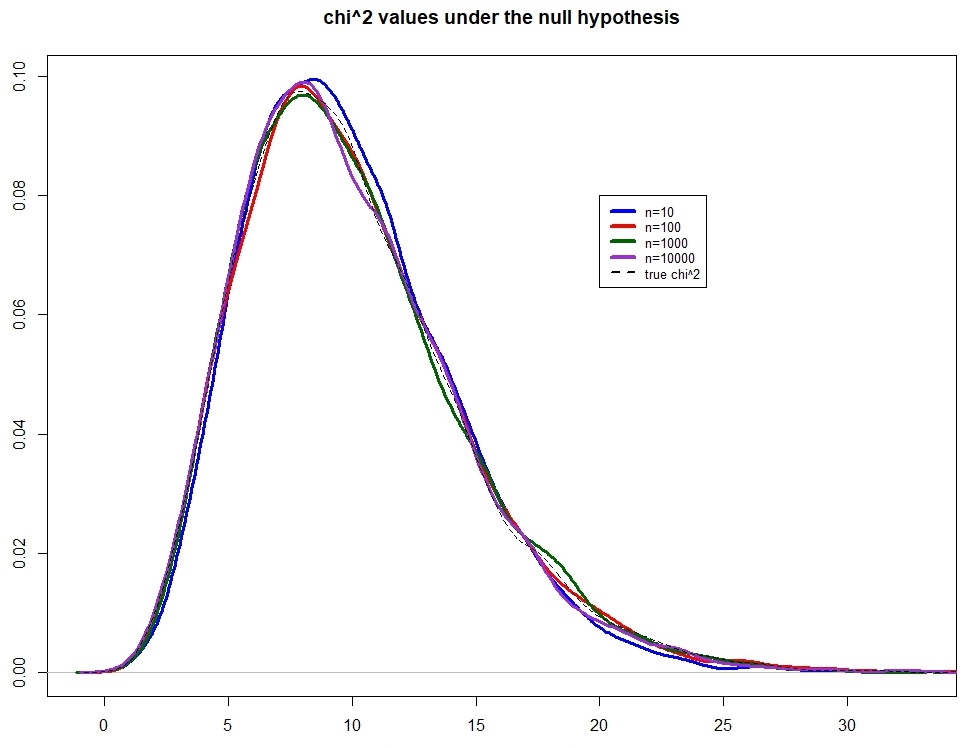}
				\caption{\#Types = 10}
			\end{subfigure}
			\begin{subfigure}[t]{0.45\textwidth}
				\centering
				\includegraphics[scale=0.25]{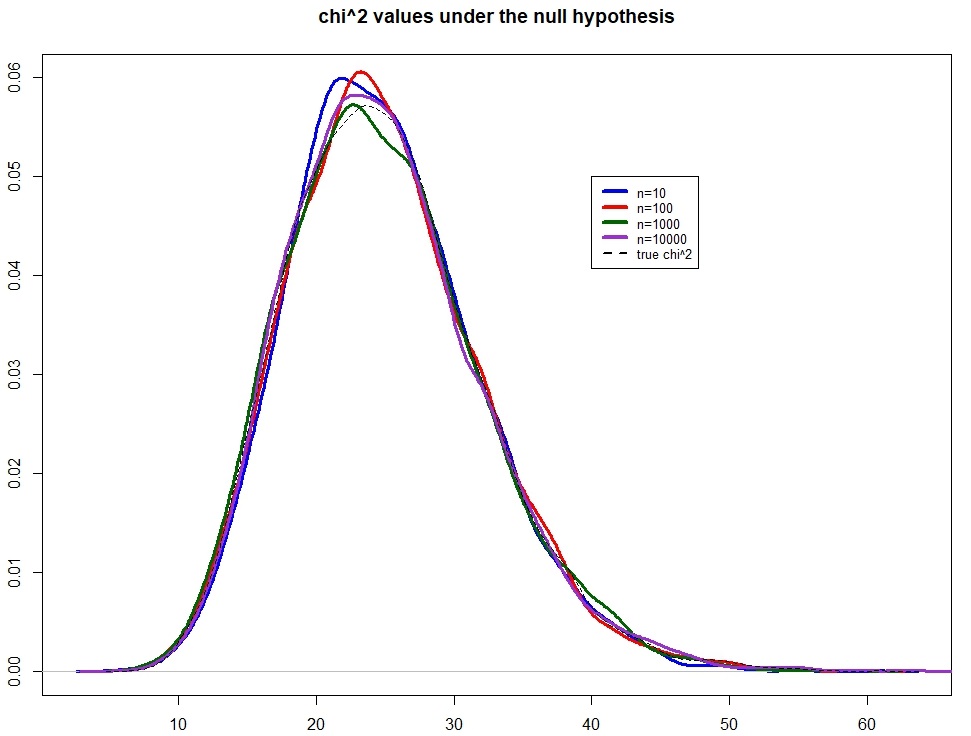}
				\caption{\#Types = 25}
			\end{subfigure}
			\begin{subfigure}[b]{0.45\textwidth}
				\centering
				\includegraphics[scale=0.25]{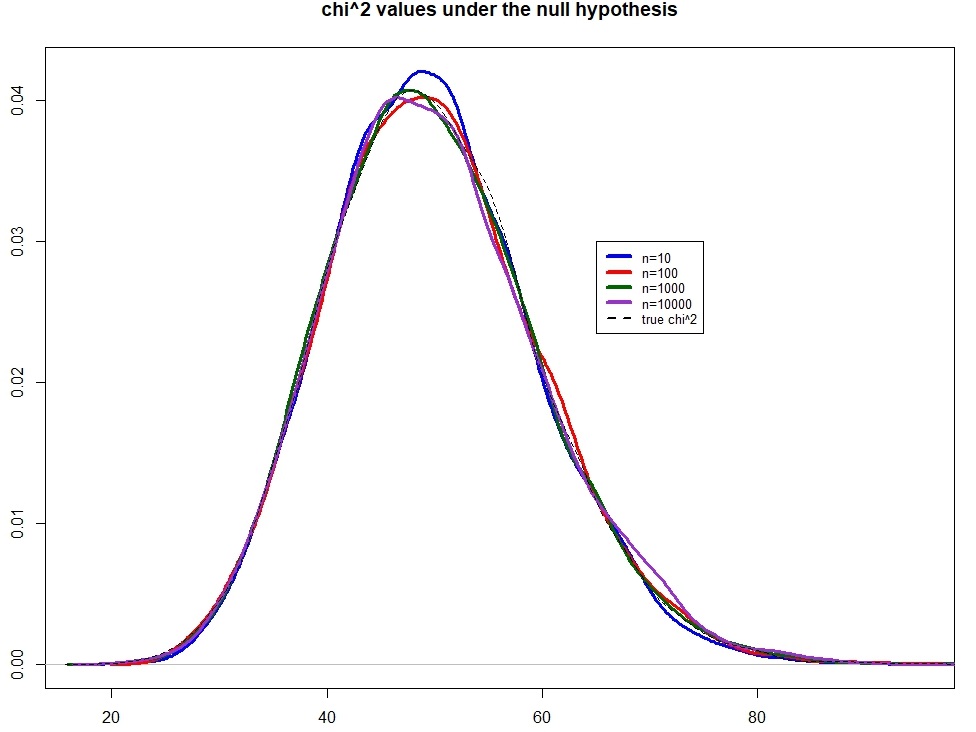}
				\caption{\#Types = 50}
			\end{subfigure}
			\begin{subfigure}[b]{0.45\textwidth}
				\centering
				\includegraphics[scale=0.25]{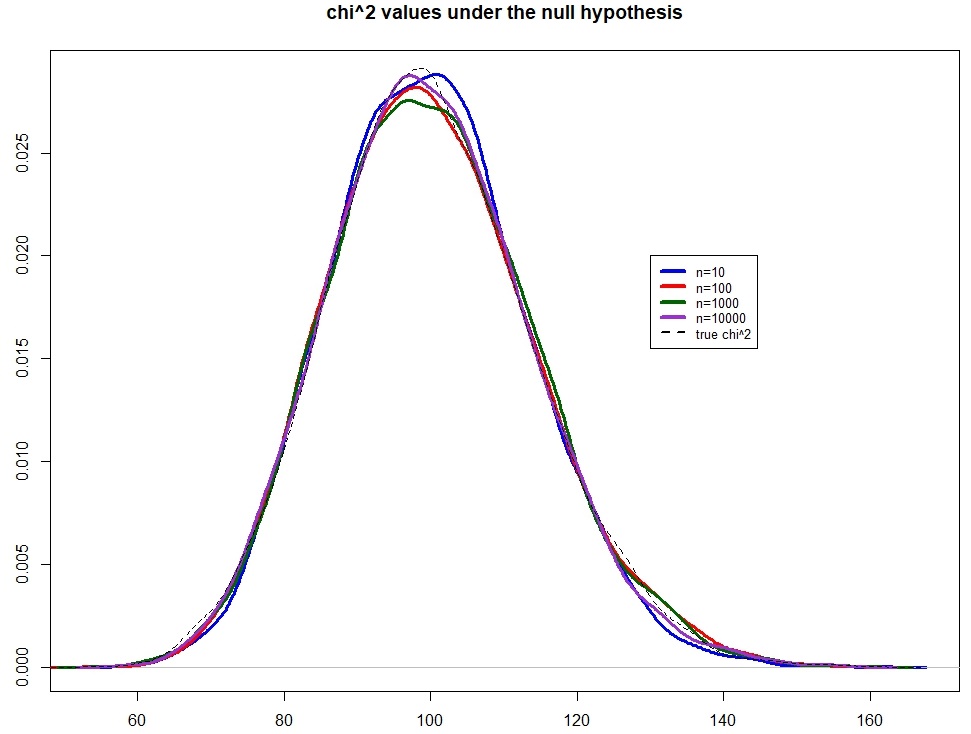}
				\caption{\#Types = 100}
			\end{subfigure}
			\captionsetup[subfigure]{justification=centering}
			\captionsetup{width=.825\linewidth}
			\caption{{\label{fig:exp1} The empirical distribution of our test quantity under the null-hypothesis.}\\ {{\small (Best seen in color) We ran our $\chi^2$-based test under the null-hypothesis. Not surprisingly, the results we get seem to be taken from a $\chi^2$-distribution (also plotted in a dotted black line). In all of the experiments we set $\epsilon=0.25$.}}
				 }
		\end{figure*}}{\begin{figure}[h]
		\centering
		\begin{subfigure}[h]{0.48\textwidth}
			\includegraphics[scale=0.3]{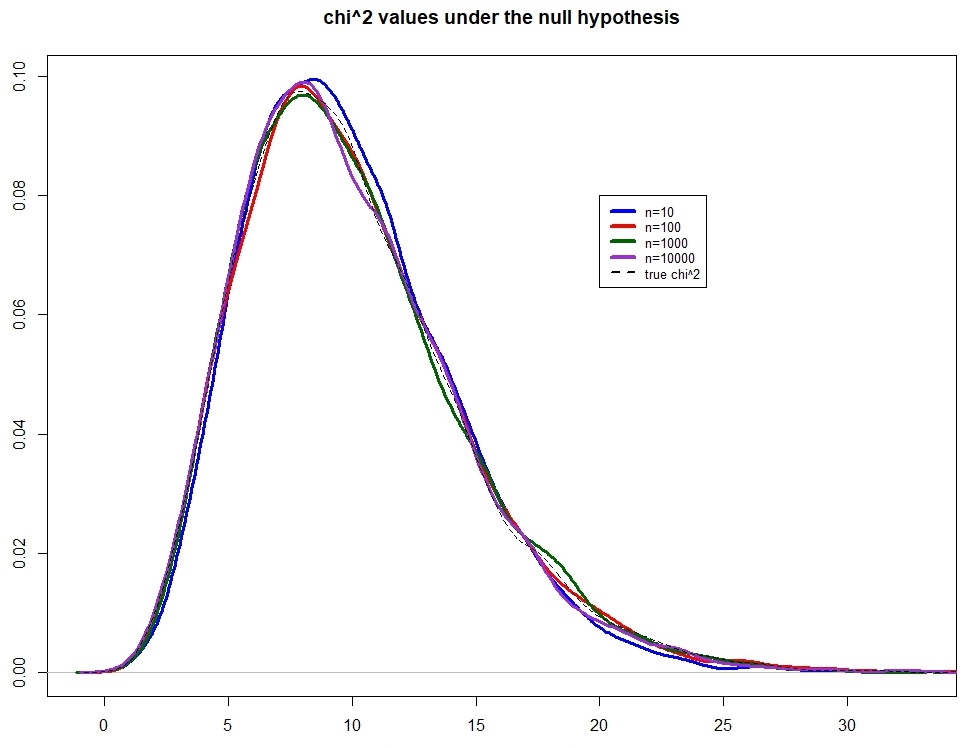}
			\caption{\#Types = 10}
		\end{subfigure}
		\begin{subfigure}[h]{0.48\textwidth}
			\includegraphics[scale=0.3]{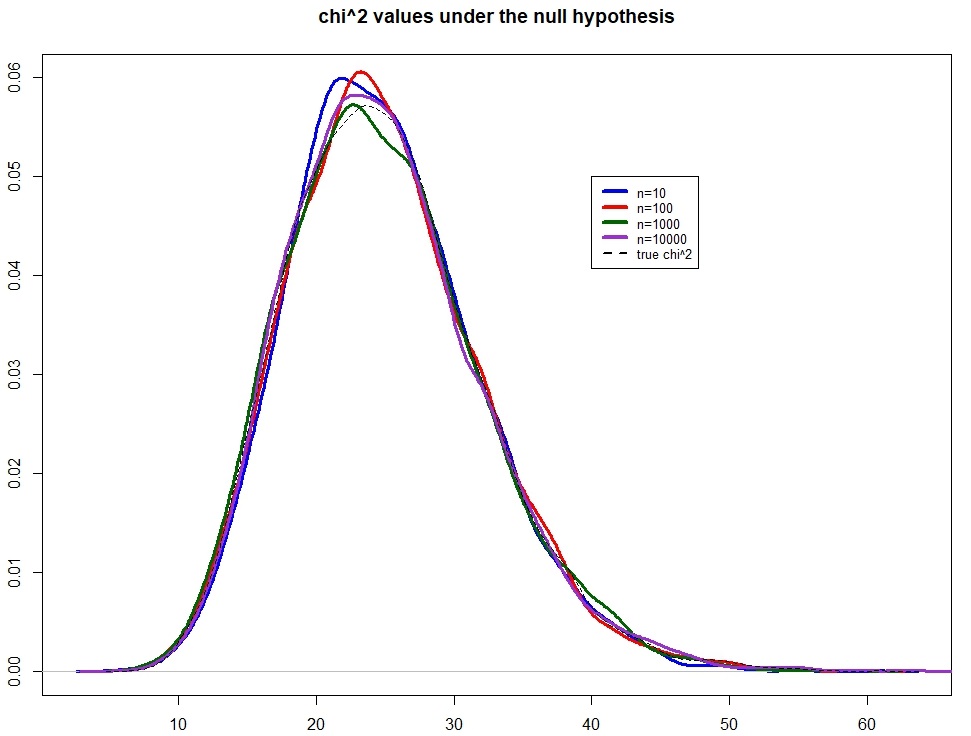}
			\caption{\#Types = 25}
		\end{subfigure}
		\begin{subfigure}[h]{0.48\textwidth}
			\includegraphics[scale=0.3]{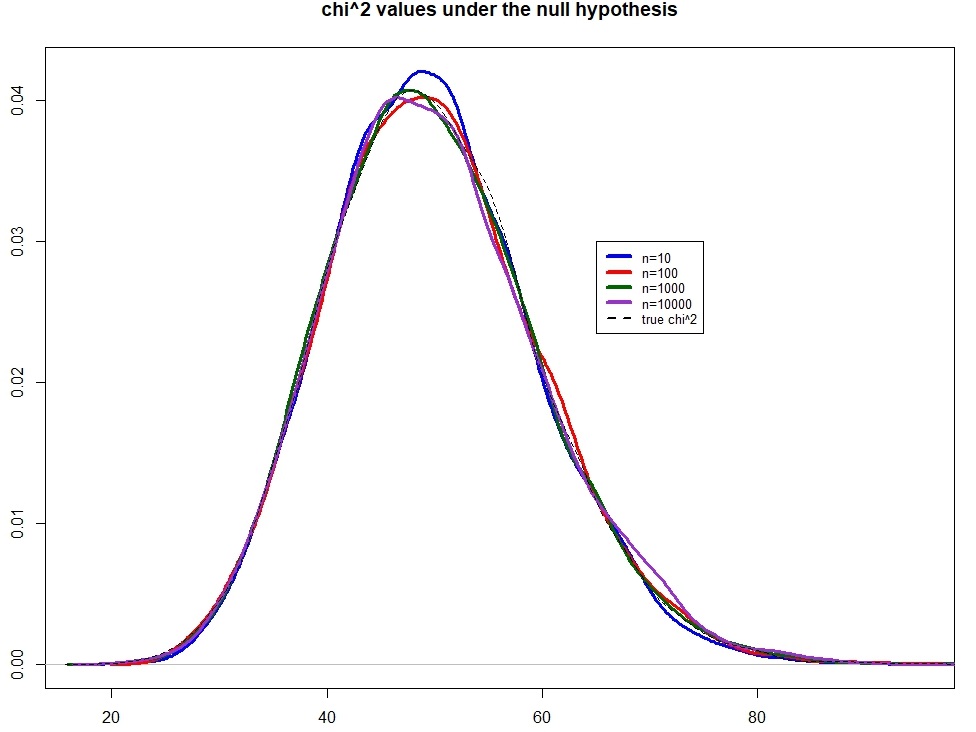}
			\caption{\#Types = 50}
		\end{subfigure}
		\begin{subfigure}[h]{0.48\textwidth}
			\includegraphics[scale=0.3]{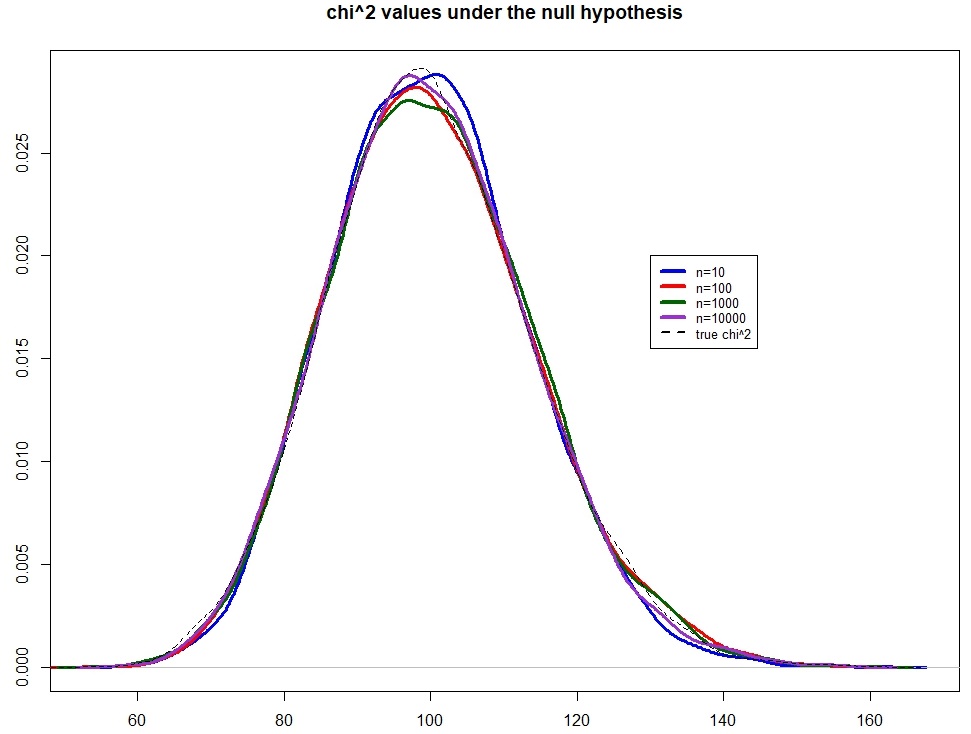}
			\caption{\#Types = 100}
		\end{subfigure}
		\caption{ \label{fig:exp1} The empirical distribution of our test quantity under the null-hypothesis.\\  {\small (Best seen in color) We ran our $\chi^2$-based test under the null-hypothesis. Not surprisingly, the results we get seem to be taken from a $\chi^2$-distribution (also plotted in a dotted black line). In all of the experiments we set $\epsilon=0.25$.} }
		\end{figure}}

	Figure~\ref{fig:exp2} details the empirical distribution of $P(\vec \theta)$ we get under the alternative-hypothesis, under different sample complexities ($n=\{2500,5000,7500,10000,200000\}$) for different TV-distances from the null-hypothesis ($\alpha=\{0.25,0.2,0.15,0.1\}$). The results show the same pattern, as $n$ increases, the distribution of $P(\vec \theta)$ shifts away from the $\chi^2_T$-distribution. This is clearly visible in the case where the total-variation distance is $0.25$, and becomes less apparent as we move closer to the null-hypothesis.
	\toggle{\begin{figure*}[h]
			\centering
			\begin{subfigure}[h]{0.45\textwidth}
				\includegraphics[scale=0.3]{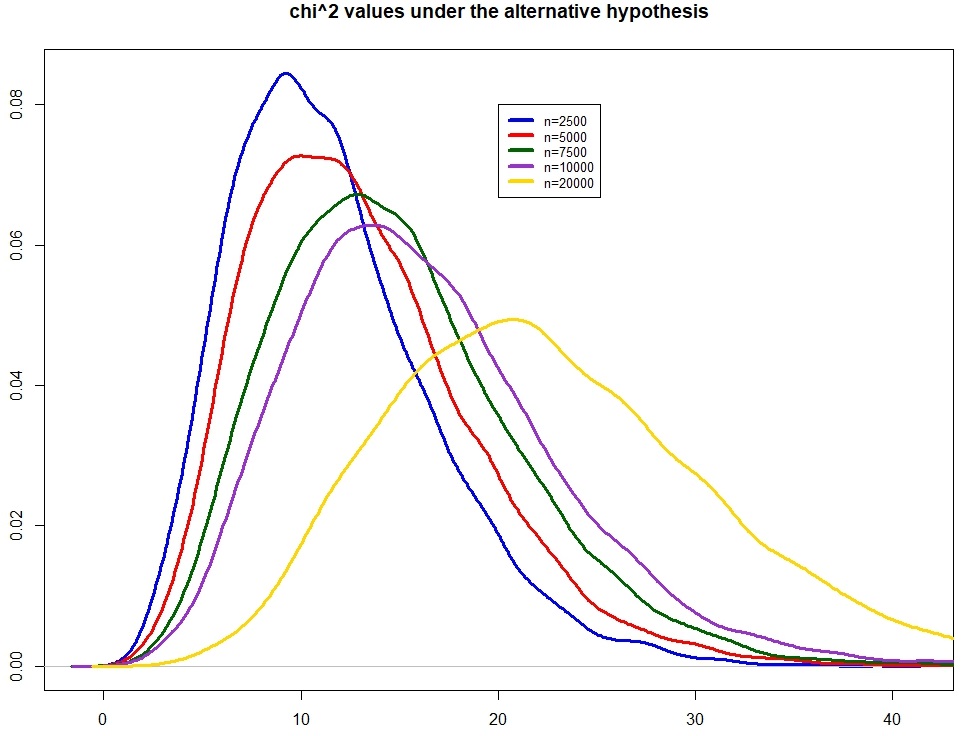}
				\caption{TV-dist = 0.25}
			\end{subfigure}
			\begin{subfigure}[h]{0.45\textwidth}
				\includegraphics[scale=0.3]{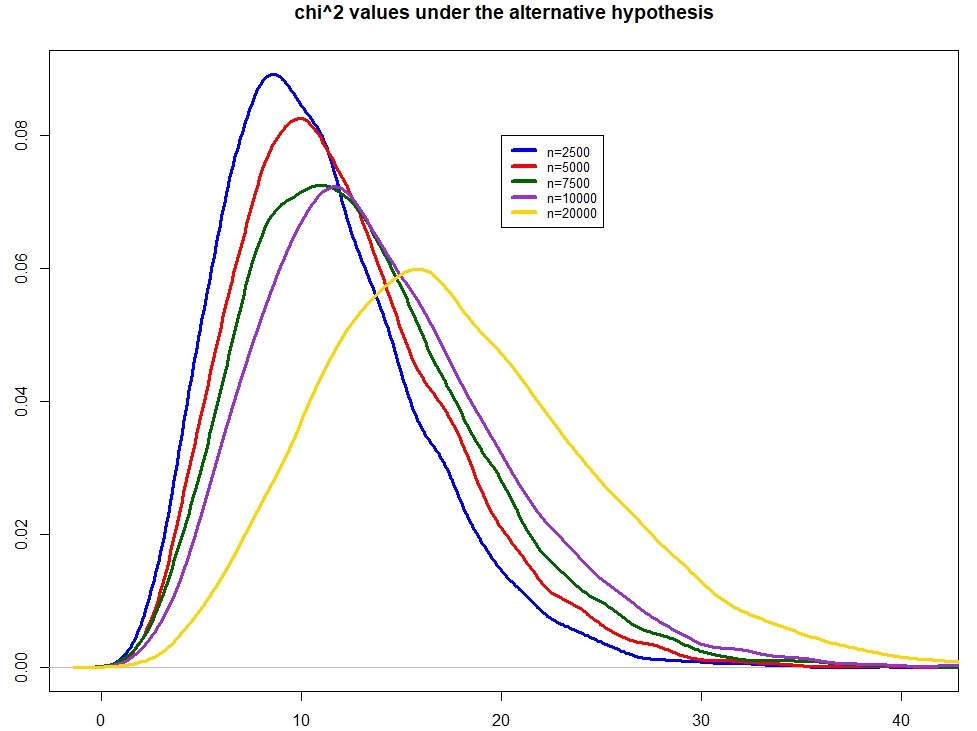}
				\caption{TV-dist = 0.2}
			\end{subfigure}
			\begin{subfigure}[h]{0.45\textwidth}
				\includegraphics[scale=0.3]{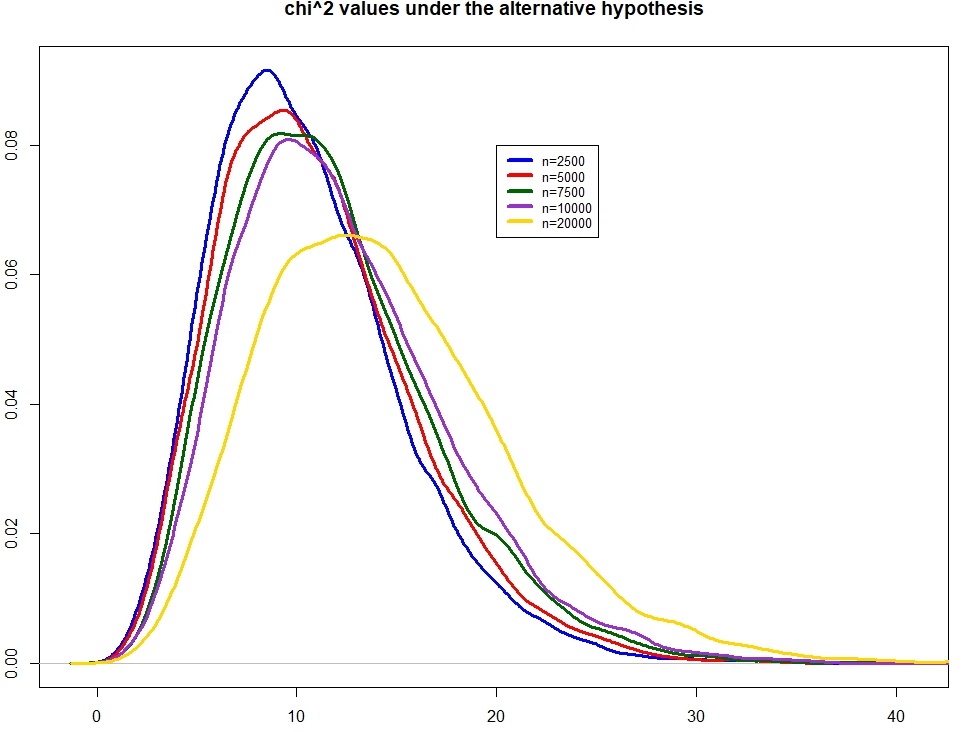}
				\caption{TV-dist = 0.15}
			\end{subfigure}
			\begin{subfigure}[h]{0.45\textwidth}
				\includegraphics[scale=0.3]{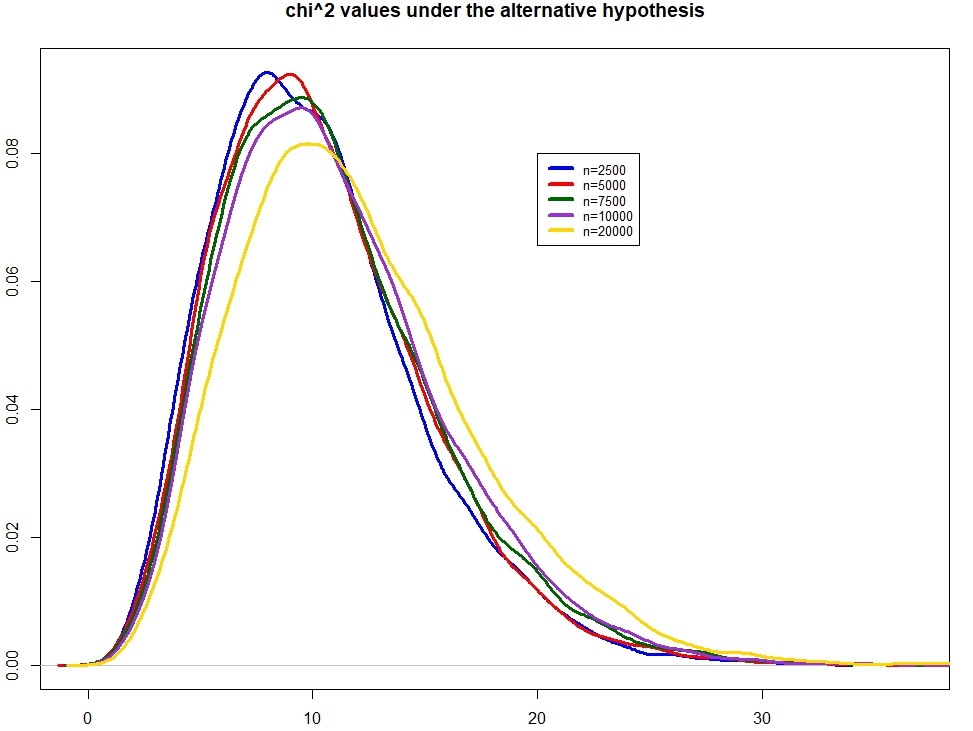}
				\caption{TV-dist = 0.1}
			\end{subfigure}
			\captionsetup[subfigure]{justification=centering}
			\captionsetup{width=.825\linewidth}
			\caption{ \label{fig:exp2} The empirical distribution of our test quantity under the alternative-hypothesis.\\  {\small (Best seen in color) }  We ran our $\chi^2$-based test under the alternative-hypothesis with various choices of TV-distance. As the number of samples increases, the empirical distribution of the test-quantity becomes further away from the $\chi^2$-distribution. In all of the experiments, the number of types is $10$ and $\epsilon=0.25$.}
		\end{figure*}}{\begin{figure}[h]
		\centering
		\begin{subfigure}[h]{0.48\textwidth}
			\includegraphics[scale=0.3]{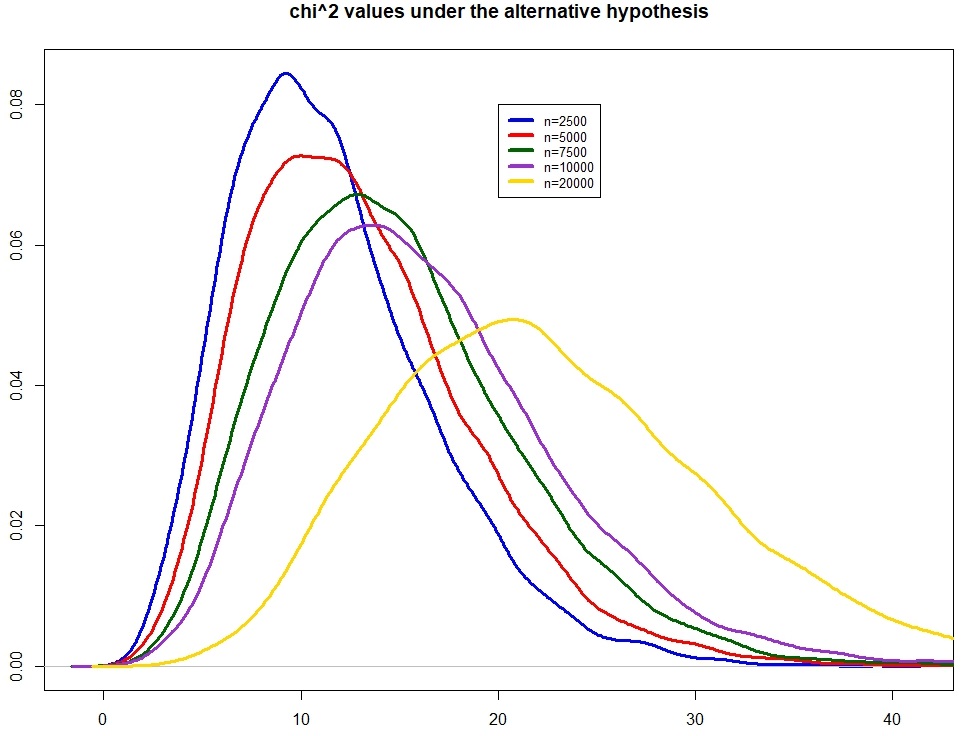}
			\caption{TV-dist = 0.25}
		\end{subfigure}
		\begin{subfigure}[h]{0.48\textwidth}
			\includegraphics[scale=0.3]{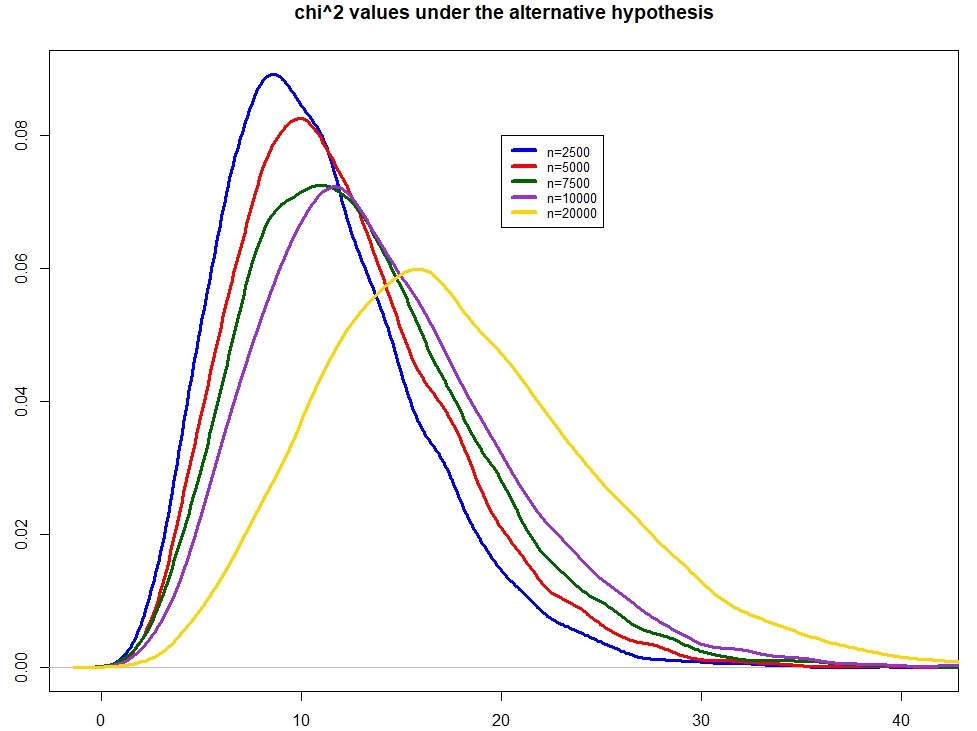}
			\caption{TV-dist = 0.2}
		\end{subfigure}
		\begin{subfigure}[h]{0.48\textwidth}
			\includegraphics[scale=0.3]{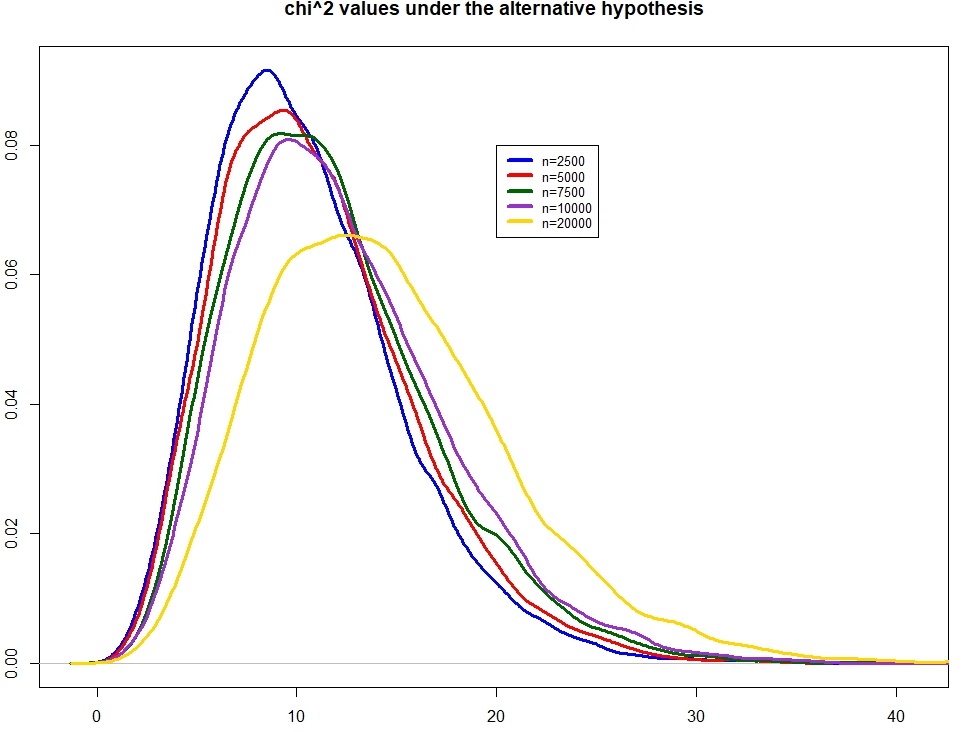}
			\caption{TV-dist = 0.15}
		\end{subfigure}
		\begin{subfigure}[h]{0.48\textwidth}
			\includegraphics[scale=0.3]{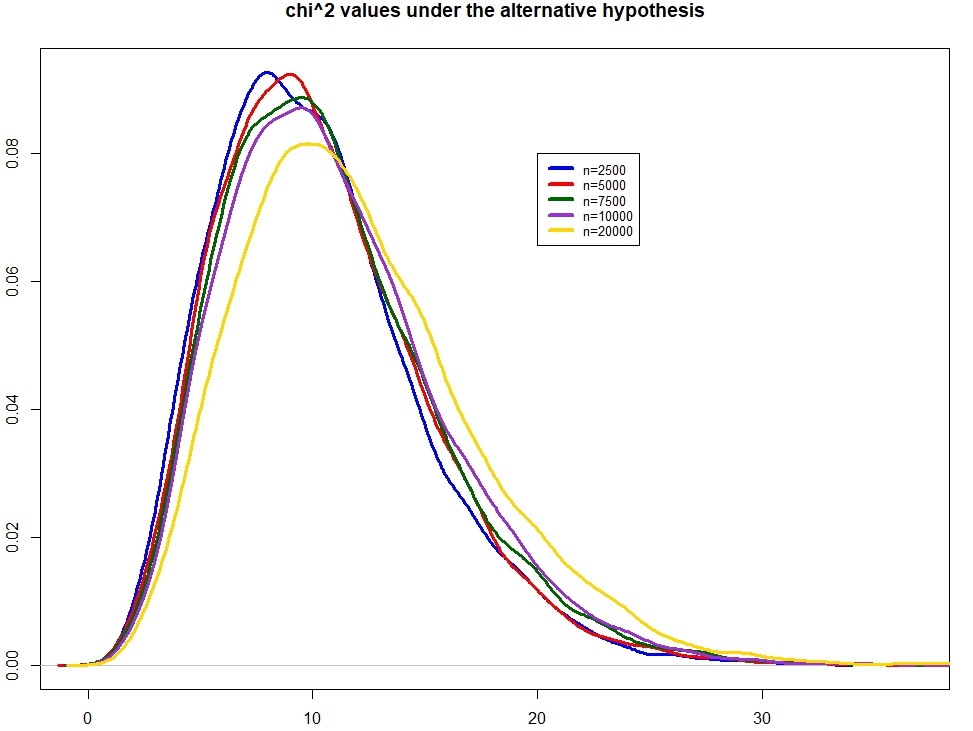}
			\caption{TV-dist = 0.1}
		\end{subfigure}
		\caption{ \label{fig:exp2} The empirical distribution of our test quantity under the alternative-hypothesis.\\  {\small (Best seen in color) }  We ran our $\chi^2$-based test under the alternative-hypothesis with various choices of TV-distance. As the number of samples increases, the empirical distribution of the test-quantity becomes further away from the $\chi^2$-distribution. In all of the experiments, the number of types is $10$ and $\epsilon=0.25$.}
	\end{figure}}

	{\paragraph{Open Problems.} The results of our experiment, together with the empirical results of the 3rd experiment (shown in Figure~\ref{fig:exp3}) give rise to the conjecture that the testers in Section~\ref{subset:non-symmetric-RR} are not optimal. In particular, we suspect that the $\chi^2$-based test we experiment with is indeed a valid tester of sample complexity $T^{1.5}/(\eta\alpha)^2$. Furthermore, there could be other testers of even better sample complexity. Both the improved upper-bound and finding a lower-bound are two important open problem for this setting. We suspect that the way to tackle this problem is similar to the approach of Acharya et al~\mycite{AcharyaDK15}; however following their approach is difficult for two reasons. First, one would technically need to give a bound on the $\chi^2$-divergence between $\tfrac 1 {2\eta}\vec \theta$ and $\vec q$ (or $\vec f$). Secondly, and even more challenging, one would need to design a tester to determine whether the observed collection of random vectors in $\{1,-1\}^{T}$ is likely to come from the mechanism operating on a distribution close to $\tfrac 1 {2\eta}\vec \theta$. This distribution over vectors is a \emph{mixture model} of product-distributions (but not a product distribution by itself); and while each product-distribution is known (essentially each of the $T$ product distributions is a product of random $\{1,-1\}$ bits except for the $x$-coordinate which equals $1$ w.p. $\tfrac 1 2 +\eta$) it is the weights of the distributions that are either $\vec p$ or $\alpha$-far from  $\vec p$. Thus one route to derive an efficient tester can go through learning mixture models~--- and we suspect that is also a route for deriving lower bounds on the tester. A different route could be to follow the maximum-likelihood (or the loss-function $f$ from Equation~\eqref{eq:loss_func_nonsymmetricRR}), with improved convexity bounds proven directly on the $L_1/L_\infty$-norms.

	As explained in Section~\ref{subsec:experiment}, we could not establish that 
	\[Q(\vec\theta) \stackrel{\rm def}=n \sum_x \frac{  (\tfrac 1 {2\eta} \theta(x) - \bar{\theta}(x))^2} {\bar \theta(x)}\]
	can serve as a test quantity, since we could not assess its asymptotic distribution. Nonetheless, we do believe it be a test quantity, as the following empirical results.
	We empirically measure the quantity $Q(\vec\theta) \stackrel{\rm def}=n \sum_x \frac{  (\tfrac 1 {2\eta} \theta(x) - \bar{\theta}(x))^2} {\bar \theta(x)}$ under the null ($\alpha=0$) and the alternative ($\alpha=0.25$) hypothesis with $n = 25,000$ samples in each experiment. The results under a variety of bin sizes are given in Figure~\ref{fig:exp4}. The results point to three facts: (1) the empirical distribution of $Q$ under the null hypothesis is \emph{not} a $\chi^2$-distribution (it is not as centered around the mean and the tail is longer). (2) there is a noticeable gap between the distribution of $Q(\vec{\theta})$ under the null-hypothesis and under the alternative-hypothesis. Indeed, the gap becomes less and less clear under $25,000$ samples as the size of the domain increases, but it is present. (3) The empirical sample complexity required to differentiate between the null- and the alternative-hypothesis is quite large. Even for modest-size domains, $25,000$ samples weren't enough to create a substantial differentiation between the two scenarios.
	Designing a tester based on the quantity $Q(\vec \theta)$ is thus left as an open problem.
	
	\toggle{\begin{figure*}[h]
			\centering
			\begin{tabular}{m{0.45\hsize}m{0.45\hsize}}
			\begin{subfigure}[h]{0.45\textwidth}
				\includegraphics[scale=0.25]{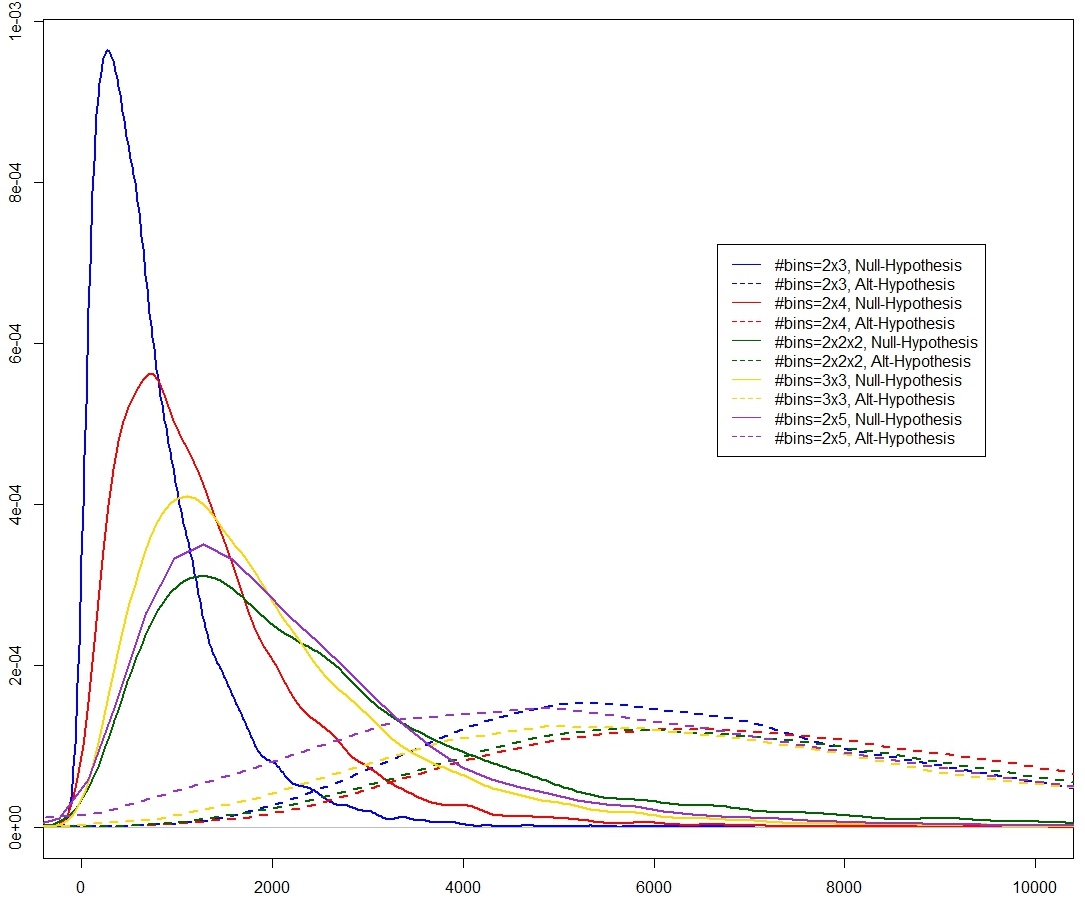}
				\caption{Small Domain (between 6-10 types)}
			\end{subfigure}
			\begin{subfigure}[h]{0.45\textwidth}
				\includegraphics[scale=0.25]{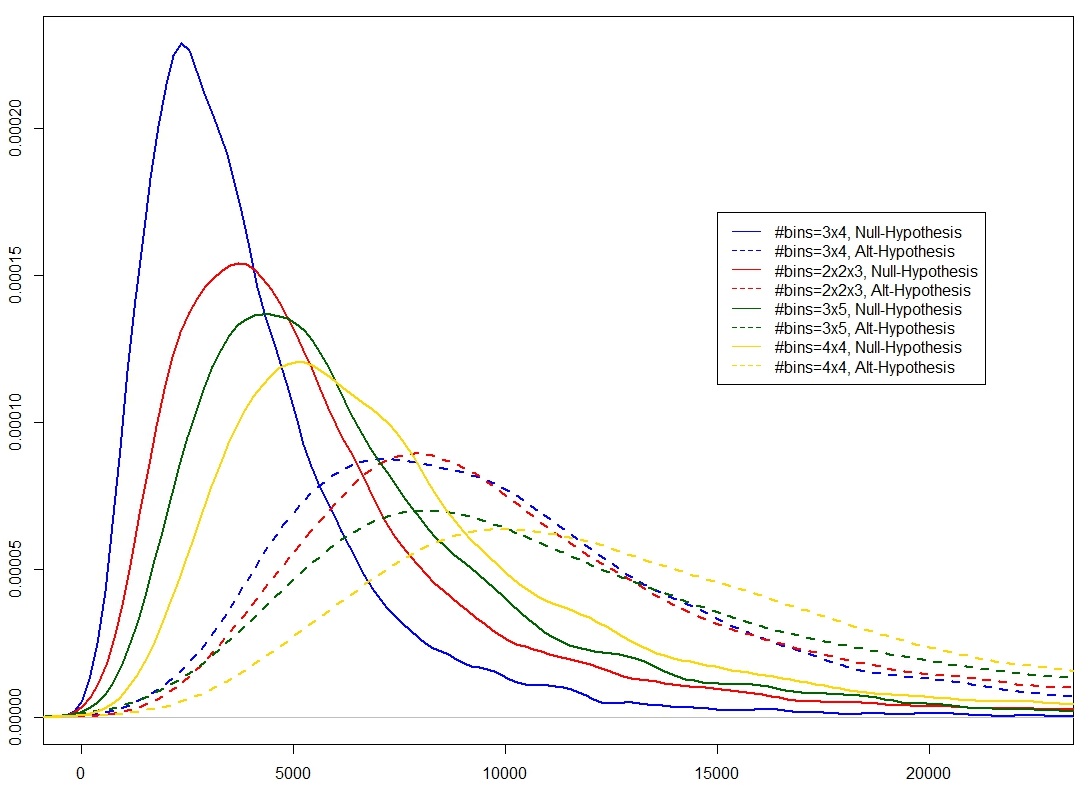}
				\caption{Mid-Size Domain (12-16 types)}
			\end{subfigure}
			\begin{subfigure}[h]{0.45\textwidth}
			\includegraphics[scale=0.25]{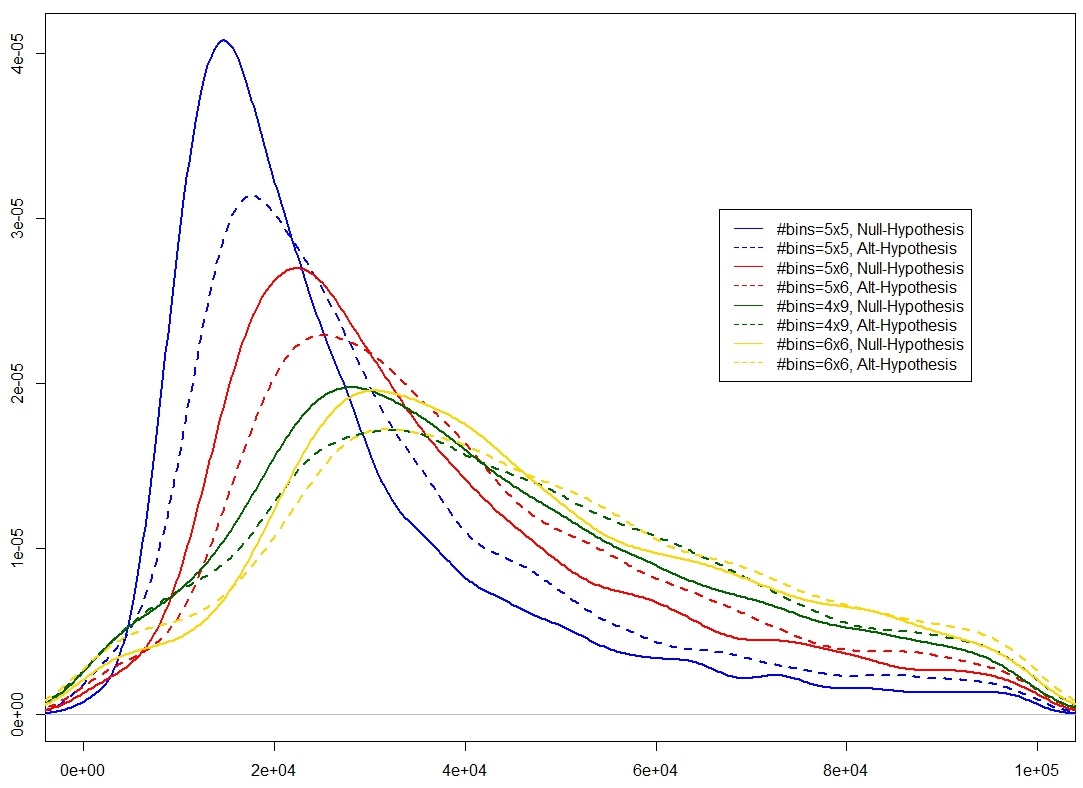}
			\caption{Large Domain (25-36 types)}
			\end{subfigure}
			&
			\begin{subfigure}[h]{0.45\textwidth}
				\includegraphics[scale=0.25]{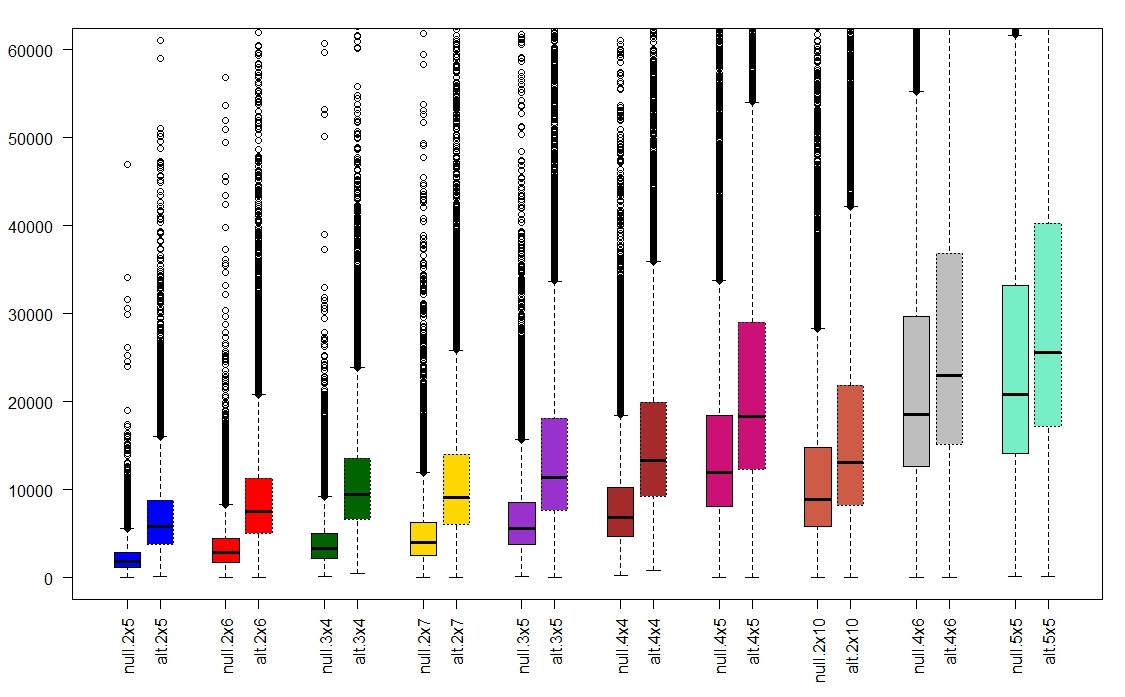}
				\caption{2-Dimensional Domains (12-30 types)}
			\end{subfigure}
			\begin{subfigure}[h]{0.45\textwidth}
				\includegraphics[scale=0.25]{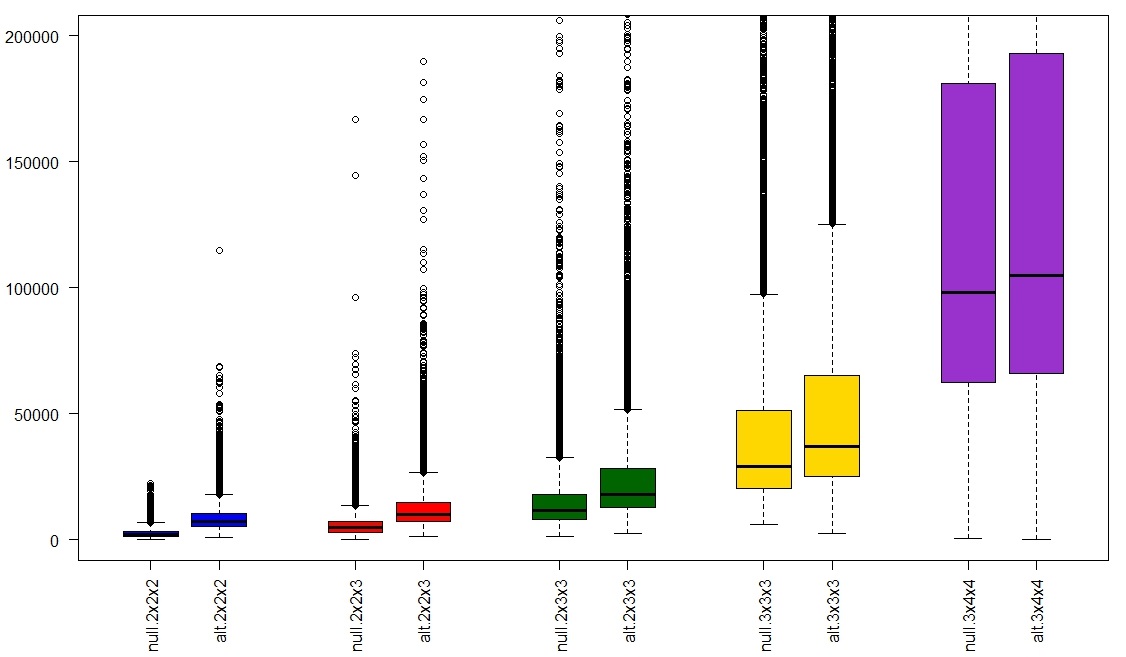}
				\caption{3-Dimensional Domains (8-48 types)}
			\end{subfigure}
			\begin{subfigure}[h]{0.45\textwidth}
				\includegraphics[scale=0.25]{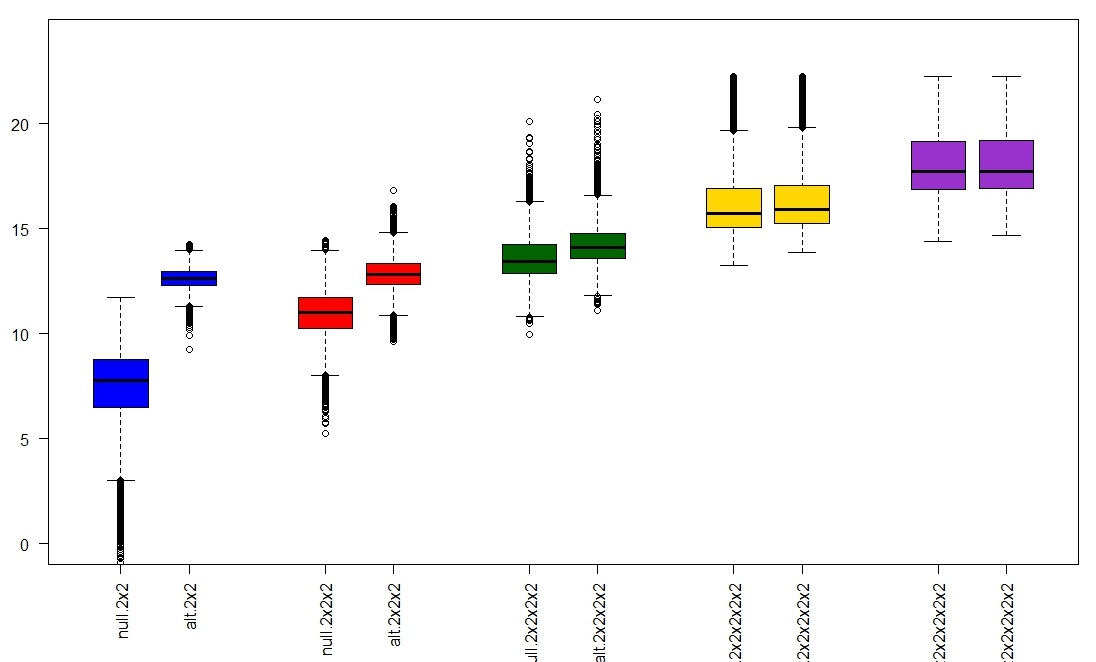}
				\caption{Powers of 2 Size-Domains \underline{shown in log-scale}  (4-64 types) }
			\end{subfigure}
			\end{tabular}
			\captionsetup[subfigure]{justification=centering}
			\captionsetup{width=.825\linewidth}
			\caption{ \label{fig:exp4} The empirical distribution of $Q(\vec \theta)$ for domains of multiple size, under the null- and the alternative-hypothesis.\\  {\small (Best seen in color) }  We measured $Q(\vec \theta)$ under both the null-hypothesis (solid line) and the alternative-hypothesis (dotted line) with various choices of domain sizes. As the size of the domain increases, it is evident the $25,000$ samples aren't enough to differentiate between the null and the alternative. In all of the experiments $\epsilon=0.25$ and the alternative hypothesis is $0.25$-far in TV-dist from a product distribution.}
		\end{figure*}}{\begin{figure}[h]
		\centering
		\begin{tabular}{m{0.48\hsize}m{0.48\hsize}}
			\begin{subfigure}[h]{0.48\textwidth}
				\includegraphics[scale=0.25]{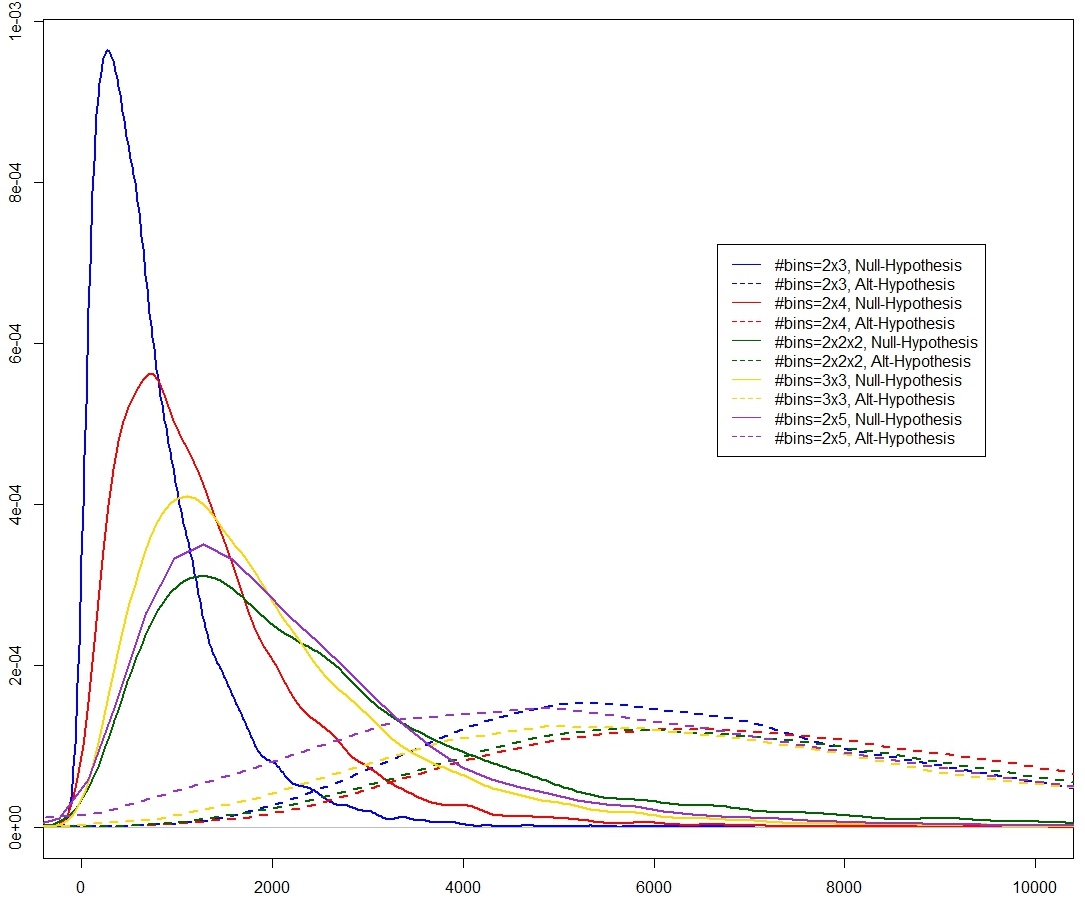}
				\caption{Small Domain (between 6-10 types)}
			\end{subfigure}
			\begin{subfigure}[h]{0.48\textwidth}
				\includegraphics[scale=0.25]{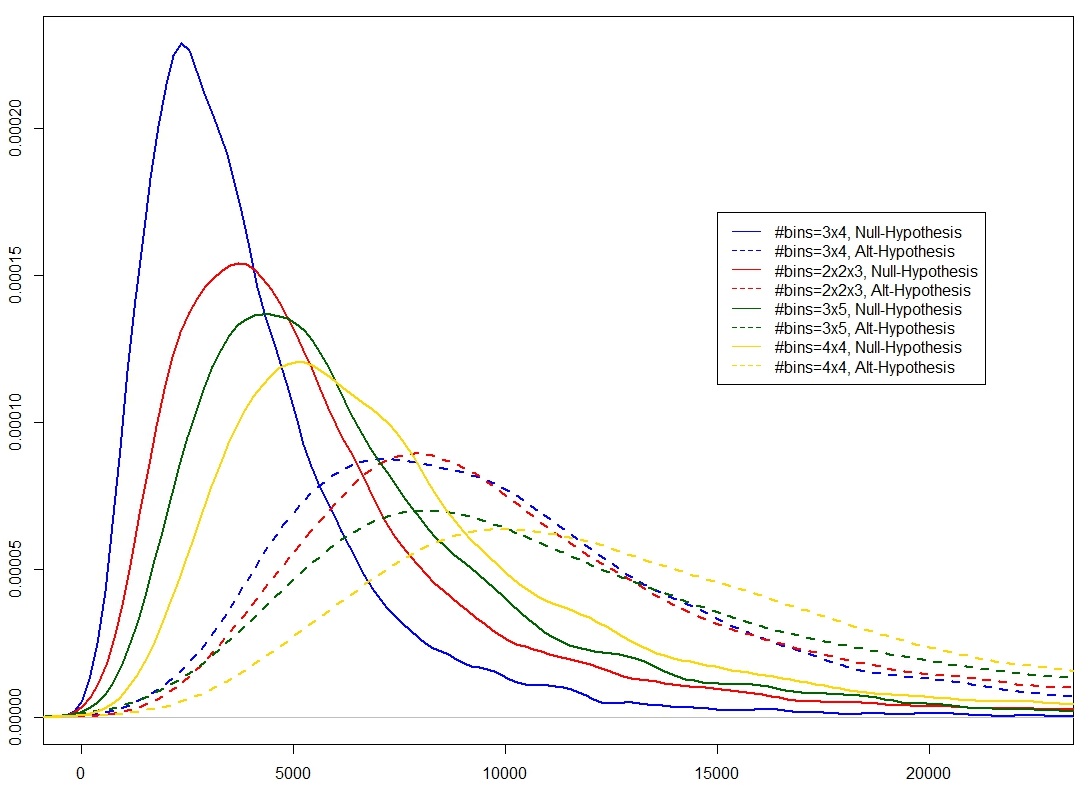}
				\caption{Mid-Size Domain (12-16 types)}
			\end{subfigure}
			\begin{subfigure}[h]{0.48\textwidth}
				\includegraphics[scale=0.25]{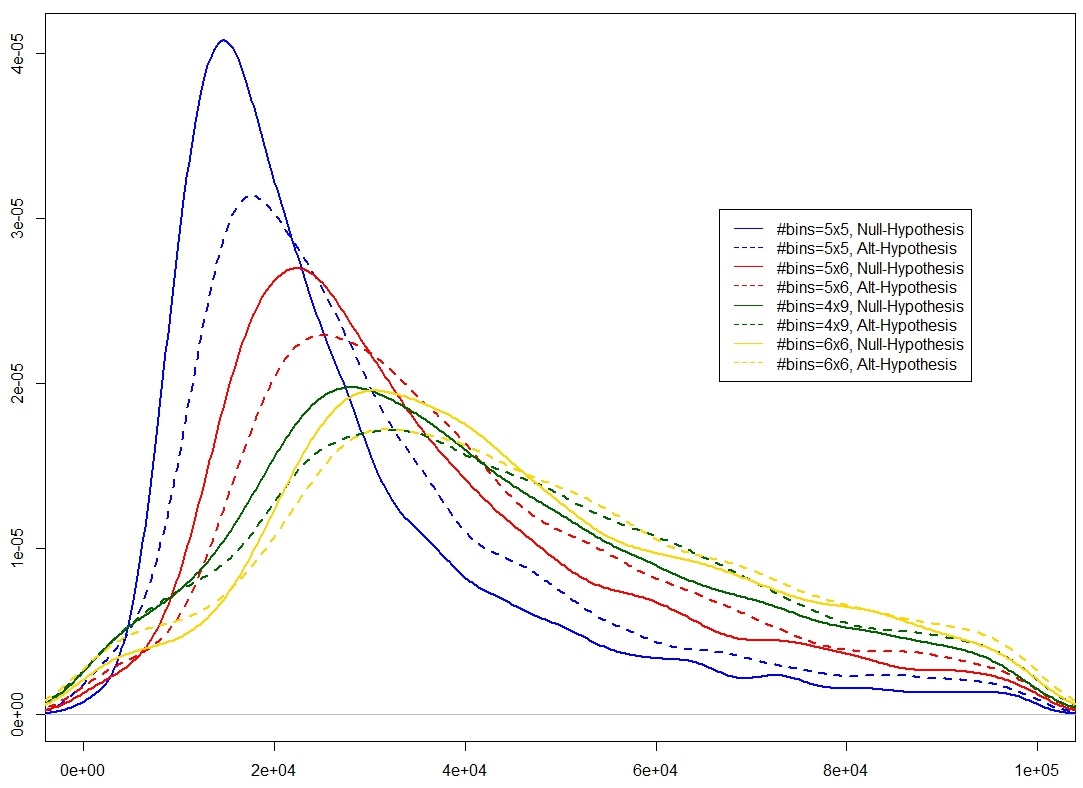}
				\caption{Large Domain (25-36 types)}
			\end{subfigure}
			&
			\begin{subfigure}[h]{0.48\textwidth}
				\includegraphics[scale=0.25]{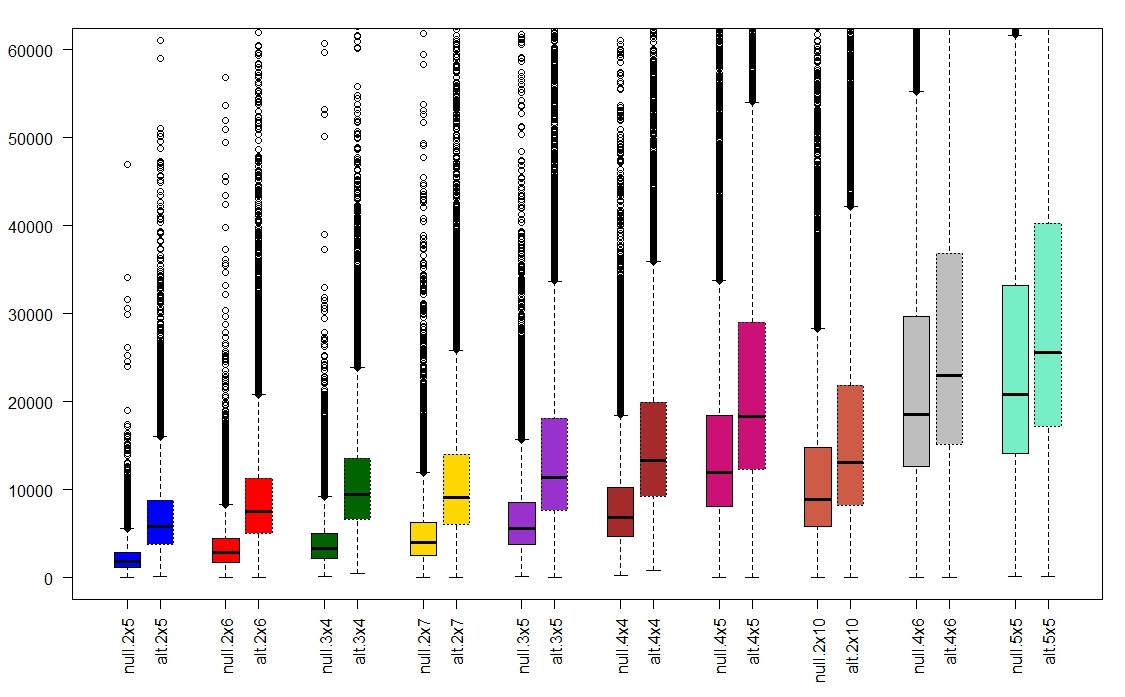}
				\caption{2-Dimensional Domains (12-30 types)}
			\end{subfigure}
			\begin{subfigure}[h]{0.48\textwidth}
				\includegraphics[scale=0.25]{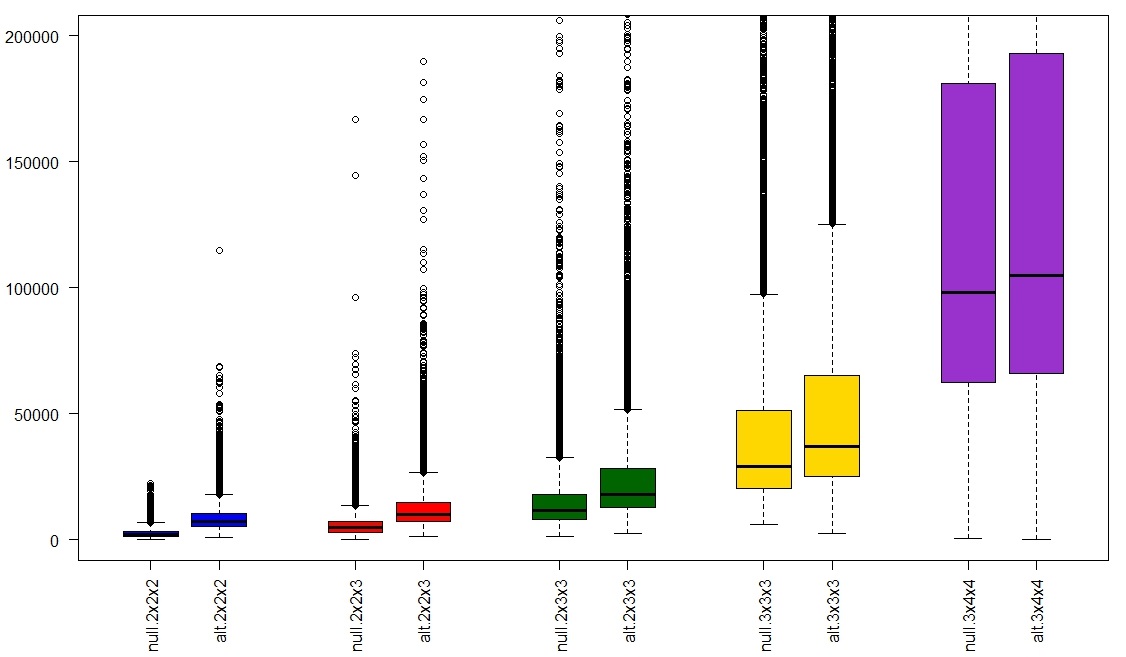}
				\caption{3-Dimensional Domains (8-48 types)}
			\end{subfigure}
			\begin{subfigure}[h]{0.48\textwidth}
				\includegraphics[scale=0.25]{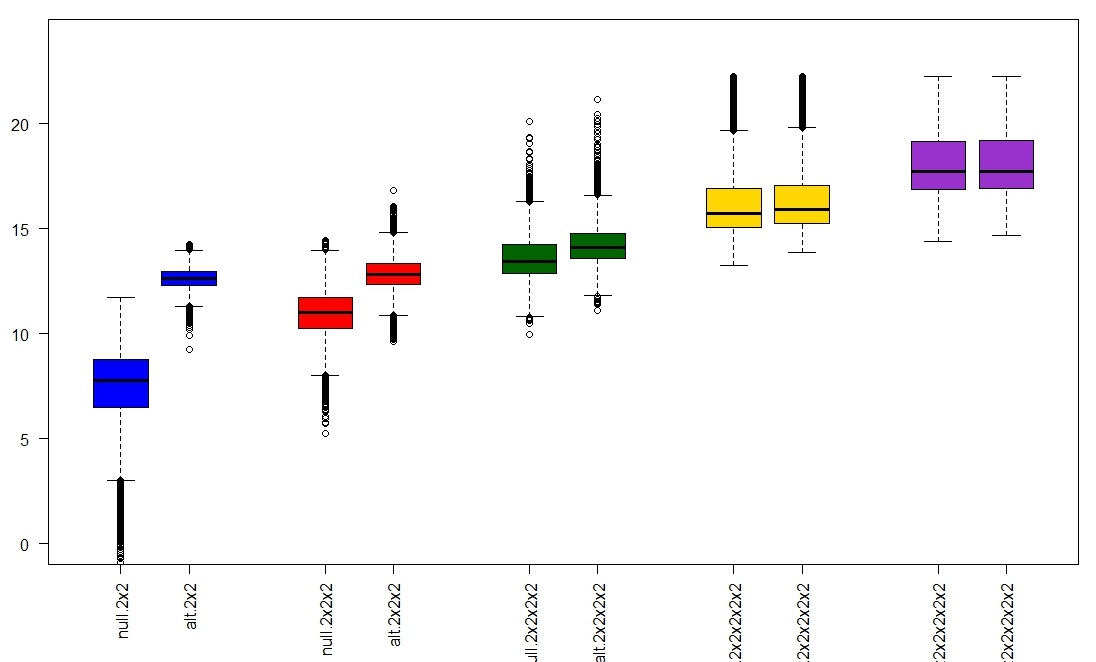}
				\caption{Powers of 2 Size-Domains \emph{shown in log-scale} (4-64 types) }
			\end{subfigure}
		\end{tabular}
		\captionsetup[subfigure]{justification=centering}
		\captionsetup{width=.825\linewidth}
		\caption{ \label{fig:exp4} The empirical distribution of $Q(\vec \theta)$ for domains of multiple size, under the null- and the alternative-hypothesis.\\  {\small (Best seen in color) }  We measured $Q(\vec \theta)$ under both the null-hypothesis (solid line) and the alternative-hypothesis (dotted line) with various choices of domain sizes. As the size of the domain increases, it is evident the $25,000$ samples aren't enough to differentiate between the null and the alternative. In all of the experiments $\epsilon=0.25$ and the alternative hypothesis is $0.25$-far in TV-dist from a product distribution.}
	\end{figure}}
	}

\end{document}